\documentclass[11pt, leqno]{article}
\usepackage{amsmath,setspace,multirow,lineno}
\usepackage[font=small,labelfont=bf,singlelinecheck=off]{caption}
\usepackage[top=.5in, bottom=.5in, left=.75in, right=.75in]{geometry}

\usepackage[round]{natbib}
\usepackage{color,soul}

\DeclareCaptionStyle{italic}[justification=centering]{labelfont={bf},textfont={it},labelsep=colon}

\usepackage{graphicx,psfrag,epsf,textcomp,epstopdf, amsthm, paralist, amssymb}
\usepackage{enumerate, dsfont, alltt, verbatim}
\usepackage{natbib}
\usepackage{url,xcolor}
\usepackage{booktabs, subfig, bm, paralist,mathpazo,tikz,longtable,microtype, authblk}
\usepackage[linesnumbered,ruled,vlined]{algorithm2e}

\usepackage[pdftex,colorlinks=true]{hyperref}
\definecolor{darkblue}{rgb}{0,0,.6}
\hypersetup{citecolor=darkblue,linkcolor=darkblue,urlcolor=darkblue}
\definecolor{DarkRed}{rgb}{.7,0,.4}

\usepackage{comment,orcidlink}

\newcommand{\blind}{0}

\newcommand{\E}{\text{E}}
\newcommand{\X}{\mathcal{X}}
\newcommand{\Y}{\mathcal{Y}}

\newcommand{\Rlogo}{\protect\includegraphics[height=1.8ex,keepaspectratio]{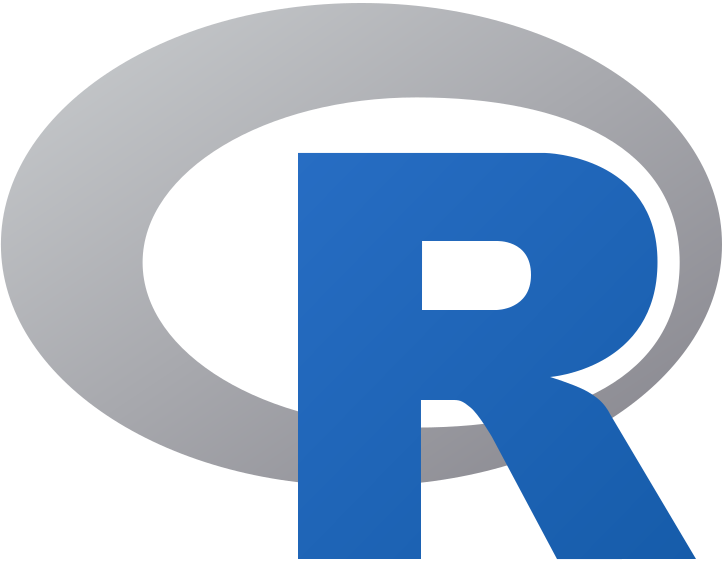}}

\graphicspath{{plots/}}
\DeclareMathOperator*{\argmin}{\arg\!\min}

\newsavebox\CBox

\newtheorem{@definition}{\sc Definition}[section]

\newtheorem{corollary}{\sc Corollary}[section]
\newtheorem{theorem}{\sc Theorem}[section]
\newtheorem{remark}{\sc Remark}[section]
\newtheorem{lemma}{Lemma}[section]
\newtheorem{assumption}{Assumption}

\renewcommand\X{\mathcal{X}}

\date{}

\begin{document}

\def\spacingset#1{\renewcommand{\baselinestretch}{#1}\small\normalsize} \spacingset{1}

\if0\blind
{
\title{\bf Spatial function-on-function quantile regression}}
\author[1]{\normalsize Eylul Fidan}
\author[1]{\normalsize Ufuk Beyaztas\orcidlink{0000-0002-5208-4950}}
\author[2]{\normalsize Soutir Bandyopadhyay\thanks{Corresponding address: Department of Applied Mathematics and Statistics, Colorado School of Mines, Golden, CO; Email: sbandyopadhyay@mines.edu}\orcidlink{0000-0003-2213-3333}}

\affil[1]{\normalsize Department of Statistics, Marmara University, Turkey}
\affil[2]{\normalsize Department of Applied Mathematics and Statistics, Colorado School of Mines, USA}
\maketitle
\fi

\if1\blind
{
\title{\bf Spatial function-on-function quantile regression}
\author{}
} \fi

\maketitle

\begin{abstract}
This paper introduces a novel penalized spatial function-on-function quantile regression framework for analyzing spatially indexed functional data, bridging a critical gap between spatial functional models and quantile regression. Our work makes three key contributions. First, we propose the first spatial function-on-function quantile regression model that jointly accounts for spatial correlation across curves through a functional spatial autoregressive structure while allowing inference on arbitrary conditional quantiles of the functional response. Unlike traditional mean-based alternatives, this approach successfully captures state-dependent volatility and distributional dynamics beyond the conditional mean. Second, we develop a two-stage instrumental-variable estimation strategy to address endogeneity induced by the functional spatial lag. By utilizing tensor-product B-spline expansions with tensor-product roughness penalties, our method ensures optimal smoothness without the destructive information loss inherent in principal component truncation. Third, for fixed spline dimensions, we establish $\sqrt n$-asymptotic normality of the spline coefficient estimators and the induced finite-rank Gaussian-process limits for the reconstructed coefficient surfaces. Extensive Monte Carlo experiments and a high-resolution analysis of Italian PM$_{2.5}$ air quality data demonstrate that spatial function-on-function quantile regression significantly outperforms non-spatial and mean-based competitors, providing a robust and informative tool for environmental risk management and complex functional data analysis. Our method has been implemented in the \texttt{SpatialFoFReg} \Rlogo \ package. 
\end{abstract}

\newpage
\spacingset{1.65} 

\section{Introduction} \label{sec:1}

Air-quality monitoring networks now routinely generate spatially indexed functional observations. Our motivating application is the analysis of PM$_{2.5}$ air pollution using the Italian monitoring network, with particular emphasis on the Lombardy region of Northern Italy, where station coverage is dense and pollution episodes are environmentally important. The data, obtained from the \texttt{ARPALData} \Rlogo\ package \citep{arpaldata}, consist of temporally evolving pollution curves observed at 1481 monitoring stations across Italy, with PM$_{2.5}$ concentration treated as the functional response and PM$_{10}$ concentration as the functional predictor. This application naturally requires a function-on-function regression (FoFR) framework, where both the response and predictor are curves \citep{ramsay1991, RamsaySilverman2006}; it also calls for spatial modeling because nearby stations are connected through common emission sources, atmospheric transport, and regional meteorological conditions \citep{Baccini2011, Maranzano2022, Robotto2021}. Moreover, environmental risk assessment is not only concerned with average PM$_{2.5}$ levels, but also with upper-tail pollution episodes, where health impacts are most severe and where the predictor-response relationship may differ from that at the center of the distribution.

Figure~\ref{fig:motivation} summarizes these motivating features. The spatial distribution of the monitoring sites is irregular, the PM$_{2.5}$ and PM$_{10}$ curves display clear temporal variation, and the functional Moran's $I$ statistic indicates strong positive spatial dependence in PM$_{2.5}$ over the study period. These features create both statistical and computational challenges. One must simultaneously model a bivariate functional regression surface $\beta_\tau(t,s)$ and a bivariate spatial spillover surface $\rho_\tau(t,u)$, while the spatially lagged response is endogenous and therefore requires an instrumental-variable strategy. Computationally, the tensor-product spline representation leads to a high-dimensional nonsmooth quantile optimization problem, repeated over multiple quantile levels and smoothing-parameter choices. These considerations motivate the penalized spatial function-on-function quantile regression (SFoF-QR) framework developed in this paper.

\begin{figure}[!t]
\centering
\includegraphics[width=.42\textwidth]{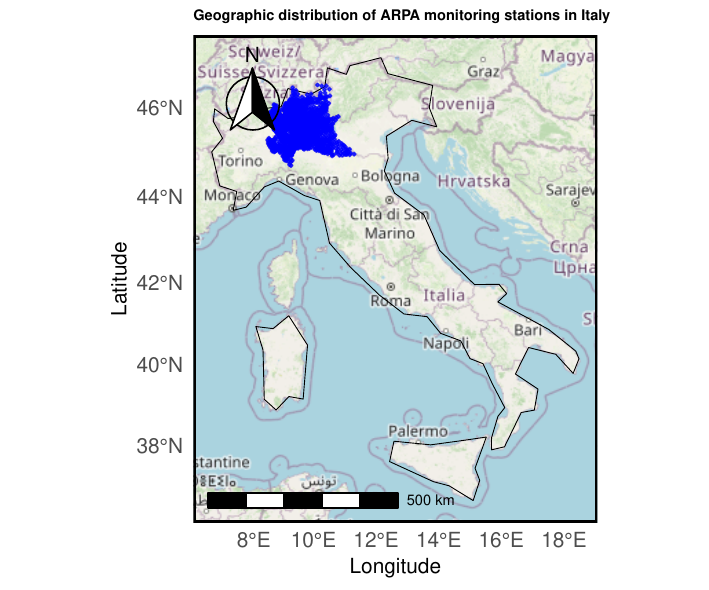}
\quad
\includegraphics[width=.42\textwidth]{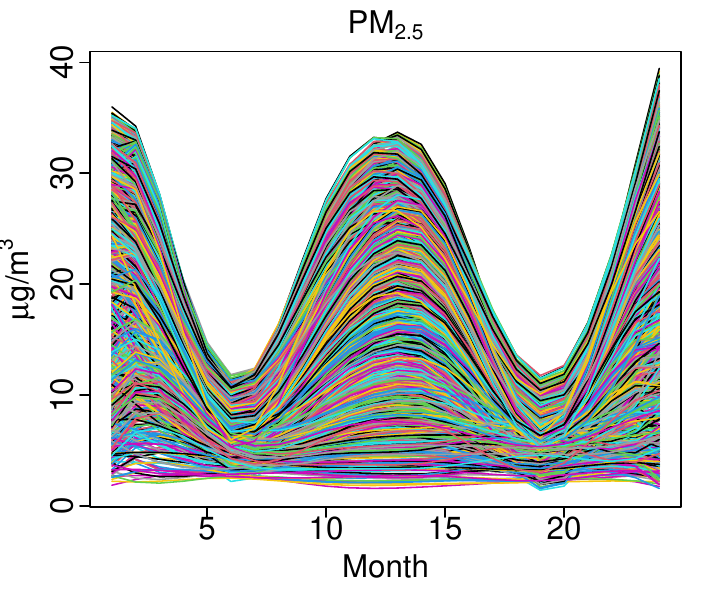}
\\[-0.25em]
\includegraphics[width=.42\textwidth]{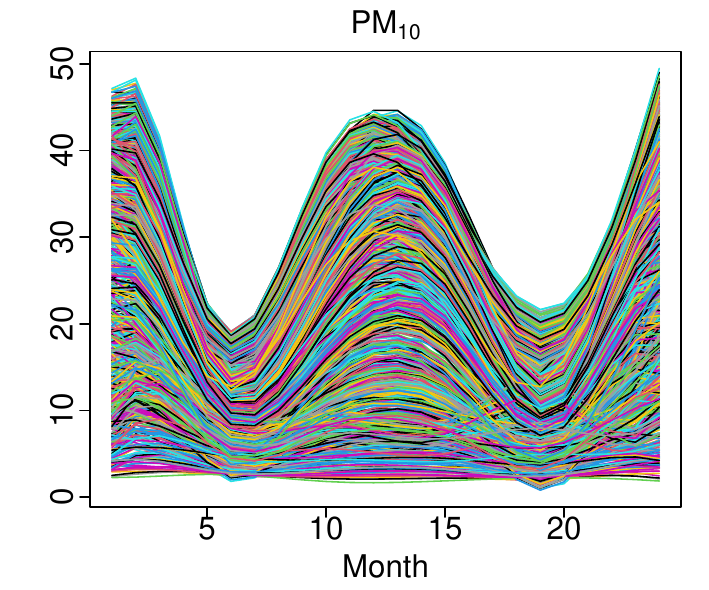}
\quad
\includegraphics[width=.42\textwidth]{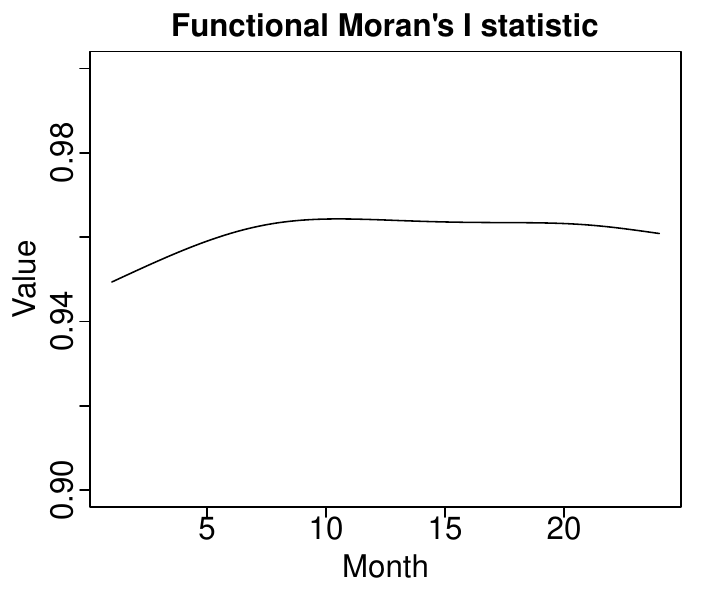}
\caption{Motivating features of the PM$_{2.5}$ air-quality data. The top left panel shows the spatial distribution of the 1481 monitoring stations in Italy. The top right panel shows the smoothed PM$_{2.5}$ response curves. The bottom left panel shows the smoothed PM$_{10}$ predictor curves. The bottom right panel displays the functional Moran's $I$ statistic for PM$_{2.5}$. These plots illustrate irregular spatial sampling, temporally varying functional observations, and strong positive spatial dependence, motivating a spatial function-on-function quantile regression analysis.}
\label{fig:motivation}
\end{figure}

Methodologically, this work builds on the functional data analysis literature, including functional principal component analysis (FPCA)-based and sparse functional regression approaches \citep{MullerYao2008}, penalized/basis-expansion FoFR methods \citep{ivanescu2015}, and general developments summarized in \citet{Ferraty2006}, \citet{Horvath2012}, and \citet{Kokoszka2018}. It is also connected to spatial functional data analysis, where kriging-inspired and spatially regularized methods have been developed for point-referenced functional observations \citep{Aguilera2017, Giraldo2018}. However, these approaches do not provide a spatial autoregressive function-on-function quantile regression model for areal or network-indexed functional responses.

Within spatial econometrics, most specifications fall into three archetypes: the spatial autoregressive (SAR) specification, spatial error, and spatial Durbin models \citep{Lesage2009}. We focus on the SAR paradigm because it injects neighborhood effects directly through a spatially lagged response, scales well computationally, and naturally supports instrumental-variable/two-stage estimation for the endogeneity generated by the spatially lagged functional response \citep[see, e.g.,][]{Kelejian1998, Huang2021}. In functional SAR settings, existing contributions remain somewhat fragmented, including approaches that effectively model each time point separately \citep[e.g.,][]{Zhu2022} or that allow functional responses with only scalar regressors \citep[e.g.,][]{Hoshino2024}. More recently, \citet{Camille2025} proposed a signature-based functional SAR approach that provides a basis-free alternative. Closest in spirit to our setting are spatial function-on-function regression (SFoFR) models, which augment FoFR with a functional spatial lag. In particular, \citet{BSGARC2025} introduced an SFoFR model whose estimation relies on spatial FPCA to reduce the problem to a multivariate SAR model on scores. While effective, FPCA-based strategies require selecting truncation levels, and this choice can materially affect both the recovered shapes and smoothness of the regression coefficient functions. Motivated by these issues, \citet{BSS2025} proposed a penalized SFoFR that represents coefficient surfaces using tensor-product B-splines and estimates them via a functional extension of generalized spatial 2SLS \citep{Kelejian1998}, yielding a penalized spatial two-stage least squares estimator.

A common limitation of the existing SFoFR literature is that it targets conditional means. Mean-based spatial functional models cannot describe distributional features that are central in many applications, such as tail behavior, spatially varying volatility, heteroskedasticity, or asymmetric contamination. Quantile regression \citep{Koenker1978} addresses precisely this limitation by modeling conditional quantiles rather than only conditional expectations. Quantile ideas have recently been extended to function-on-function settings \citep[see, e.g.,][]{Zhu2023, ZhouQ2023, Mutis2025}, providing robust and distributionally informative alternatives to least squares. However, these functional quantile FoFR methods assume independently and identically distributed (i.i.d.) curves and, to our knowledge, do not incorporate spatial dependence among functional observations.

In this paper, we bridge these two lines of research by introducing a penalized SFoF-QR framework for modeling the pointwise conditional quantile process of a spatially dependent functional response. Our contributions are threefold. (i) We propose, to our knowledge, the first FoFR-type model that jointly accounts for spatial dependence via a functional SAR mechanism and enables inference on arbitrary pointwise conditional quantiles of a functional response, thereby capturing distributional dynamics beyond the mean at each grid point of the response domain. (ii) To address the endogeneity induced by the functional spatial lag, we develop a two-stage instrumental-variable quantile estimation strategy. In contrast to FPCA truncation, our approach estimates the quantile coefficient surfaces directly using tensor-product B-splines together with tensor-product roughness penalties, allowing accurate recovery of potentially non-separable interaction surfaces while controlling smoothness. (iii) We establish an asymptotic theory for the proposed IV-based estimators, including fixed-spline $\sqrt n$-asymptotic normality and the induced finite-rank Gaussian-process limits for the reconstructed coefficient surfaces.

The remainder of the paper is organized as follows. Section~\ref{sec:2} introduces the SFoF-QR model and notation. Section~\ref{sec:3} presents the two-stage penalized IV quantile estimation procedure and selection of smoothing parameters. Section~\ref{sec:4} reports Monte Carlo evidence, and Section~\ref{sec:5} provides an empirical analysis of Italian air-quality data. Section~\ref{sec:6} concludes with directions for future research. Additional methodological and computational details, theoretical proofs, complete simulation results, and graphical analyses of the air-quality application are provided in the online supplementary material.

\section{Model and notation}\label{sec:2}

Let $\{\Y_v(t), \X_v(s)\}_{v \in \mathcal{D}}$ denote the functional response and predictor processes observed at $n$ spatial locations $v_1,\ldots,v_n$ in a domain $\mathcal{D}\subset\mathbb{R}^d$, $d\geq1$. For each $v_i$, $\Y_i(t)\in\mathcal{L}^p(\mathcal{I}_Y)$ with $t\in\mathcal{I}_Y\subset\mathbb{R}$ and $\X_i(s)\in\mathcal{L}^p(\mathcal{I}_X)$ with $s\in\mathcal{I}_X\subset\mathbb{R}$, where $2\leq p<\infty$. Without loss of generality, we take $\mathcal{I}_Y=\mathcal{I}_X=[0,1]$. The observations $\{\Y_i,\X_i\}_{i=1}^n$ are not independent; rather, they exhibit spatial dependence. We encode the spatial neighborhood structure through a known $n\times n$ weight matrix $\bm W=(w_{ij})$, where $w_{ij}\geq0$ reflects the proximity or influence of location $j$ on location $i$ and $w_{ii}=0$. The spatially lagged response curve at location $i$ is $\widetilde{\Y}_i(u)=\sum_{j=1}^n w_{ij}\Y_j(u)$, $u\in\mathcal{I}_Y$, which aggregates the neighboring response trajectories over the response domain.

For a scalar random variable $Z$ and a conditioning $\sigma$-field $\mathcal{G}$, we use $Q_\tau(Z\mid\mathcal{G}) = \inf\{q\in\mathbb{R}: \Pr(Z\leq q\mid\mathcal{G})\geq \tau\}$, $0<\tau<1$, to denote the population conditional $\tau$-quantile. Thus, in the present functional setting, the object modeled at each fixed $t\in\mathcal{I}_Y$ is the scalar conditional quantile of $\Y_i(t)$ given the full covariate curves, not the empirical quantile across the observed curves at time $t$. Let $\mathcal{F}_i=\sigma\{\widetilde{\Y}_i(u):u\in\mathcal{I}_Y; \X_i(s):s\in\mathcal{I}_X\}$ denote the information generated by the neighboring response curve and the local predictor curve. Throughout the manuscript, $\tau\in(0,1)$ denotes the quantile level, $t\in\mathcal I_Y$ denotes the response-domain argument, $u\in\mathcal I_Y$ denotes the integration argument of the spatially lagged response, and $s\in\mathcal I_X$ denotes the predictor-domain argument. For a given quantile level $\tau$, we postulate the following pointwise SFoF-QR model:
\begin{equation}\label{eq:quantile-model}
Q_{\tau}\{\Y_i(t)\mid \mathcal{F}_i\} = \int_{\mathcal{I}_Y} \widetilde{\Y}_i(u)\rho_{\tau}(t,u) du + \int_{\mathcal{I}_X} \X_i(s)\beta_{\tau}(t,s) ds, \qquad t\in\mathcal{I}_Y .
\end{equation}
Equivalently, $t\mapsto Q_{\tau}\{\Y_i(t)\mid\mathcal{F}_i\}$ is a conditional quantile function of the response trajectory. No stronger notion of a functional quantile in $\mathcal{L}^p(\mathcal I_Y)$ is invoked; the model consists of scalar conditional quantile restrictions indexed pointwise by $t\in\mathcal I_Y$. The variable $t$ is the target argument at which the quantile of the response is evaluated, whereas $u$ is an integration variable indexing the neighboring response curve. Hence, the kernel $\rho_\tau(t,u)$ maps the entire spatially lagged response trajectory $\widetilde{\Y}_i(\cdot)$ into its contribution to the conditional quantile of $\Y_i(t)$. This allows, for example, neighboring pollution levels at time $u$ to affect the local conditional quantile at time $t$, rather than imposing a purely contemporaneous spatial effect. Similarly, $\beta_\tau(t,s)$ maps the local predictor trajectory $\X_i(\cdot)$ into the conditional quantile at $t$. Model~\eqref{eq:quantile-model} can therefore be viewed as the quantile analogue of the mean-based SFoFR model, with both the predictor effect surface $\beta_\tau$ and the spatial spillover surface $\rho_\tau$ allowed to vary with $\tau$. Our focus is on estimating these quantile-specific coefficient surfaces in the presence of spatial dependence. 

Equation~\eqref{eq:quantile-model} is intended as a structural simultaneous IV quantile specification, not as the quantile of a reduced-form mean-SAR transformation and not merely as a two-stage prediction equation. To make this interpretation explicit, for each quantile level $\tau$, define the structural quantile error process $e_{i,\tau}(t) = \Y_i(t) - \int_{\mathcal I_Y}\widetilde{\Y}_i(u)\rho_\tau(t,u) du - \int_{\mathcal I_X}\X_i(s)\beta_\tau(t,s) ds$. The structural restriction is that, after conditioning on valid instruments generated from the exogenous functional predictor and its spatial lags, the $\tau$th conditional quantile of this error is zero. Thus, the spatially lagged response is allowed to be endogenous, but identification is obtained through the IV quantile restrictions stated in Section~\ref{sec:3}.

This distinction is important because quantiles are nonlinear: in general, the conditional quantile of a spatially transformed response is not obtained by applying the same linear inverse operator to the conditional quantile of the innovation process. Therefore, we do not claim that the quantile operator commutes with the functional SAR inverse. Instead, the boundedness condition on the spatial operator is used only to rule out explosive spatial feedback and to ensure that the underlying structural simultaneous system is well-posed. Specifically, let
$(\mathcal R_{\rho_\tau,\bm W}\bm f)_i(t) = \sum_{j=1}^n w_{ij} \int_{\mathcal I_Y} f_j(u)\rho_\tau(t,u) du$. If $\|\mathcal R_{\rho_\tau,\bm W}\|_{\mathrm{op}}<1$, then for any admissible exogenous functional predictor process and structural error process, the linear structural system $\bm\Y = \mathcal R_{\rho_\tau,\bm W}\bm\Y + \mathcal B_{\beta_\tau}\bm X + \bm e_\tau$ has the unique representation $\bm\Y = (\mathbb I-\mathcal R_{\rho_\tau,\bm W})^{-1} \{\mathcal B_{\beta_\tau}\bm X+\bm e_\tau\}$. This uniqueness concerns the structural simultaneous equation itself, not a reduced-form conditional quantile transformation. The interpretation of $\rho_\tau(t,u)$ is therefore structural: it measures the direct quantile-specific spatial spillover in the simultaneous quantile index, while $\beta_\tau(t,s)$ measures the direct quantile-specific effect of the local functional predictor.

\begin{remark}\label{rem:identifiability}
The coefficient surfaces are identifiable only on the part of the tensor-product function space that is excited by the covariate processes and the spatial lag. At the population level, this means that if a pair of perturbation surfaces $(\Delta_\rho,\Delta_\beta)$ satisfies $\int_{\mathcal I_Y}\widetilde{\Y}_i(u)\Delta_\rho(t,u)du+ \int_{\mathcal I_X}\X_i(s)\Delta_\beta(t,s)ds=0$ a.s. for all $t\in\mathcal I_Y$, then $\Delta_\rho=0$ and $\Delta_\beta=0$ in the working function space. Equivalently, after projection onto the spline sieve, the corresponding finite-dimensional structural design must have full column rank. Because the spatially lagged response is endogenous, this rank condition is imposed after the Stage~1 instrumental-variable projection: the instruments $\X_i,\bm W\X_i,\bm W^2\X_i,\ldots$ must be relevant for the spatial lag and must generate a projected design matrix of full column rank for the spline coefficients of both $\rho_\tau$ and $\beta_\tau$. The roughness penalties stabilize estimation and control smoothness, but they are not used to create identification; identification comes from the injectivity/full-rank condition of the IV-projected spline design. The detailed regularity conditions used for the asymptotic theory are collected in the supplementary material.
\end{remark}

\section{Parameter estimation}\label{sec:3}

Before giving the matrix formulation, we briefly summarize the estimation strategy. For a fixed quantile level $\tau$, the proposed SFoF-QR model contains two unknown smooth coefficient surfaces: $\rho_\tau(t,u)$ and $\beta_\tau(t,s)$. The estimation proceeds in four steps. First, we represent both coefficient surfaces by tensor-product B-splines, reducing the infinite-dimensional problem to a finite-dimensional one. Second, because the spatially lagged response is endogenous, we regress it on exogenous functional instruments generated from $\X_i$, $\bm W\X_i$, and $\bm W^2\X_i$ to obtain an instrumented lag-score vector. Third, we plug this instrumented lag into a penalized quantile regression criterion and estimate the spline coefficients of $\rho_\tau$ and $\beta_\tau$. Finally, we select the smoothing parameters by a Bayesian information criterion (BIC) and reconstruct the estimated coefficient surfaces from the fitted spline coefficients.

\subsection{Tensor product B-spline basis expansion and penalized representation}

Direct nonparametric estimation of a bivariate coefficient surface is challenging due to the infinite-dimensional nature of functional variables. We address this by expanding $\rho_{\tau}$ and $\beta_{\tau}$ in a tensor-product B-spline basis, which provides a flexible yet parsimonious representation. Fix two univariate B--spline systems on $[0,1]$: $\phi_1, \phi_2, \ldots$ for the $\Y$-direction and $\psi_1, \psi_2, \ldots$ for the $\X$-direction, each of predetermined order with quasi-uniform knot sequences. For integers $K_y$ and $K_x$, let $\mathcal S_y(K_y)=\operatorname{span} \{\phi_1, \ldots, \phi_{K_y} \}$, $\mathcal S_x(K_x)=\operatorname{span} \{\psi_1, \ldots, \psi_{K_x}\}$, and elevate these marginal spaces to bivariate sieves via tensor products: $\mathcal S_\rho(K_y)=\mathcal S_y(K_y) \otimes \mathcal S_y(K_y)$ and $\mathcal S_\beta(K_y,K_x)=\mathcal S_y(K_y) \otimes \mathcal S_x(K_x)$, respectively. Within these finite-rank manifolds, we represent the latent surfaces by their spline coordinates,
\begin{equation}\label{eq:beta-rho-expansion}
\rho_{\tau}(t,u) \approx \sum_{\ell=1}^{K_y} \sum_{m=1}^{K_y} \rho^{(\tau)}_{\ell m} \phi_{\ell}(t) \phi_{m}(u), \qquad \beta_{\tau}(t,s) \approx \sum_{\ell=1}^{K_y} \sum_{k=1}^{K_x} b^{(\tau)}_{\ell k} \phi_{\ell}(t) \psi_{k}(s),
\end{equation}
so the inferential task becomes the estimation of the coefficient arrays $\rho^{(\tau)}_{\ell m}$ and $b^{(\tau)}_{\ell k}$, from which $\widehat\rho_\tau$ and $\widehat\beta_\tau$ are reconstructed by the same expansions. Crucially, the true kernels $\rho_\tau$ and $\beta_\tau$ are not presumed to live exactly inside these spans. Our asymptotic theory treats $K_y$ and $K_x$ as fixed working dimensions and centers the estimators at the population working-model (pseudo-true) sieve targets $\rho_{\tau,K}$ and $\beta_{\tau,K}$ defined in the supplementary material; the deterministic approximation errors $\|\rho_\tau-\rho_{\tau,K}\|$ and $\|\beta_\tau-\beta_{\tau,K}\|$ are fixed and are not required to vanish (see Assumption~\ref{ass:fixed-spline} and Remark~\ref{rem:fixed-spline-scope}). If the spline dimensions are instead increased along an auxiliary sieve sequence, classical Jackson-type bounds imply that these approximation errors shrink at smoothness-governed rates; see Remark~S4.1 in the online supplementary material. We can think of $\bm{\rho}^{(\tau)} = [\rho^{(\tau)}_{\ell m}]$ as a $K_y \times K_y$ matrix of basis coefficients for the spatial autocorrelation surface, and $\bm{b}^{(\tau)} = [b^{(\tau)}_{\ell k}]$ as a $K_y \times K_x$ coefficient matrix for the predictor effect surface.

By substituting the expansions \eqref{eq:beta-rho-expansion} into the model \eqref{eq:quantile-model}, we obtain a finite-dimensional representation of the SFoF-QR model. To write it compactly, it is helpful to discretize the problem. Suppose that each response curve $\Y_i$ is observed at a fine grid of $R$ points $t_{i1},\ldots,t_{iR} \in [0,1]$ (for dense functional data, we can take these points to be common across $i$, i.e. $t_{i \ell} = t_{\ell}$). Likewise, let the predictor $\X_i$ be observed on a grid of $G$ points $s_{1},\ldots,s_{G}$. In the estimation procedure, we will also need values of the spatial lag $\widetilde{\Y}_i$ on a grid; for convenience, we take the same $R$-point grid for $u$ as for $t$ (since $u$ and $t$ range over the same domain). With these grids, we approximate the integrals in \eqref{eq:quantile-model} using the left-endpoint Riemann rule. Let $0=u_1<\cdots<u_R=1$ and define $\Delta_r=u_{r+1}-u_r$, $r=1,\ldots,R-1$. Then, $\int_0^1 \widetilde{\Y}_i(u)\rho_\tau(t,u) du \approx \sum_{r=1}^{R-1} \Delta_r\widetilde{\Y}_i(u_r)\rho_\tau(t,u_r)$. Substituting the tensor-product expansion of $\rho_\tau$ gives
\begin{equation*}
\sum_{r=1}^{R-1} \Delta_r\widetilde{\Y}_i(u_r) \left\{\sum_{\ell=1}^{K_y}\sum_{m=1}^{K_y} \rho_{\ell m}^{(\tau)} \phi_\ell(t)\phi_m(u_r) \right\} = \sum_{\ell=1}^{K_y}\phi_\ell(t) \left\{\sum_{m=1}^{K_y} \rho_{\ell m}^{(\tau)} \underset{:= \widetilde{\phi}_{m,i}} {\underbrace{ \sum_{r=1}^{R-1} \Delta_r\phi_m(u_r)\widetilde{\Y}_i(u_r)}} \right\}
\end{equation*}
Define $\widetilde{\phi}_{m,i} = \sum_{r=1}^{R-1} \Delta_r \phi_{m}(u_r) \widetilde{\Y}_i(u_r)$, which can be interpreted as the coefficient of the $m$-th basis function in an $\mathcal{L}^2$ projection of the lagged curve $\widetilde{\Y}_i$ onto the B-spline basis $\{\phi_m\}$. Stacking these for $m \in \{1,\dots,K_y \}$, we obtain a vector $\widetilde{\bm{\phi}}_{i} = (\widetilde{\phi}_{1,i}, \ldots,\widetilde{\phi}_{K_y,i})^\top$ of length $K_y$ for each $i$. Analogously, for $0=s_1<\cdots<s_G=1$ and $\Delta_g=s_{g+1}-s_g$, $g=1,\ldots,G-1$, the predictor integral is approximated using the same left-endpoint rule:
\begin{equation*}
\int_{0}^{1} \X_i(s) \beta_{\tau}(t,s) ds \approx \sum_{g=1}^{G-1} \Delta_g \X_i(s_g) \beta_{\tau}(t, s_g) = \sum_{\ell=1}^{K_y} \phi_{\ell}(t) \left\{\sum_{k=1}^{K_x} b_{\ell k}^{(\tau)} \underset{:= \widetilde{\psi}_{k,i}}{\underbrace{\sum_{g=1}^{G-1} \Delta_g \psi_k(s_g) \X_i(s_g)}} \right\},
\end{equation*}
where the $\Delta_g$s are the $s$-grid spacings, and we define $\widetilde{\psi}_{k,i} = \sum_{g=1}^{G-1} \Delta_g \psi_{k}(s_g) \X_i(s_g)$ as the $k$-th basis coefficient for $\X_i$. Let $\widetilde{\bm{\psi}}_{i} = (\widetilde{\psi}_{1,i}, \ldots, \widetilde{\psi}_{K_x,i})^\top$ be the $K_x$-vector of scores for $\X_i$. Substituting both approximations into \eqref{eq:quantile-model} and collecting terms by $\phi_{\ell}$, we obtain the finite-dimensional system:
\begin{align}
Q_{\tau}\{\Y_i(t_{r}) \mid \mathcal{F}_i\} &\approx \sum_{\ell=1}^{K_y} \phi_{\ell}(t_r) \Bigg(\sum_{m=1}^{K_y} \rho^{(\tau)}_{\ell m} \widetilde{\phi}_{m,i} + \sum_{k=1}^{K_x} b^{(\tau)}_{\ell k} \widetilde{\psi}_{k,i}\Bigg), \quad i=1, \ldots, n;~ r = 1, \ldots, R \nonumber \\
&= \left\lbrace \widetilde{\bm{\phi}}_{i}^\top \otimes \bm{\phi}(t_r)^{\top} \right\rbrace \mathrm{vec} \left\lbrace \bm{\rho}^{(\tau)} \right\rbrace + \left\lbrace \widetilde{\bm{\psi}}_{i}^\top \otimes \bm{\phi}(t_r)^{\top} \right\rbrace \mathrm{vec} \left\lbrace \bm{b}^{(\tau)} \right\rbrace \label{eq:approx-model}
\end{align}

The approximation given in~\eqref{eq:approx-model} can be interpreted as a linear quantile regression model for the discretized responses $\Y_i(t_r)$, where the regression variables are the basis functions weighted by the ``covariate scores'' $\widetilde{\phi}_{m,i}$ and $\widetilde{\psi}_{k,i}$. If we concatenate the coefficient vectors, $\bm{\theta}^{(\tau)} = \{\mathrm{vec}(\bm{b}^{(\tau)})^\top, \mathrm{vec}(\bm{\rho}^{(\tau)})^\top\}^\top$, and define the design matrix $\bm{\Pi}$ appropriately, the entire system for all $i,r$ can be written as $\mathrm{vec}(\Y) \approx \bm{\Pi} \bm{\theta}^{(\tau)}$. In particular, $\bm{\Pi}$ has dimension $(n R) \times (K_y K_x + K_y^2)$, and each row of $\bm{\Pi}$ corresponds to a specific observation $\Y_i(t_r)$ and contains the segment $\widetilde{\bm{\psi}}_{i}^\top \otimes \bm{\phi}(t_r)^\top$ in the columns for $\mathrm{vec}(\bm{b}^{(\tau)})$ and $\widetilde{\bm{\phi}}_{i}^\top \otimes \bm{\phi}(t_r)^\top$ in the columns for $\mathrm{vec}(\bm{\rho}^{(\tau)})$. In implementation, the design matrix is additionally augmented by an unpenalized functional-intercept block with $K_0$ basis coefficients $\bm b_0$, accommodating the nonzero level of the response; the full coefficient vector used in Section~\ref{sec:3}.2 is therefore $\bm\theta = \{\bm b_0^\top, \mathrm{vec}(\bm b^{(\tau)})^\top, \mathrm{vec}(\bm\rho^{(\tau)})^\top\}^\top$ of dimension $K_0 + K_yK_x + K_y^2$. For notational simplicity, the intercept block is suppressed in the remainder of this subsection. In the finite-dimensional estimation problem, identifiability therefore requires the IV-projected Stage~2 design to have full column rank, or equivalently that no nonzero linear combination of the spline coefficients of $\rho_\tau$ and $\beta_\tau$ produces the same fitted conditional quantile surface after projection on the instrument space. This condition is the sample analogue of Remark~\ref{rem:identifiability} and is enforced in practice by choosing $K_y$ and $K_x$ so that the number of spline coefficients is compatible with the effective sample size and by using the spatially lagged instruments $\X_i,\bm W\X_i,\bm W^2\X_i$ to identify the endogenous spatial-lag component. In more compact notation, we can write the quantile model as $Q_\tau\{\Y_i(t)\mid \mathcal{F}_i\} = \left\lbrace \widetilde{\bm{\phi}}_{i}^{\top} \otimes \bm{\phi}(t)^{\top} \right\rbrace \widetilde{\bm{\rho}}^{(\tau)} + \left\lbrace \widetilde{\bm{\psi}}_{i}^{\top} \otimes \bm{\phi}(t)^{\top} \right\rbrace \widetilde{\bm{b}}^{(\tau)}$, where $\widetilde{\bm{\rho}}^{(\tau)} = \mathrm{vec} (\bm{\rho}^{(\tau)})$ and $\widetilde{\bm{b}}^{(\tau)} = \mathrm{vec}(\bm{b}^{(\tau)})$. This linear algebra formulation will be convenient for describing the estimation procedure. Moreover, it highlights that our model is linear in the coefficients $\bm{\theta}^{(\tau)}$, even though it is nonlinear in the original data (due to the use of $\widetilde{\Y}_i$ which itself depends on $\Y_j$). The endogeneity of $\widetilde{\Y}_i$ will be handled by an instrumental variable approach in the estimation stage.

Finally, to avoid overfitting and to ensure that the estimated surfaces $\rho_\tau$ and $\beta_\tau$ are smooth in both arguments, we use integrated squared-second-derivative roughness penalties. Define the marginal roughness matrices $\bm D_t,\bm D_u\in\mathbb R^{K_y\times K_y}$ and $\bm D_s\in\mathbb R^{K_x\times K_x}$ by $(\bm D_t)_{\ell\widetilde{\ell}} = \int_0^1 \phi_{\ell}''(t) \phi_{\widetilde{\ell}}''(t) dt$, $(\bm D_u)_{m\widetilde m} = \int_0^1 \phi_m''(u)\phi_{\widetilde m}''(u) du$, and $(\bm D_s)_{k\widetilde k} = \int_0^1 \psi_k''(s)\psi_{\widetilde k}''(s) ds$. Thus, $\bm D_t$, $\bm D_u$, and $\bm D_s$ are Gram matrices of the second derivatives of the corresponding B-spline bases, rather than finite-difference matrices. Define the coefficient-space roughness matrices $\bm R_\rho = \bm\Phi_u\otimes\bm D_t + \bm D_u\otimes\bm\Phi_t$ and $\bm R_\beta = \bm\Psi\otimes\bm D_t + \bm D_s\otimes\bm\Phi_t$, where $\bm{\Phi}_u = \int_0^1 \bm{\phi}(u) \bm{\phi}^\top(u) du$ with $\bm{\phi}(u) = \{\phi_1(u), \ldots, \phi_{K_y}(u) \}^\top$, $\bm{\Phi}_t = \int_0^1 \bm{\phi}(t) \bm{\phi}^\top(t) dt$, and $\bm{\Psi} = \int_0^1 \bm{\psi}(s) \bm{\psi}^\top(s) ds$ with $\bm{\psi}(s) = \{\psi_1(s), \ldots, \psi_{K_x}(s) \}^\top$. The corresponding scalar roughness functionals are
\begin{align*}
\mathcal J_\rho(\rho_\tau) &:= \int_0^1\int_0^1 \{\partial_{tt}\rho_\tau(t,u)\}^2 du dt + \int_0^1\int_0^1 \{\partial_{uu}\rho_\tau(t,u)\}^2 du dt = \widetilde{\bm\rho}^{(\tau)\top} \bm R_\rho \widetilde{\bm\rho}^{(\tau)}, \\ \mathcal J_\beta(\beta_\tau) &:= \int_0^1\int_0^1 \{\partial_{tt}\beta_\tau(t,s)\}^2 ds dt + \int_0^1\int_0^1 \{\partial_{ss}\beta_\tau(t,s)\}^2 ds dt = \widetilde{\bm b}^{(\tau)\top} \bm R_\beta \widetilde{\bm b}^{(\tau)}.
\end{align*}
Thus, $\mathcal J_\rho$ and $\mathcal J_\beta$ are scalar penalty functionals, whereas $\bm R_\rho$ and $\bm R_\beta$ are the associated coefficient-space roughness matrices. In the implemented formulation, $\lambda_\rho$ jointly scales the $t$- and $u$-direction roughness components of $\rho_\tau$, while $\lambda_\beta$ jointly scales the $t$- and $s$-direction components of $\beta_\tau$; thus, smoothing may differ between the two surfaces, but not between the marginal directions within a given surface. Both penalty operators are symmetric positive semidefinite; if the bases include low–order polynomials, the nullspaces correspond to functions with zero curvature in the penalized direction (e.g., linear drift).

Adopting a full tensor–product B–spline expansion, $\beta_{\tau}(t,s) \approx \sum_{\ell=1}^{K_y} \sum_{k=1}^{K_x} b^{(\tau)}_{\ell k} \phi_{\ell}(t) \psi_{k}(s)$, places no separability constraint on the surface. A separable coefficient field of the form $\beta_{\tau}(t,s) = \beta_{\tau,1}(t) \beta_{\tau, 2}(s)$ would imply that the coefficient array $\bm{b}^{(\tau)} = [b_{\ell k}^{(\tau)}] \in \mathbb{R}^{K_y \times K_x}$ factors as an outer product; hence ${\rm rank}(\bm{b}^{(\tau)}) = 1$. Our estimator instead allows $\bm{b}^{(\tau)}$ to be of unrestricted rank, so $\beta_\tau$ may inhabit the entire smooth tensor space spanned by $\phi_\ell \otimes \psi_k $, admitting arbitrary interactions between the $t$- and $s$-margins. 

The penalization contribution may therefore be written either in functional form as $\frac{\lambda_\rho}{2}\mathcal J_\rho(\rho_\tau) + \frac{\lambda_\beta}{2}\mathcal J_\beta(\beta_\tau)$ or in coefficient-matrix form. Let $\bm\vartheta^{(\tau)} = \bigl(\widetilde{\bm b}^{(\tau)\top}, \widetilde{\bm\rho}^{(\tau)\top}\bigr)^{\top}$ and $\bm{\mathcal P}_{\mathrm{surf},\lambda_\beta,\lambda_\rho} = \begin{bmatrix} \lambda_\beta\bm R_\beta & \bm0\\ \bm0 & \lambda_\rho\bm R_\rho \end{bmatrix}$. Then, $\frac{\lambda_\rho}{2}\mathcal J_\rho(\rho_\tau) + \frac{\lambda_\beta}{2}\mathcal J_\beta(\beta_\tau) = \frac12 \bm\vartheta^{(\tau)\top} \bm{\mathcal P}_{\mathrm{surf}, \lambda_\beta, \lambda_\rho} \bm\vartheta^{(\tau)}$. The scalar functionals and matrix penalties are therefore distinct objects.

When data-driven smoothing-parameter selection is used, the smoothing parameters $(\lambda_\beta,\lambda_\rho)$ are selected by a two-dimensional BIC grid search. For every candidate pair, the model is fitted by minimizing the smoothed penalized Stage~2 criterion introduced in Section~\ref{sec:3}.2, and the resulting fit is evaluated using $\operatorname{BIC}(\lambda_\beta,\lambda_\rho) = \log\left[ \frac{1}{N} \sum_{a=1}^{N} \ell_\tau(\widehat u_a) \right] + \frac{\log N}{N} d$, where $N=nR$, $\ell_\tau(u)=u\{\tau-\mathbf 1(u<0)\}$ is the exact check loss, $\widehat u_a$ denotes the fitted residual, and $d=K_0+K_yK_x+K_y^2$ is the number of fitted coefficients, including the functional intercept, regression-surface, and spatial-surface coefficients. Thus, the candidate models are fitted using the differentiable approximation to the check loss, whereas the BIC is evaluated using the exact check loss. The pair minimizing the BIC is selected, and the final model is refitted using that pair. Further implementation details are provided in Section~S3 of the online supplementary material.

\subsection{Two-stage estimation via instrumental quantile regression}

Estimating \eqref{eq:quantile-model} is nontrivial because the spatially lagged response $\widetilde{\Y}_i(u)=\sum_j w_{ij}\Y_j(u)$ is endogenous. As in scalar SAR models, $\Y_i$ and its neighboring responses are jointly determined through the spatial feedback mechanism, so the lag-score vector constructed from $\widetilde{\Y}_i$ is generally correlated with the structural quantile error. Treating $\widetilde{\Y}_i$ as an ordinary covariate in a one-stage quantile regression would therefore violate the quantile-score orthogonality condition and could bias estimation of the spatial spillover surface $\rho_\tau$. To address this, we use a two-stage instrumental quantile regression strategy, extending \citet{Kim_Muller2004} to the present spatial function-on-function setting. The first stage projects the endogenous spatial lag onto instruments generated from the exogenous functional predictor and its spatial lags, such as $\X_i$, $(\bm W\X)_i$, and $(\bm W^2\X)_i$; the second stage replaces the endogenous lag-score vector by its instrumented version in the penalized quantile objective. This preserves the spatial feedback information needed to estimate $\rho_\tau$ while removing the component correlated with the structural quantile error.

We emphasize that this two-stage procedure is not justified by a generic plug-in argument. It is interpreted as a finite-dimensional sieve extension of the two-stage IV quantile framework of \citet{Kim_Muller2004}. Let $\bm Z_i=\{\bm\gamma_i^{(0)\top},\ldots,\bm\gamma_i^{(H-1)\top}\}^\top$ collect the instrument scores generated from $\X_i,\bm W\X_i,\ldots,\bm W^{H-1}\X_i$, and let $\bm\theta_{\tau,K}^{0}$ denote the pseudo-true working-model spline coefficient vector at quantile level $\tau$, defined as the unique solution of the population IV quantile moment condition below (see Assumption~2 in the supplementary material). Identification is based on the structural IV quantile restriction $\Pr\{e_{ir,\tau}(\bm\theta_{\tau,K}^{0})\le0\mid\bm Z_i\}=\tau$, $r=1,\ldots,R$, or equivalently, $\E\left[\bm Z_i \left\{\tau-\mathbf{1} (e_{ir,\tau}(\bm\theta_{\tau,K}^{0})<0) \right\} \right] = \bm0$. Thus, the instruments must be exogenous with respect to the structural quantile error, relevant for the endogenous lag-score vector, and sufficiently rich so that the IV-projected Stage~2 design has full column rank. The first-stage quantile level is matched to the second-stage level because the structural moment condition is quantile-specific; for each $\tau$, the first stage estimates the $\tau$-specific reduced-form component of the endogenous lag scores that enters the corresponding structural quantile equation. The generated-regressor error from this first stage is accounted for in the asymptotic expansion, as detailed in the supplementary proof.

In the first stage, we obtain an instrumented or predicted version of the spatial lag function for each location by regressing $\widetilde{\Y}_i$ on exogenous spatially lagged covariates. Specifically, we construct a set of instrumental functions that are correlated with $\widetilde{\Y}_i$ but uncorrelated with the noise in $\Y_i$. A natural choice of instruments in spatial models are the spatial lags of the original predictor $\X_i$. For example, we can use $\bm{W} \X_i$ and $\bm{W}^2 \X_i$ as instrumental functional predictors. Denote these instrument functions as $\X_i^{(1)} = (\bm{W} \X)_i = \sum_{j}w_{ij} \X_j$ and $\X_i^{(2)} = (\bm{W}^2 \X)_i = \sum_{j,k }w_{ij}w_{jk} \X_k$, and we may also include the original $\X_i$ itself as an instrument. In general, one can consider $H$ instrument functions $\X_i^{(h)} = (\bm{W}^h \X)_i$ for $h=0, 1, \ldots, H-1$ (with $\X_i^{(0)} = \X_i$). Each $\X_i^{(h)}$ is a function on $[0,1]$ and can be expanded in the B-spline basis $\psi_k$, $k \in \{1, \ldots, K_x \}$. We represent $\X_i^{(h)}(s) \approx \sum_{k=1}^{K_x} \gamma_{ik}^{(h)} \psi_k(s)$ and compute the coefficient vector $\bm{\gamma}_{i}^{(h)} = (\gamma^{(h)}_{i1}, \ldots, \gamma^{(h)}_{i K_x})^\top$ by numerical integration (analogous to $\widetilde{\bm{\psi}}_{i}$). Stage~1 then performs a FoF-QR of $\widetilde{\Y}_i$ on the set of instrument functions $\{\X_i^{(h)}\}_{h=0}^{H-1}$. Concretely, for each fixed $u$, consider the $\tau$-quantile of $\widetilde{\Y}_i$ given the instrument processes. We assume a linear model $Q_{\tau} \left\{ \widetilde{\Y}_i(u) \Big| \X_i^{(0)} \X_i^{(1)}, \ldots, \X_i^{(H-1)} \right\} = \sum_{h=0}^{H-1} \int_{0}^{1} \X_i^{(h)}(s) \alpha_{\tau}^{(h)}(u, s) ds$, $u\in\mathcal I_Y$, with appropriate coefficient surfaces $\alpha_{\tau}^{(h)}$. All these surfaces can be expanded in $\phi_m$ and $\psi_k$ bases similar to \eqref{eq:beta-rho-expansion}. Fitting this model via quantile regression (which can be done by existing FoF-QR methods, since it is a standard FoFR without spatial response on the left side) yields estimated surfaces $\widehat{\alpha}_{\tau}^{(h)}$ and thereby an estimated spatial lag $\widehat{\widetilde{\Y}}_i(u) = \sum_{h=0}^{H-1} \int_{0}^{1} \X_i^{(h)}(s) \widehat{\alpha}_{\tau}^{(h)}(u, s) ds$. 

In practice, we carry out this Stage~1 quantile regression in the B-spline space: the instrument scores $\bm{\gamma}_i^{(h)}$ serve as covariates, and the response $\widetilde{\Y}_i$ is expanded in the $\phi$ basis. This is analogous to performing a multivariate quantile regression of the coefficient vectors of $\widetilde{\Y}_i$ on the instrument score vectors. The result is a set of predicted B-spline coefficients for $\widetilde{\Y}_i$, which we denote by $\widehat{\widetilde{\bm{\phi}}}_{i} = (\widehat{\widetilde{\phi}}_{i1}, \ldots, \widehat{\widetilde{\phi}}_{i K_y})^\top$. In other words, we obtain an instrumented version of the lagged-response score vector for each location. This serves as a plug-in replacement for $\widetilde{\bm{\phi}}_{i}$ in the next stage. We emphasize that Stage~1 is crucial for addressing endogeneity: by regressing $\widetilde{\Y}_i$ on $\X$-based instruments, we aim to capture the variation in $\widetilde{\Y}_i$ that is exogenous and filter out the endogenous part correlated with $\Y_i$’s errors. Note that practical guidance on selecting the lag depth $H$, as well as diagnostics and optional dimension reduction for the instrument set, are provided in the online supplementary material \citep[see also][]{Hoshino2024}.

In the second stage, we treat the predicted spatial lag functions $\widehat{\widetilde{\Y}}_i$ (from Stage~1) as known covariates and proceed to estimate the coefficient surfaces $\beta_{\tau}$ and $\rho_{\tau}$ by solving a penalized quantile regression problem. Intuitively, we now fit the model $Q_{\tau}\{\Y_i\} \approx \int \widehat{\widetilde{\Y}}_i(u) \rho_{\tau}(t,u) du + \int \X_i(s) \beta_{\tau}(t,s) ds$ in the least-absolute-deviation sense (for general $\tau$) with penalties to enforce smoothness. Formally, we define a loss function for a given $\tau$:
\begin{equation}\label{eq:loss}
L_{\tau}(\rho_{\tau},\beta_{\tau}) = \frac{1}{n} \sum_{i=1}^n \int_0^1 \eta_{\tau} \left\{\Y_i(t) - \int_0^1 \widehat{\widetilde{\Y}}_i(u) \rho_{\tau}(t,u) du - \int_0^1 \X_i(s) \beta_{\tau}(t,s) ds\right\} dt,
\end{equation}
where $\eta_{\tau}(c) = c \varphi_\tau(c)$ with standard check (pinball) loss function for quantile $\tau$, $\varphi_\tau(c) = \{\tau - \mathrm{1}(c < 0) \}$. This objective $L_{\tau}$ sums the residual quantile losses over all subjects $i$ and all time points $t$. In practice, the integral in \eqref{eq:loss} is approximated by a sum over the time grid. We then add the roughness penalties for $\rho_{\tau}$ and $\beta_{\tau}$ discussed above, with tuning parameters $\lambda_{\rho}$ and $\lambda_{\beta}$. The penalized quantile criterion is:
\begin{equation}\label{eq:pen-criterion}
Q_\tau(\rho_\tau,\beta_\tau) = L_\tau(\rho_\tau,\beta_\tau) + \frac{\lambda_\rho}{2}\mathcal J_\rho(\rho_\tau) + \frac{\lambda_\beta}{2}\mathcal J_\beta(\beta_\tau).
\end{equation}
We seek $\widehat{\bm{\theta}}^{(\tau)} = \arg\min_{\bm{\theta}^{(\tau)}} Q_\tau(\bm{\theta}^{(\tau)})$. This minimization is performed under the linear model approximation developed in the previous subsection, with $\widetilde{\bm{\phi}}_{i}$ replaced by the Stage~1 fitted $\widehat{\widetilde{\bm{\phi}}}_{i}$. For the matrix formulation, let $\bm\theta^{(\tau)} = \left\{\bm b_0^\top, \operatorname{vec}\bigl(\bm b^{(\tau)}\bigr)^\top, \operatorname{vec}\bigl(\bm\rho^{(\tau)}\bigr)^\top \right\}^{\top}$ where $\bm b_0$ denotes the functional-intercept basis coefficients. Define the full block penalty matrix
\begin{equation*}
\bm{\mathcal P}_{\lambda_\beta,\lambda_\rho} = 
\begin{bmatrix}
\bm0 & \bm0 & \bm0\\
\bm0 & \lambda_\beta\bm R_\beta & \bm0\\
\bm0 & \bm0 & \lambda_\rho\bm R_\rho
\end{bmatrix}.
\end{equation*}
where the leading zero block has dimension equal to that of $\bm b_0$. The discretized penalized criterion is then $Q_\tau(\bm\theta^{(\tau)}) = \frac{1}{n} \sum_{i=1}^{n}\sum_{r=1}^{R} \eta_\tau \left\{\Y_i(t_r) - [\widehat{\bm\Pi}\bm\theta^{(\tau)}]_{ir} \right\} + \frac12 \bm\theta^{(\tau)\top} \bm{\mathcal P}_{\lambda_\beta,\lambda_\rho} \bm\theta^{(\tau)}$. The term inside $\eta_{\tau}( \cdot )$ is the quantile regression residual for observation $(i,t_r)$. For computation, we use a differentiable approximation to the check loss. Let $u_{ir}(\bm\theta)=\Y_i(t_r)-[\widehat{\bm\Pi}\bm\theta]_{ir}$ denote the residual corresponding to the IV-projected Stage~2 design. Instead of the nonsmooth check loss $\ell_\tau(u)=u\{\tau-\mathbf{1}(u<0)\}$, we minimize its smooth approximation $\ell_{\tau,\alpha}(u) = \tau u+\alpha\log\{1+\exp(-u/\alpha)\}$, $\alpha>0$, which converges pointwise to $\ell_\tau(u)$ as $\alpha\downarrow0$. Thus, for fixed $(\lambda_\rho,\lambda_\beta)$, the final Stage~2 optimization criterion is
\begin{equation}\label{eq:final-smooth-objective}
\widehat{\bm\theta}^{(\tau)} = \argmin_{\bm\theta} \left[ \frac{1}{n}\sum_{i=1}^{n}\sum_{r=1}^{R} \ell_{\tau,\alpha}\{u_{ir}(\bm\theta)\} + \frac{1}{2}\bm\theta^\top \bm{\mathcal P}_{\lambda_\beta,\lambda_\rho} \bm\theta
\right],
\end{equation}
where $\bm\theta = \left\{\bm b_0^\top, \mathrm{vec}(\bm b^{(\tau)})^\top, \mathrm{vec}(\bm\rho^{(\tau)})^\top \right\}^{\top}$ and $\bm{\mathcal P}_{\lambda_\beta,\lambda_\rho}$ is the block roughness penalty, with the intercept block unpenalized, the $\beta_\tau$ block multiplied by $\lambda_\beta$, and the $\rho_\tau$ block multiplied by $\lambda_\rho$. The derivative of the smooth loss is $\ell'_{\tau,\alpha}(u) = \tau-\{1+\exp(u/\alpha)\}^{-1}$ so the analytic gradient used in the optimization is $\nabla_{\bm\theta} = -\frac{1}{n}\widehat{\bm\Pi}^{\top} \bm\psi_{\tau,\alpha}(\bm u) + \bm{\mathcal P}_{\lambda_\beta,\lambda_\rho}\bm\theta$ where $\bm\psi_{\tau,\alpha}(\bm u)$ stacks $\ell'_{\tau,\alpha}\{u_{ir}(\bm\theta)\}$ over all $(i,r)$. In all numerical studies we used $\alpha=10^{-2}$ and optimized \eqref{eq:final-smooth-objective} by L-BFGS-B with this analytic gradient; a small ridge constant $10^{-8}$ was added to the penalty matrix only for numerical stabilization. Further computational details of the smooth L-BFGS-B optimization, including the objective, gradient, initialization, convergence criterion, and tuning grid, are provided in the online supplementary material.

After convergence of the Stage~2 optimization, we obtain the estimates $\widehat{\rho}_{\tau}$ and $\widehat{\beta}_{\tau}$ by plugging the estimated coefficients $\widehat{\rho}^{(\tau)}_{\ell m}$ and $\widehat{b}^{(\tau)}_{\ell k}$ into the expansions \eqref{eq:beta-rho-expansion}, as follows:
\begin{equation*}
\widehat{\rho}_{\tau}(t,u) = \{\bm{\phi}^\top(u) \otimes \bm{\phi}^\top(t)\} \widehat{\widetilde{\bm{\rho}}}^{(\tau)}, \qquad \widehat{\beta}_{\tau}(t,s) = \{\bm{\psi}^\top(s) \otimes \bm{\phi}^\top(t)\} \widehat{\widetilde{\bm{b}}}^{(\tau)}.
\end{equation*}

The following result describes the large-sample behavior of the proposed estimators in the working spline space. The detailed regularity conditions and proofs are provided in the online supplementary material; the conditions below highlight those most relevant for interpreting the result.

\begin{assumption}\label{ass:fixed-spline} 
Let $\widehat{\rho}_{\tau}$ and $\widehat{\beta}_{\tau}$ denote the proposed two-stage penalized IV quantile estimators, and let $\rho_{\tau,K}$ and $\beta_{\tau,K}$ denote the corresponding population spline-sieve targets in the tensor-product spaces with dimensions $K_y$ and $K_x$. We assume that: \emph{(i)} the spatial weight matrix $\bm W$ is fixed or conditionally fixed, has zero diagonal and uniformly bounded row and column sums, and the structural spatial operator satisfies $\|\mathcal R_{\rho_\tau,\bm W}\|_{\mathrm{op}}\leq1-\delta$ for some $\delta>0$; \emph{(ii)} the true coefficient surfaces are sufficiently smooth that the population spline-sieve targets are well defined, with approximation error $a_K=\|\rho_\tau-\rho_{\tau,K}\|_\infty+ \|\beta_\tau-\beta_{\tau,K}\|_\infty$; for the fixed-spline result below, $a_K$ is not required to vanish, whereas $a_K\to0$ if the spline dimensions increase along an auxiliary sieve sequence; \emph{(iii)} the response and predictor grids are sufficiently dense that the numerical quadrature and discretization errors are $o(n^{-1/2})$ at the fixed evaluation points considered below; \emph{(iv)} the smoothing parameters satisfy $ \lambda_\rho\vee\lambda_\beta=o(n^{-1/2})$, so that the penalty stabilizes estimation without introducing first-order bias; \emph{(v)} the instruments satisfy the structural IV quantile exogeneity and relevance/rank conditions described in Section~\ref{sec:3}; \emph{(vi)} if $\widetilde{\bm\phi}_{i,\tau}^{IV}$ denotes the population first-stage instrumented lag-score vector, then $\widehat{\widetilde{\bm\phi}}_i- \widetilde{\bm\phi}_{i,\tau}^{IV}=O_p(n^{-1/2})$ uniformly for fixed spline dimension, and its asymptotically linear representation is included in $\bm\Sigma_{\theta,\tau}$; \emph{(vii)} the conditional density of the structural quantile error at zero exists, is bounded away from zero and infinity, and is Lipschitz in a neighborhood of zero; and \emph{(viii)} the smoothing constant in the differentiable check-loss approximation satisfies $\alpha=\alpha_n\downarrow0$ and $\sqrt n\,\alpha_n\longrightarrow0$. 
\end{assumption}

\begin{theorem}\label{th:1} 
Under Assumption~\ref{ass:fixed-spline}, for fixed spline dimensions $K_y$ and $K_x$, $\sqrt n \left( \widehat{\bm\theta}^{(\tau)} - \bm\theta_{\tau,K}^{0} \right) \xrightarrow{d} N \left( \bm0,\bm\Sigma_{\theta,\tau} \right)$ where $\bm\theta_{\tau,K}^{0}$ is the pseudo-true working-model spline coefficient vector defined by the population IV quantile moment condition (Assumption~2 in the supplementary material), and $\bm\Sigma_{\theta,\tau}$ includes both the second-stage quantile-score variability and the first-stage generated-regressor contribution.\end{theorem} 

\begin{corollary}\label{cor:surface-limits}
Under the conditions of Theorem~\ref{th:1}, for any fixed finite collection of evaluation points $(t,u)$ and $(t,s)$, $\sqrt n \left\{ \widehat{\rho}_{\tau}(t,u)-\rho_{\tau,K}(t,u) \right\}$ and $\sqrt n \left\{ \widehat{\beta}_{\tau}(t,s)-\beta_{\tau,K}(t,s) \right\}$ converge jointly to a mean-zero Gaussian vector. Consequently, the estimators are pointwise asymptotically normal and jointly asymptotically normal over every fixed finite grid. Because $K_y$ and $K_x$ are fixed and the spline basis functions are continuous, the reconstruction maps from the coefficient vectors to the coefficient surfaces are continuous and finite dimensional. Therefore, the coefficient central limit theorem in Theorem~\ref{th:1} also implies $\sqrt n \left( \widehat{\rho}_{\tau}-\rho_{\tau,K} \right) \rightsquigarrow \mathbb G_{\rho,\tau}$ in $C(\mathcal I_Y\times\mathcal I_Y)$ and $\sqrt n \left( \widehat{\beta}_{\tau}-\beta_{\tau,K} \right) \rightsquigarrow \mathbb G_{\beta,\tau}$ in $C(\mathcal I_Y\times\mathcal I_X)$, where $\mathbb G_{\rho,\tau}$ and $\mathbb G_{\beta,\tau}$ are mean-zero finite-rank Gaussian processes with covariance kernels $\Sigma_{\rho,\tau}(t,u;t',u') = \bm a_\rho(t,u)^\top \bm\Sigma_{\rho,\tau}\, \bm a_\rho(t',u')$ and $\Sigma_{\beta,\tau}(t,s;t',s') = \bm a_\beta(t,s)^\top \bm\Sigma_{\beta,\tau}\, \bm a_\beta(t',s')$, respectively, where $\bm a_\rho(t,u)=\bm\phi(u)\otimes\bm\phi(t)$ and $\bm a_\beta(t,s)=\bm\psi(s)\otimes\bm\phi(t)$ are the tensor-product basis evaluation vectors, and $\bm\Sigma_{\rho,\tau}$ and $\bm\Sigma_{\beta,\tau}$ denote the $K_y^2\times K_y^2$ and $K_yK_x\times K_yK_x$ blocks of $\bm\Sigma_{\theta,\tau}$ corresponding to $\widetilde{\bm\rho}^{(\tau)}$ and $\widetilde{\bm b}^{(\tau)}$. If the true coefficient surfaces belong to the fixed working spline spaces, then $\rho_{\tau,K}=\rho_\tau$ and $\beta_{\tau,K}=\beta_\tau$, and hence the true surfaces may be used as the centering functions.
\end{corollary}

\begin{remark}\label{rem:fixed-spline-scope} 
The function-space convergence in Corollary~\ref{cor:surface-limits} is a fixed-spline result. If $K_y$ or $K_x$ increases with $n$, additional high-dimensional and uniform-tightness conditions would be required to establish a whole-surface $\sqrt n$ weak limit; no such growing-dimension function-space result is claimed here.
\end{remark}

\section{Monte Carlo experiments}\label{sec:4}

To systematically evaluate the finite-sample performance of the proposed SFoF-QR estimator, we conduct a Monte Carlo study in which all data are generated from a functional SAR model on the compact domains $\mathcal{I}_x=\mathcal{I}_y=[0,1]$. For $n$ spatial locations, we observe paired functions $\{(\X_i,\Y_i)\}_{i=1}^n$ and generate responses according to $\Y_{i}(t) = \int_0^1 \Bigg\{\sum_{j=1}^{n} w_{ij} \Y_{j}(u) \Bigg\} \rho(t,u) du + \int_0^1 \X_{i}(s)\beta(t,s) ds + \epsilon_{i}(t)$, $t\in[0,1]$. For compactness, define the stacked response vector $\Y=(\Y_1,\ldots,\Y_n)^\top$ and the stacked error $\epsilon=(\epsilon_1,\ldots,\epsilon_n)^\top$. Let $\mathcal{G}(t)=\Big\{\int_0^1 \X_1(s)\beta(t,s) ds, \ldots, \int_0^1 \X_n(s)\beta(t,s) ds\Big\}^\top$ denote the functional linear component, and define the functional SAR operator $\mathcal{R}_{\bm W}$ by its action on an $n$-vector of functions $\bm f(\cdot)$: $(\mathcal{R}_{\bm W}\bm f)(t)=\bm W\int_0^1 \bm f(u) \rho(t,u) du$. Then the data-generating model can be written as $\Y(t)=\mathcal{R}_{\bm W}\{\Y\}(t)+\mathcal{G}(t)+\epsilon(t)$. When $(\mathbb{I}-\mathcal{R}_{\bm W})$ is invertible (e.g., $\|\mathcal{R}_{\bm W}\|<1$), the unique reduced-form solution is $\Y(t)=(\mathbb{I}-\mathcal{R}_{\bm W})^{-1}\{\mathcal{G}(t)+\epsilon(t)\}$. 
In implementation, we generate $\Y$ by fixed-point iteration (Neumann-series recursion) $\Y_{k+1}=\mathcal{R}_{\bm W}(\Y_k)+\mathcal{G}+\epsilon$, which converges to the reduced-form solution under $\|\mathcal{R}_{\bm W}\|<1$. The predictors $\X_i$ are generated independently from a smooth $\mathcal{L}^2$ process using a truncated Fourier expansion with $n_\phi=10$ cosine and $n_\phi=10$ sine terms: $\X_i(s)=\sum_{j=1}^{n_{\phi}} \xi_{i,j}^{(1)}\Big\{j^{-1.5}\sqrt{2}\cos(j\pi s)\Big\}
+\sum_{j=1}^{n_{\phi}} \xi_{i,j}^{(2)}\Big\{j^{-1.5}\sqrt{2}\sin(j\pi s)\Big\}$ where $\xi_{i,j}^{(1)},\xi_{i,j}^{(2)}\stackrel{i.i.d.}{\sim}N(0,1)$. The $j^{-1.5}$ decay yields trace-class covariance and smooth sample paths.

Spatial dependence is induced through a row-normalized inverse-distance weight matrix $\bm W$. Specifically, for $i\neq j$, let $\widetilde w_{ij}=(1+|i-j|)^{-1}$ and $\widetilde w_{ii}=0$, and set $w_{ij}=\widetilde w_{ij}/\sum_{k=1}^n \widetilde w_{ik}$. The regression surface is chosen to be smooth and non-separable, $\beta(t,s)=2+s+t+0.5\sin(2\pi st)$, and the spatial kernel is $\rho(t,u)=s_d \frac{1+ut}{1+|u-t|}$, where $s_d\in\{0.1,0.5,0.9\}$ corresponds to weak, moderate, and strong spatial dependence.

We focus on a challenging error design (Case 2) combining signal-dependent heteroskedasticity and upper-tail contamination. Let $\mu_i$ denote the signal (non-error) component of $\Y_i$ and $\mu_i^{(a)}=|\mu_i|$. After rescaling $\mu_i^{(a)}$ to $[0,1]$ we set a time-location specific scale $\sigma_i=\sigma_\epsilon(1+a_\sigma \mu_i^{a,\star})$ and generate $\epsilon_i=\sigma_i\iota_i(t)+B_i c_{\text{out}}\sigma_i$, where $\iota_i\stackrel{i.i.d.}{\sim}N(0,1)$ and $B_i \sim \text{Bernoulli} (p_i)$ with $p_i=p_0+p_1\mu_i^{a,\star}$. Thus, both the variance and the probability of a one-sided positive shock increase with the local signal strength, producing a strongly right-skewed, heteroskedastic conditional error distribution. For completeness, we also considered ideal Gaussian errors (Case 1) and symmetric heavy-tailed errors (Case 3); detailed designs and results are reported in the online supplementary material.

We benchmark the proposed method (M$_6$) against five competitors: the FPCA-based SFoFR of \citet{BSGARC2025} (M$_1$), the penalized mean-based SFoFR of \citet{BSS2025} (M$_2$), the penalized FoFR of \citet{ivanescu2015} (M$_3$), boosted mean regression via \texttt{FDboost} \citep{FDboost} (M$_4$), and function-on-function linear quantile regression \citep{BSMMA} (M$_5$). The proposed method $\mathrm{M}_6$ and the penalized SFoFR model $\mathrm{M}_2$ are both implemented with $K_y=K_x=10$ and $\lambda_\rho=\lambda_\beta=0.001$. We additionally report $\mathrm{M}_6^{(\mathrm{Opt})}$, which differs from $\mathrm{M}_6$ only in that $(\lambda_\rho,\lambda_\beta)$ is selected by minimizing BIC over a prespecified candidate grid. We consider training sample sizes $n\in\{100,250,500\}$ and evaluate on an independent test set of size $n_{\text{test}}=1000$. Each configuration is repeated for $N_{\text{sim}}=250$ Monte Carlo replications across the three dependence levels and three error cases. All competing methods were evaluated on the same Monte Carlo replications; the main paper reports the Case~2 results for $n=500$, while the full set of results is provided in the supplementary material. The supplementary material also reports representative computing times for all methods under Case~2; for the proposed M$_6$, one complete replication ranges from 4.31 seconds for $n=100$ and weak dependence to 54.34 seconds for $n=500$ and strong dependence.

For $\mathrm{M}_6^{(\mathrm{Opt})}$, the spline dimensions were fixed at $K_y=K_x=10$, while $(\lambda_\rho,\lambda_\beta)$ was selected from $\Lambda_\rho=\Lambda_\beta= \{10^{-4},10^{-3},10^{-2},10^{-1},1\}$ by minimizing BIC for each Monte Carlo replication and quantile level.

Performance is assessed using $\mathcal{L}^2$-based metrics. For surface estimation, we report the relative root integrated squared percentage error (RRISPEE) for $\beta$ and $\rho$, $\mathrm{RRISPEE}(\widehat{\beta}) = 100\left\{\frac{\|\beta- \widehat{\beta}\|_2^2}{\|\beta\|_2^2} \right\}^{1/2}$, $\mathrm{RRISPEE}(\widehat{\rho}) = 100\left\{\frac{\|\rho-\widehat{\rho}\|_2^2}{\|\rho\|_2^2} \right\}^{1/2}$. For prediction, we use the root mean squared percentage error on independent test functions, $\mathrm{RMSPE} = 100\left\{\frac{\|\Y^{\mathrm{new}}-\widehat{\Y}^{\mathrm{new}}\|_2^2} {\|\Y^{\mathrm{new}}\|_2^2} \right\}^{1/2}$. For the quantile methods, the reported surface-recovery metrics compare $\widehat{\beta}_{0.5}$ and $\widehat{\rho}_{0.5}$ with the known structural data generating process surfaces $\beta$ and $\rho$ appearing in the additive SAR data-generating equation. Under Cases~1 and~3, the errors are symmetric with conditional median zero, so these structural surfaces also coincide with the coefficient surfaces of the conditional median equation. Under the heteroskedastic contaminated Case~2, however, the error scale and the probability of the one-sided positive contamination depend on the signal. Consequently, $Q_{0.5}\{\epsilon_i(t)\mid\text{signal}\}$ is generally nonzero and signal-dependent, and the exact conditional median need not be represented by the same additive coefficient surfaces. Thus, for M$_5$ and M$_6$ in Case~2, $\mathrm{RRISPEE}(\widehat{\beta})$ and $\mathrm{RRISPEE}(\widehat{\rho})$ are structural-DGP-surface recovery measures, not estimation errors relative to separately derived median coefficient surfaces. Prediction accuracy is assessed separately through RMSPE.

Figure~\ref{fig:tab2_graphical_summary} summarizes the main Case~2 simulation results for $n=500$. The three panels display $\mathrm{RRISPEE}(\widehat\beta)$, $\mathrm{RRISPEE}(\widehat\rho)$, and RMSPE across the spatial-dependence levels $s_d=0.1,0.5,0.9$. These are the main surface-recovery and prediction metrics common to the competing methods. The graphical summary shows that the proposed spatial quantile estimator M$_6$ and its BIC-tuned version M$_6^{(\mathrm{Opt})}$ recover the structural regression and spatial-spillover surfaces accurately under asymmetric contamination. The advantage of spatial methods becomes clearer as spatial dependence strengthens, while non-spatial competitors deteriorate because they omit the SAR component. The RMSPE panel also shows that incorporating the spatial lag improves test-set prediction accuracy, especially under moderate and strong dependence.

Although Figure~\ref{fig:tab2_graphical_summary} focuses on the heteroskedastic upper-tail contamination setting, the supplementary results for Gaussian errors (Case~1) and symmetric heavy-tailed errors (Case~3) lead to the same qualitative conclusion. Under Gaussian errors, where mean-based spatial FoFR methods are expected to perform well, M$_6$ remains competitive with the penalized mean-based spatial estimator M$_2$. Under symmetric heavy-tailed errors, M$_6$ retains the benefit of spatial modeling while preserving the quantile interpretation. Thus, the gains in Case~2 do not arise solely from choosing a setting favorable to quantile regression; rather, M$_6$ remains competitive in the Gaussian benchmark and becomes especially useful when the error distribution is heavy-tailed or asymmetric. We refer readers to the supplementary material for complete numerical results across all simulation settings and for the finite-sample check of the pointwise asymptotic confidence intervals for $\beta_\tau$ and $\rho_\tau$.

\begin{figure}[!htb]
\centering
\includegraphics[width=.98\textwidth]{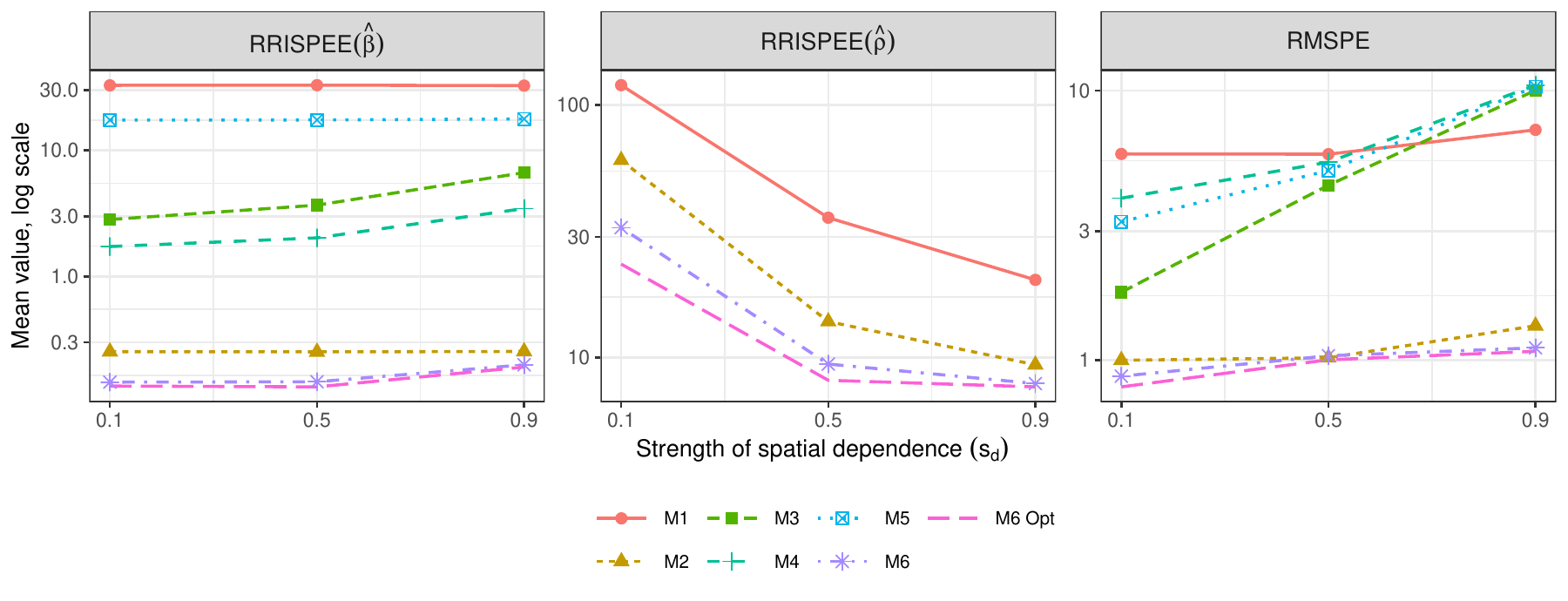}
\caption{Graphical summary of the main Case~2 simulation metrics for $n=500$ across three spatial-dependence levels, $s_d=0.1,0.5,0.9$. The panels report the Monte Carlo mean $\mathrm{RRISPEE}(\widehat{\beta})$, $\mathrm{RRISPEE}(\widehat{\rho})$, and RMSPE for the competing methods. The vertical axis is shown on a logarithmic scale to improve readability because the methods differ substantially in magnitude.}
\label{fig:tab2_graphical_summary}
\end{figure}
\vspace{-1cm}

\section{Air quality data analysis}\label{sec:5}

The empirical analysis uses the PM$_{2.5}$/PM$_{10}$ air-quality data introduced in Section~\ref{sec:1} and Figure~\ref{fig:motivation}. Briefly, the data are obtained from the \texttt{ARPALData}  \Rlogo \ package \citep{arpaldata} and cover January~1,~2023 to December~31,~2024. For station $i$, the functional response $\Y_i$ is the monthly mean PM$_{2.5}$ concentration $(\mu g/m^3)$, and the functional predictor $\X_i$ is the monthly mean PM$_{10}$ concentration $(\mu g/m^3)$. This predictor is meaningful because PM$_{2.5}$ is a subset of PM$_{10}$, although their relationship varies over time according to source composition and meteorological conditions \citep{Gianquintieri2023}. The monthly series are observed at the grid points $t_r=s_r=r$, $r=1,\ldots,24$, and are converted to functional objects by cubic B-spline smoothing with 13 equally spaced knots over the continuous common domain $\mathcal I_Y=\mathcal I_X=[1,24]$. The PM$_{2.5}$ and PM$_{10}$ curves are fitted on their original concentration scales and are not centered globally, by month, or by station; the nonzero response level is accommodated by the unpenalized functional intercept included in the Stage~2 model.

For the SAR component, we construct a row-standardized $K$-nearest-neighbors weight matrix $\bm W=(w_{ij})$ with $K_n=8$, selected by cross-validation. Letting $N_8(i)$ denote the eight nearest stations to site $i$, we define $w_{ij} = d_{ij}^{-1}/\sum_{k\in N_8(i)} d_{ik}^{-1}$ when $j\in N_8(i)$, and $w_{ij}=0$ otherwise, where $d_{ij}$ is the Euclidean distance between sites $i$ and $j$. This yields the functional spatial lag $\widetilde{\Y}_i(u)=\sum_{j=1}^n w_{ij}\Y_j(u)$, which is interpreted as a locally weighted average of neighboring PM$_{2.5}$ trajectories. As shown in Figure~\ref{fig:motivation}, the functional Moran's $I$ statistic remains approximately between 0.94 and 0.96, confirming strong positive spatial dependence throughout the study period.

For the PM$_{2.5}$ application, we used $K_y=K_x=10$ and selected $(\lambda_\rho,\lambda_\beta)$ separately for each $\tau\in\{0.1,0.2,\ldots,0.9\}$ from $\Lambda_\rho=\Lambda_\beta=\{10^{-4},10^{-3},10^{-2},10^{-1},1\}$ using the BIC criterion in Section~\ref{sec:3}.

Using the available air-quality data, we fit the proposed SFoF-QR model $Q_{\tau}\{\Y_i(t)\mid \widetilde{\Y}_i,\X_i\} = \int_{1}^{24}\widetilde{\Y}_i(u)\rho_{\tau}(t,u) du + \int_{1}^{24}\X_i(s)\beta_{\tau}(t,s) ds$, $t\in[1,24]$, for each quantile level $\tau\in\{0.1,0.2,\ldots,0.9\}$. In computation, this pointwise functional model is evaluated at the observed monthly grid points $t_r=r$, $r=1,\ldots,24$. Because an independent external test sample is not available for this application, the numerical results in Table~\ref{tab:tab_6} should be interpreted as in-sample fitted-performance summaries rather than out-of-sample prediction errors. We also note that the out-of-sample prediction in spatial autoregressive models requires specifying how the prediction locations are connected to the training locations. In the simulation study, training and test samples are generated as separate sets with no common sample points, so test prediction is well defined from the new functional predictors and the corresponding spatial structure. If training and prediction samples share spatial units or if the prediction set is spatially linked to the training response through the SAR lag, additional prediction strategies, such as the trend-corrected approach of \citet{Goulard2017}, may be used; incorporating such spatial forecasting corrections is beyond the scope of the present paper.

At the median level $\tau=0.5$, we compare the proposed method M$_6$ with five competing approaches. M$_1$ denotes the FPCA-based spatial function-on-function mean regression estimator; M$_2$ denotes the penalized B-spline spatial function-on-function mean regression estimator; M$_3$ denotes a classical non-spatial function-on-function mean regression estimator; M$_4$ denotes a boosted non-spatial function-on-function regression estimator; and M$_5$ denotes a non-spatial function-on-function quantile regression estimator. Hence, M$_1$ and M$_2$ are spatial mean-regression competitors, M$_3$ and M$_4$ are non-spatial mean-regression competitors, and M$_5$ is the non-spatial quantile-regression analogue of the proposed spatial quantile method M$_6$. To examine distributional heterogeneity, we also compare M$_5$ and M$_6$ over $\tau\in\{0.1,0.2,\ldots,0.9\}$.

The fitted performance is summarized by the relative functional RMSE and the functional coefficient of determination, $\mathrm{RMSE} = 100\times \sqrt{ \frac{\sum_i\int_{1}^{24}\{\Y_i(t)-\widehat{\Y}_i(t)\}^2dt} {\sum_i\int_{1}^{24}\{\Y_i(t)\}^2dt}}$, and $R^2 = 1- \frac{\sum_{i=1}^{n}\int_{1}^{24} \{\Y_i(t)-\widehat{\Y}_i(t)\}^2 dt}{ \sum_{i=1}^{n}\int_{1}^{24}\{\Y_i(t)-\overline{\Y}(t)\}^2 dt}$ where $\widehat{\Y}_i$ is the fitted PM$_{2.5}$ curve for station $i$ and
$\overline{\Y}(t)=n^{-1}\sum_{i=1}^{n}\Y_i(t)$ is the pointwise mean PM$_{2.5}$ curve across all stations. For quantile-regression methods, the fitted curve is not a conditional mean but the estimated conditional quantile curve $\widehat{Q}_{\tau,i}(t)=\widehat{Q}_{\tau}\{\Y_i(t)\mid\widetilde{\Y}_i,\X_i\}$. Hence, the reported $R^2$ for M$_5$ and M$_6$ is an $R^2$-type squared-error agreement measure and should not be interpreted as the usual conditional-mean variance-explained measure. Quantile-specific performance is therefore also evaluated using quantile-oriented criteria such as check loss, empirical coverage, CPD, and interval score. Because these quantities are computed on the same data used for fitting, they measure in-sample fitted accuracy and should not be interpreted as independent test-set prediction performance.

The results in Table~\ref{tab:tab_6} show that accounting for spatial dependence is important for explaining the observed PM$_{2.5}$ curves. Interpreted as in-sample fitted accuracy, the proposed SFoF-QR method gives the smallest RMSE, $6.316$, and an $R^2$-type agreement measure of $0.970$, closely matching the penalized spatial mean model M$_2$ and clearly improving over the non-spatial alternatives M$_3$-M$_5$. In particular, the non-spatial quantile model M$_5$ yields RMSE $=8.556$, whereas incorporating the functional spatial lag in M$_6$ reduces this error by about 26\%. We also compare M$_5$ and M$_6$ in terms of prediction interval performance where the interval summaries are obtained from the constructed 95\% prediction intervals from $\tau\in\{0.025,0.975\}$ and are evaluated via functional empirical coverage (EC), coverage probability deviance (CPD), and the integrated interval score (score). M$_6$ provides substantially better in-sample interval summaries than M$_5$, with EC $=0.949$, CPD $=0.071$, and score $=4.365$. Additional fitted-curve comparisons for all six methods are
provided in the online supplementary material.

\begin{table}[!htb]
\tabcolsep 0.3in
\caption{\small In-sample fitted performance for the air-quality data. The table reports RMSE, $R^2$, 95\% nominal EC, CPD, and interval score. The RMSE and $R^2$ values for the quantile methods are obtained at $\tau=0.5$, whereas EC, CPD, and interval score are computed from the fitted lower and upper quantile curves at $\tau=0.025$ and $0.975$. Metrics EC, CPD, and score are not applicable for mean-regression models and are reported as ----.}\label{tab:tab_6}
\begin{center}
\begin{small}
\begin{tabular}{@{}lcccccc@{}} 
\toprule
{Metric} & \multicolumn{6}{c}{Method} \\
& M$_1$ & M$_2$ & M$_3$ & M$_4$ & M$_5$ & M$_6$ \\
\midrule
RMSE & 14.711 & 6.333 & 12.211 & 34.996 & 8.556 & 6.316 \\
$R^2$ & 0.887 & 0.970 & 0.926 & 0.499 & 0.950 & 0.970 \\
EC & ---- & ---- & ---- & ---- & 0.000 & 0.949 \\
CPD  & ---- & ---- & ---- & ---- & 0.950 & 0.071 \\
score & ---- & ---- & ---- & ---- & 99.484 & 4.365 \\
\bottomrule
\end{tabular}
\end{small}
\end{center}
\end{table}

Table~\ref{tab:tab_7} shows that M$_6$ outperforms the non-spatial quantile model M$_5$ at every quantile level. Both methods exhibit the expected U-shaped RMSE pattern, with larger fitted errors in the tails, but the advantage of M$_6$ persists throughout the conditional distribution. This improvement remains important in the upper tail, where pollution-risk assessment is particularly relevant: at $\tau=0.9$, the RMSE decreases from 15.366 for M$_5$ to 12.706 for M$_6$. Figure~\ref{fig:Fig_4} displays representative lower, median, and upper fitted quantile curves. Their separation is substantially greater in winter than in summer, indicating pronounced state-dependent heteroscedasticity that cannot be adequately represented by homoscedastic mean regression. The complete fitted-curve display for all nine quantile levels is provided in the online
supplementary material.

\begin{table}[!htb]
\tabcolsep 0.02in
\caption{\small Computed RMSE and $R^2$ values for the quantile regression models (M$_5$ and M$_6$) from the air quality data. The performance metrics are obtained for several quantile levels, i.e., $\tau\in\{0.1,0.2,\ldots,0.9\}$. }\label{tab:tab_7}
\begin{center}
\begin{scriptsize}
\begin{tabular}{@{}lcccccccccccccccccc@{}} 
\toprule
{Method} & \multicolumn{9}{c}{M$_5$} & \multicolumn{9}{c}{M$_6$} \\ 
\cmidrule(lr){2-10}\cmidrule(lr){11-19} 
{$\tau$} & 0.1 & 0.2 & 0.3 & 0.4 & 0.5 & 0.6 & 0.7 & 0.8 & 0.9 &  0.1 & 0.2 & 0.3 & 0.4 & 0.5 & 0.6 & 0.7 & 0.8 & 0.9 \\
\cmidrule(lr){2-10}\cmidrule(lr){11-19} 
RMSE & 15.420 & 10.924 & 9.573 & 8.806 & 8.556 & 8.827 & 10.133 & 12.199 & 15.366 & 11.861 & 9.255 & 8.042 & 7.458 & 6.316 & 7.593 & 8.485 & 9.955 & 12.706 \\
$R^2$ & 0.843 & 0.921 & 0.936 & 0.945 & 0.950 & 0.948 & 0.936 & 0.916 & 0.871 & 0.897 & 0.935 & 0.951 & 0.958 & 0.970 & 0.958 & 0.948 & 0.929 & 0.892 \\
\bottomrule
\end{tabular}
\end{scriptsize}
\end{center}
\end{table}
\vspace{-0.4cm}

\begin{figure}[!t]
\centering
\includegraphics[width=.31\textwidth]{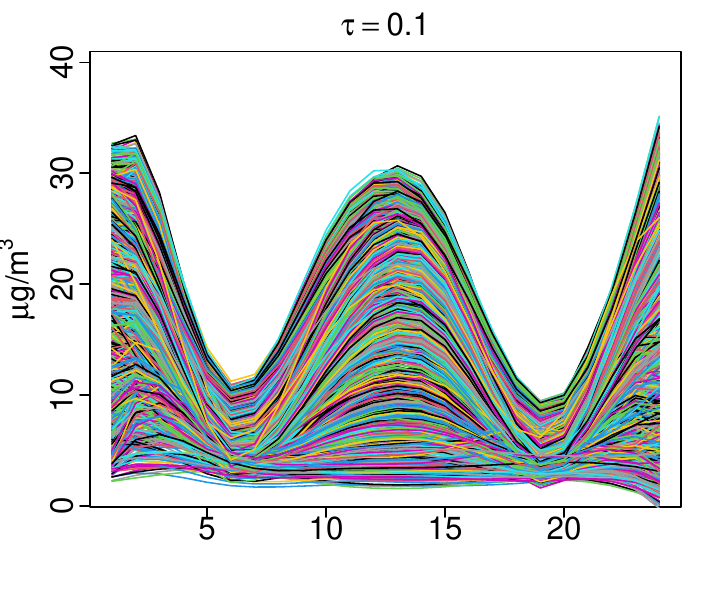}
\hfill
\includegraphics[width=.31\textwidth]{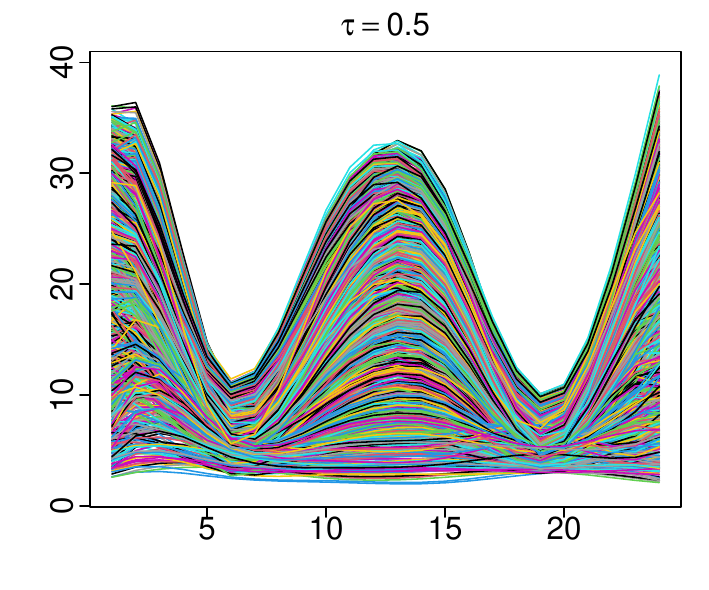}
\hfill
\includegraphics[width=.31\textwidth]{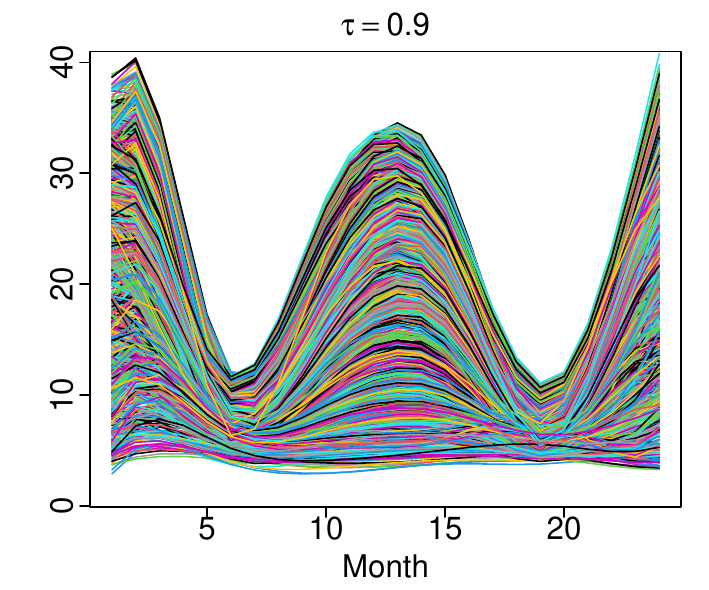}
\caption{Fitted conditional PM$_{2.5}$ quantile curves obtained
by the proposed method at representative quantile levels
$\tau=0.1$, $0.5$, and $0.9$, shown from left to right. The complete
display for $\tau\in\{0.1,0.2,\ldots,0.9\}$ is provided in Figure~S2
of the online supplementary material.}
\label{fig:Fig_4}
\end{figure}

The estimated coefficient surfaces provide additional insight. Figures~\ref{fig:Fig_5} and~\ref{fig:Fig_6} display the heatmap representations of the estimated $\beta(t,s)$ and $\rho(t,u)$, respectively. Figure~\ref{fig:Fig_5} shows that the spatial spline estimators, especially M$_2$ and M$_6$, recover smooth and localized regression surfaces for the PM$_{10}$-PM$_{2.5}$ relationship. The stronger color intensity near the diagonal region $s\approx t$ indicates that PM$_{10}$ at a given month contributes most strongly to PM$_{2.5}$ at nearby months, while off-diagonal effects are weaker and smoother. The proposed M$_6$ gives a structure broadly consistent with the spatial mean-based estimator M$_2$, while retaining the median quantile interpretation. Figure~\ref{fig:Fig_6} similarly shows that M$_2$ and M$_6$ produce smooth spatial spillover surfaces, whereas the FPCA-based M$_1$ is less stable. Overall, the heatmaps support the empirical conclusion that spatial modeling and penalized spline surface estimation are both important for interpretable recovery of the regression and spillover structures.

\begin{figure}[!htb]
\centering
\includegraphics[width=5.05cm]{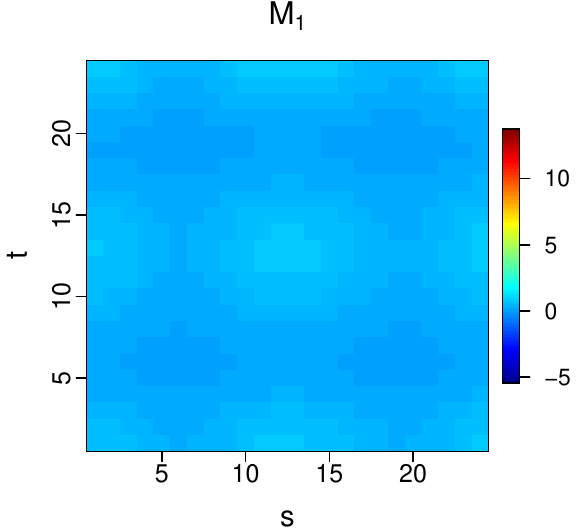}
\quad
\includegraphics[width=5.05cm]{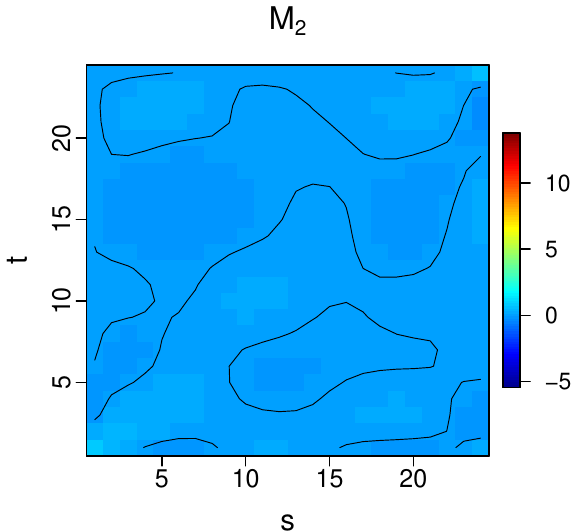}
\quad
\includegraphics[width=5.05cm]{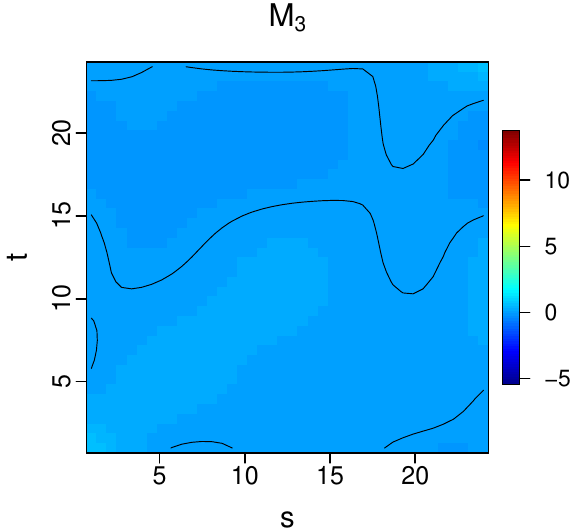}
\\
\includegraphics[width=5.05cm]{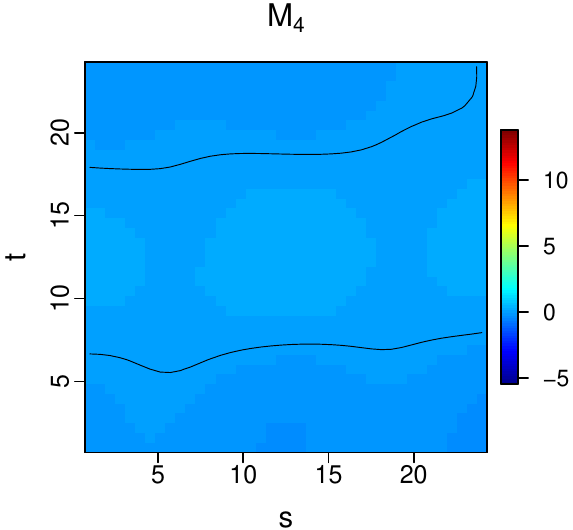}
\quad
\includegraphics[width=5.05cm]{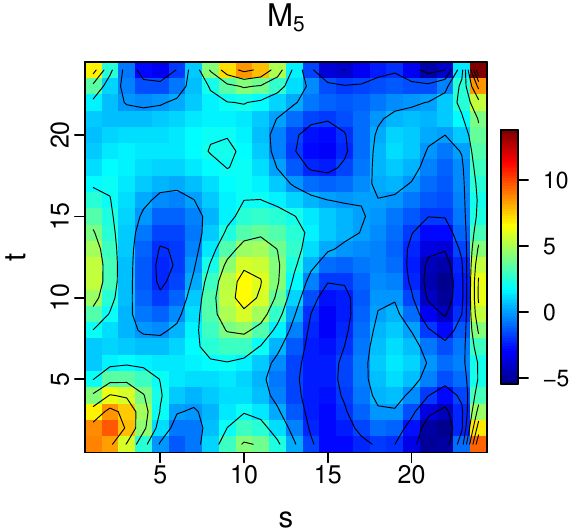}
\quad
\includegraphics[width=5.05cm]{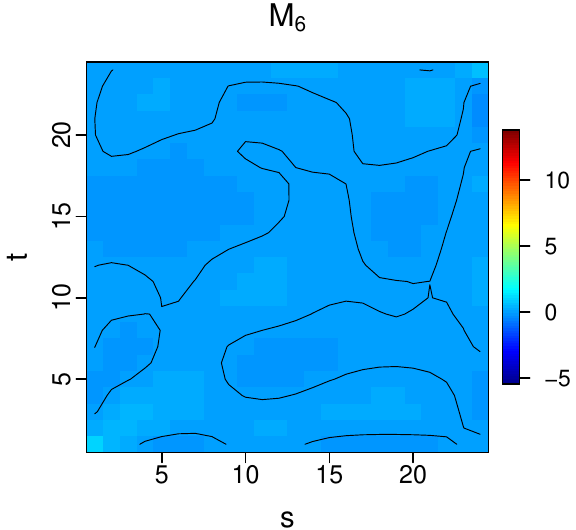}
\caption{Heatmap representation of the estimated regression coefficient surface $\beta(t,s)$ by all six methods: M$_1$, M$_2$, M$_3$, M$_4$, M$_5$, and M$_6$. The horizontal axis $s$ denotes the PM$_{10}$ predictor month, and the vertical axis $t$ denotes the PM$_{2.5}$ response month. The estimated regression coefficient surfaces for the quantile methods are obtained at $\tau=0.5$. The same color scale is used across all panels, allowing direct comparison of both magnitude and temporal structure.}\label{fig:Fig_5}
\end{figure}

\begin{figure}[!htb]
\centering
\includegraphics[width=5.05cm]{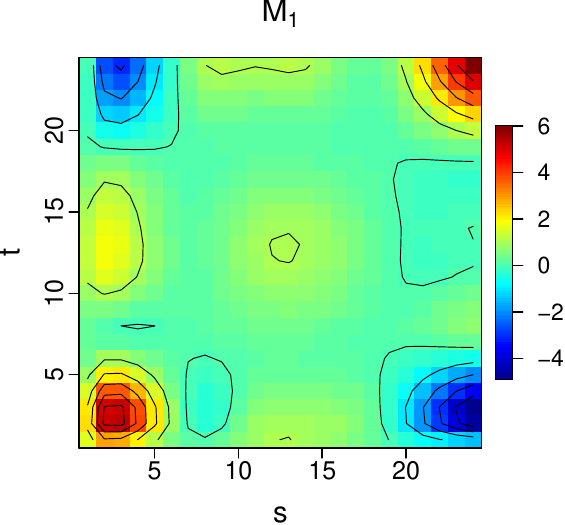}
\quad
\includegraphics[width=5.05cm]{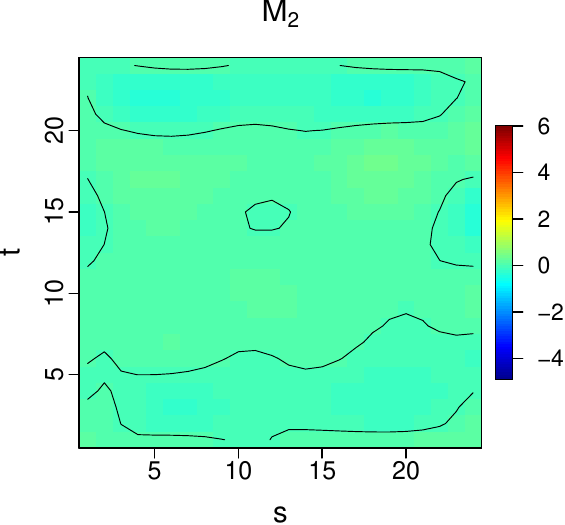}
\quad
\includegraphics[width=5.05cm]{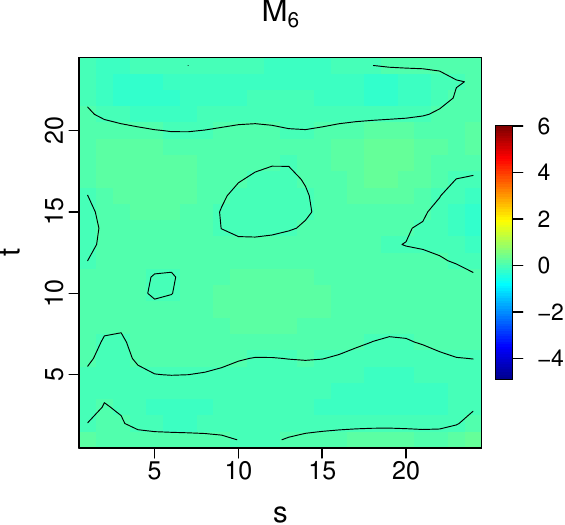}
\caption{Heatmap representation of the estimated spatial spillover surface $\rho(t,u)$ by M$_1$ (left panel), M$_2$ (middle panel), and M$_6$ (right panel). The horizontal axis $u$ denotes the neighboring-response month entering the spatial lag, and the vertical axis $t$ denotes the local PM$_{2.5}$ response month. The estimated spatial spillover surface for the quantile method is obtained at $\tau=0.5$. The same color scale is used across all panels, allowing direct comparison of the magnitude and temporal structure of the estimated spatial-dependence surfaces.}\label{fig:Fig_6}
\end{figure}

Although Figures~\ref{fig:Fig_5} and~\ref{fig:Fig_6} provide visual summaries of the estimated coefficient surfaces, uncertainty for the functional parameters can, in principle, be quantified directly from the asymptotic distribution in Theorem~\ref{th:1}. Let $\widehat{\bm\theta}_{\beta,\tau}=\mathrm{vec} (\widehat{\bm b}^{(\tau)})$ and $\widehat{\bm\theta}_{\rho,\tau}=\mathrm{vec}(\widehat{\bm\rho}^{(\tau)})$ denote the estimated spline coefficient vectors, and let $\widehat{\bm\Sigma}_{\beta,\tau}$ and $\widehat{\bm\Sigma}_{\rho,\tau}$ be the corresponding plug-in blocks of the asymptotic covariance $\bm\Sigma_{\theta,\tau}$ in Theorem~\ref{th:1}, whose explicit sandwich form is given in the supplementary proof. For fixed evaluation points, define $\bm a_\beta(t,s)=\bm\psi(s)\otimes\bm\phi(t)$ and $\bm a_\rho(t,u)=\bm\phi(u)\otimes\bm\phi(t)$, where $\bm\phi(\cdot)$ and $\bm\psi(\cdot)$ are the response- and predictor-domain B-spline basis vectors. Then $\widehat{\mathrm{se}}\{\widehat\beta_\tau(t,s)\} = \{n^{-1}\bm a_\beta(t,s)^\top \widehat{\bm\Sigma}_{\beta,\tau} \bm a_\beta(t,s)\}^{1/2}$, with an analogous expression for $\widehat\rho_\tau(t,u)$, where the factor $n^{-1}$ appears because $\bm\Sigma_{\theta,\tau}$ is the covariance of the $\sqrt n$-scaled estimator. Pointwise normal intervals follow directly, while simultaneous bands may be constructed using a multiplier or spatially structured bootstrap. The online supplementary material reports a finite-sample coverage assessment for both coefficient surfaces; since plug-in standard-error surfaces are not yet part of the accompanying software implementation, that assessment estimates the pointwise sampling variability from the Monte Carlo replications.

\section{Conclusion}\label{sec:6}

We proposed a SFoF-QR framework for spatially indexed functional data, aimed at modeling the full conditional distribution of a functional response rather than only its conditional mean. The model extends functional SAR ideas to quantiles by allowing the conditional $\tau$-th quantile of $\Y_i$ to depend on both a functional predictor $\X_i$ and a spatially lagged functional response $\widetilde{\Y}_i$ through two bivariate coefficient surfaces, $\beta_\tau$ and $\rho_\tau$. Estimation combines tensor-product B-spline representations with quadratic roughness penalties and a two-stage instrumental-variable strategy to address the endogeneity induced by the spatial lag. This construction avoids FPCA truncation, accommodates non-separable surfaces, and remains feasible for dense functional grids and large spatial samples.

Across both simulation and empirical studies, two conclusions are consistent. First, explicitly modeling spatial dependence is essential: methods that ignore the functional SAR component deteriorate rapidly as spatial dependence strengthens. Second, penalized spline surface estimation can be substantially more reliable than FPCA-based pipelines for recovering complex bivariate regression and spatial-kernel surfaces, especially when tail behavior is of interest and variance-explained truncation discards regression-relevant structure. In the air-quality application, SFoF-QR yields competitive median predictions and improves materially across quantiles, while also revealing pronounced state-dependent volatility and quantile-varying covariate and spillover effects that are not accessible under mean regression. A limitation of the real-data application is that an independent external test sample was not available; consequently, the empirical comparison in the application is based on in-sample fitted performance, while out-of-sample predictive behavior is evaluated through the simulation study.

Several limitations suggest clear directions for future work. The approach assumes a fixed prespecified weight matrix $\bm W$, so sensitivity to neighborhood choice, scale, and normalization warrants further investigation and possibly data-adaptive estimation of $\bm W$. Computation can become demanding as basis dimension, grid resolution, and the number of fitted quantiles increase; scalable optimization and warm-start strategies would broaden applicability. Since estimation is performed separately for each $\tau$, quantile crossing may occur in finite samples. However, no crossing was observed on the evaluation grids in any of the simulation settings or in the air-quality application; in particular, $\widehat Q_{0.025,i}(t_r)\leq\widehat Q_{0.975,i}(t_r)$ for all evaluated curves and grid points, so no rearrangement was required when constructing the 95\% prediction intervals. Joint estimation with noncrossing constraints and information sharing across $\tau$ remains a natural extension. Although quantile regression provides robustness to outcome contamination, leverage and spatial propagation through $\widetilde{\Y}_i$ can still affect both stages; developing a fully robust SFoF-QR procedure with robust first-stage IV/quantile components remains an important goal. Finally, the current model is linear with smooth coefficient surfaces; extensions to multiple functional/scalar predictors with structured penalties, to alternative dependence structures (e.g., spatial Durbin or spatial error analogues), and to spatiotemporal dynamics would further enhance flexibility. On the inferential side, building formal testing tools for features of $\beta_\tau$ and $\rho_\tau$ (e.g., separability, spatial neutrality, or regime changes across quantiles) would complement the asymptotic distribution theory developed here.

\clearpage

\setcounter{section}{0}
\setcounter{subsection}{0}
\setcounter{subsubsection}{0}
\setcounter{equation}{0}
\setcounter{table}{0}
\setcounter{figure}{0}
\setcounter{algocf}{0}
\setcounter{assumption}{0}

\renewcommand{\thesection}{S\arabic{section}}
\renewcommand{\thesubsection}{\thesection.\arabic{subsection}}
\renewcommand{\thesubsubsection}{\thesubsection.\arabic{subsubsection}}
\renewcommand{\thetable}{S\arabic{table}}
\renewcommand{\thefigure}{S\arabic{figure}}
\renewcommand{\thealgocf}{S\arabic{algocf}}
\numberwithin{equation}{section}

\renewcommand{\theHsection}{supp.\arabic{section}}
\renewcommand{\theHsubsection}{supp.\arabic{section}.\arabic{subsection}}
\renewcommand{\theHsubsubsection}{supp.\arabic{section}.\arabic{subsection}.\arabic{subsubsection}}
\renewcommand{\theHequation}{supp.\arabic{section}.\arabic{equation}}
\renewcommand{\theHtable}{supp.\arabic{table}}
\renewcommand{\theHfigure}{supp.\arabic{figure}}
\makeatletter
\@ifundefined{theHalgocf}
  {\newcommand{\theHalgocf}{supp.\arabic{algocf}}}
  {\renewcommand{\theHalgocf}{supp.\arabic{algocf}}}
\@ifundefined{theHassumption}
  {\newcommand{\theHassumption}{supp.\arabic{assumption}}}
  {\renewcommand{\theHassumption}{supp.\arabic{assumption}}}
\makeatother

\captionsetup[figure]{style=italic,format=hang,singlelinecheck=true}
\captionsetup[table]{style=italic,format=hang,singlelinecheck=true}

\begin{center}
{\Large\bfseries Supplementary material on\\[0.4em]
``Spatial function-on-function quantile regression''\par}
\end{center}
\vspace{1.5em}

\spacingset{1.5} 

This supplementary material provides the technical proofs, computational details, comprehensive numerical evidence, and additional graphical results supporting the main manuscript. First, we discuss instrument selection, identification, and implementation of the two-stage instrumental-variable quantile estimator. Second, we provide the formal proof of Theorem~3.1 and the complete Monte Carlo results across the considered sample sizes, spatial-dependence strengths, and error distributions. Finally, we present additional graphical results from the air-quality application, including fitted-curve comparisons, the complete quantile-specific fitted trajectories, and the full comparison of the estimated regression coefficient surfaces.

\section{Identification in the two-stage IV quantile formulation}\label{subsec:ivqr-identification}

This subsection clarifies the identification conditions underlying the two-stage estimator. After projection onto the tensor-product B-spline spaces, the SFoF-QR model becomes a finite-dimensional IV quantile problem. Let $\bm d_{ir}$ denote the structural Stage~2 design row corresponding to location $i$ and response grid point $t_r$, where $\bm d_{ir}$ contains both the endogenous lag-score block and the exogenous predictor-score block. Let $\bm Z_i=\{\bm\gamma_i^{(0)\top},\ldots,\bm\gamma_i^{(H-1)\top}\}^\top$ be the instrument-score vector generated from $\X_i,\bm W\X_i,\ldots,\bm W^{H-1}\X_i$. For the pseudo-true spline coefficient vector $\bm\theta_{\tau,K}^{0}$ of Assumption~\ref{ass:A2}(i), define the structural quantile error $e_{ir,\tau}(\bm\theta_{\tau,K}^{0}) = \Y_i(t_r)-\bm d_{ir}^\top\bm\theta_{\tau,K}^{0}$. The identifying restriction is the structural IV quantile condition $\Pr\{e_{ir,\tau}(\bm\theta_{\tau, K}^0)\le0\mid\bm Z_i\}=\tau$, $r=1,\ldots,R$, or, equivalently, $\E\left[\bm Z_i\{\tau-\mathbf{1}(e_{ir,\tau}(\bm\theta_{\tau, K}^0)<0)\} \right]=\bm0$. This condition is stronger than requiring the instruments to be correlated with the spatial lag: it requires the instruments to be exogenous for the structural quantile error.

The relevance condition requires that the instruments explain non-negligible variation in the endogenous lag-score vector. Equivalently, the conditional quantile projection of the lag-score block on $\bm Z_i$ must be nondegenerate. The rank condition requires the IV-projected structural design to have full column rank. In matrix form, a sufficient local identification condition is $\lambda_{\min}\left[\E\{f_{e,\tau}(0\mid\bm Z_i) \widehat{\bm d}_{ir}\widehat{\bm d}_{ir}^{\top}\} \right]>0$, where $\widehat{\bm d}_{ir}$ is the Stage~2 design row after replacing the endogenous lag-score block by its first-stage fitted value, and $f_{e,\tau}(0\mid\bm Z_i)$ is the conditional density of the structural quantile error at zero. This condition is the quantile analogue of the full-rank design condition in linear IV regression, with the conditional density weighting arising from the local expansion of the quantile score.

The first-stage quantile level is chosen to match the second-stage quantile level because the structural quantile restriction is itself indexed by $\tau$. Thus, for each target $\tau$, the first stage estimates the $\tau$-specific reduced-form component of the endogenous lag scores that is relevant for the $\tau$-specific structural moment equation. A first stage at a different quantile level would generally target a different reduced-form feature of the endogenous lag process and would not correspond to the structural moment condition at the target quantile. Therefore, the proposed algorithm is run separately for each $\tau$.

Finally, the first-stage fitted lag is a generated regressor. The asymptotic distribution therefore contains not only the second-stage quantile score but also the first-stage estimation error. In the proof of Theorem~3.1, this effect appears through the local perturbation term $\widehat{\triangle}_0 = (a-1)\sqrt n(\widehat{\bm\Xi}-\bm\Xi) + \sqrt n(\widehat{\bm\Theta}-\bm\Theta)\widetilde{\bm\rho}^{(\tau)}$, where $\widehat{\bm\Xi}$ and $\widehat{\bm\Theta}$ are the first-stage reduced-form estimators. Substituting the asymptotic linear representations of these first-stage estimators into the second-stage Bahadur expansion yields the final limiting distribution. Hence the covariance kernel in Theorem~3.1 incorporates the generated-regressor uncertainty from the first stage, rather than treating the fitted spatial lag as fixed.

\section{Instrument selection for the spatial lag}

The spatially lagged response term induces endogeneity in $\int_0^1 \widetilde{\Y}_i(u) \rho_\tau(t,u) du$. Consistent estimation therefore hinges on instruments that are (i) relevant for the endogenous regressors that enter the Stage~2 design and (ii) exogenous with respect to the structural quantile restrictions. In our SFoF-QR, the endogenous part enters via the lag-score vector $\widehat{\widetilde{\bm{\phi}}}_i$ (the B-spline coefficients of $\widetilde{\Y}_i(\cdot)$ used in the block $\widehat{\widetilde{\bm{\phi}}}_i^\top \otimes \bm{\phi}(t)^\top$). We construct instruments by propagating exogenous functional predictors through the spatial graph and the same spline machinery used in Stage~2.

While increasing $H$ enlarges the instrument set, higher-order lags typically suffer attenuating correlation with the endogenous lag scores as $h$ grows (heuristically governed by the spectral radius of $\bm{W}$). In practice, we recommend a small inventory $H = 1$ or $H \in \{1,2\}$, unless the empirical first-stage fit clearly benefits from $h > 2$. Relevance improves when instruments mirror the structure of the endogenous block. Concretely, Stage~1 regresses the vectorized lag $\mathrm{vec} \big(\widetilde{Y}~ \mathrm{E}_Y \big)$ on $\widetilde{\bm{\psi}}^{(0)} \otimes \mathrm{E}_Y, \widetilde{\bm{\psi}}^{(1)} \otimes \mathrm{E}_Y, \ldots$, where $\mathrm{E}_Y$ is the $t$-basis evaluation matrix. This alignment strengthens the link to the ultimate endogenous regressor
$\widehat{\widetilde{\bm{\phi}}}_i^\top \otimes \bm{\phi}(t)^\top$.

In the numerical studies, we fixed the instrument set to $\{\X_i,(\bm W\X)_i,(\bm W^2\X)_i\}$, corresponding to spatial lags up to order two.

\section{Computational details for the smoothed Stage~2 optimization}\label{subsec:stage2-optimization}

Let $\bm y=\operatorname{vec}(\Y)$ denote the response vector stacked by location and response-grid point, and let $\widehat{\bm\Pi}$ denote the Stage~2 design matrix constructed from the functional-intercept block, the local predictor-score block, and the instrumented spatial-lag-score block. In accordance with the numerical implementation, the coefficient vector is ordered as $\bm\theta = \left\{\bm b_0^\top, \operatorname{vec}\bigl(\bm b^{(\tau)}\bigr)^\top, \operatorname{vec}\bigl(\bm\rho^{(\tau)}\bigr)^\top \right\}^{\top}$. 

For fixed tuning parameters $(\lambda_\beta,\lambda_\rho)$, define the computational penalty matrix by
\begin{equation}\label{eq:supp-computational-penalty}
\bm{\mathcal P}_{\lambda_\beta,\lambda_\rho,\varepsilon} = \operatorname{blockdiag} \left\{\bm 0_{K_0\times K_0}, \lambda_\beta\bm R_\beta, \lambda_\rho\bm R_\rho \right\} + \varepsilon\bm I,
\end{equation}
where $\bm R_\beta$ and $\bm R_\rho$ are the tensor-product second-derivative roughness-penalty matrices for the regression and spatial coefficient surfaces, respectively. The leading block is unpenalized by the roughness penalty, whereas the small ridge term $\varepsilon\bm I$ is added to the complete penalty matrix solely for numerical stabilization. In all numerical studies, we used $\varepsilon=10^{-8}$.

Let $\widehat{\bm\pi}_a^\top$ denote row $a$ of $\widehat{\bm\Pi}$, and define $u_a(\bm\theta) = y_a-\widehat{\bm\pi}_a^\top\bm\theta$, $a=1,\ldots,N$ and $N=nR$. The computational Stage~2 objective is
\begin{equation} \label{eq:supp-lbfgsb-objective}
Q_{\tau,\alpha}(\bm\theta) = \frac{1}{n}\sum_{a=1}^{N} \ell_{\tau,\alpha}\bigl\{u_a(\bm\theta)\bigr\} + \frac{1}{2} \bm\theta^\top \bm{\mathcal P}_{\lambda_\beta,\lambda_\rho,\varepsilon} \bm\theta,
\end{equation}
where the differentiable approximation to the check loss is $\ell_{\tau,\alpha}(u) = \tau u + \alpha\log \left\{1+\exp\left(-\frac{u}{\alpha}\right) \right\}$, $\alpha>0$. Its derivative with respect to the residual is $\psi_{\tau,\alpha}(u) = \ell_{\tau,\alpha}'(u) = \tau - \left\{1+\exp\left(\frac{u}{\alpha}\right) \right\}^{-1}$. Consequently, the analytic gradient supplied to the optimizer is
\begin{equation}\label{eq:supp-lbfgsb-gradient}
\nabla_{\bm\theta}Q_{\tau,\alpha}(\bm\theta) = -\frac{1}{n}\widehat{\bm\Pi}^{\top} \bm\psi_{\tau,\alpha} \bigl\{\bm u(\bm\theta)\bigr\} + \bm{\mathcal P}_{\lambda_\beta,\lambda_\rho,\varepsilon} \bm\theta,
\end{equation}
where $\bm u(\bm\theta) = \left\{u_1(\bm\theta),\ldots,u_N(\bm\theta) \right\}^{\top}$ and $\bm\psi_{\tau,\alpha} \bigl\{\bm u(\bm\theta)\bigr\} = \left[ \psi_{\tau,\alpha}\bigl\{u_1(\bm\theta)\bigr\}, \ldots, \psi_{\tau,\alpha}\bigl\{u_N(\bm\theta)\bigr\} \right]^{\top}$.

The coefficient vector is initialized at $\bm\theta^{(0)}=\bm 0$. The objective function in \eqref{eq:supp-lbfgsb-objective} and the analytic gradient in \eqref{eq:supp-lbfgsb-gradient} are passed to the L-BFGS-B implementation in the \texttt{optim} function in \textsf{R}. In the reported numerical studies, the control settings were \texttt{maxit}$=8000$, \texttt{factr}$=10^7$, and \texttt{trace}$=0$. No lower or upper coefficient bounds were imposed. The optimization terminates according to the internal L-BFGS-B convergence criterion or when the maximum permitted number of iterations is reached.

For smoothing-parameter selection, candidate values of $(\lambda_\beta,\lambda_\rho)$ are considered over a Cartesian grid. For each candidate pair, the model is fitted by minimizing \eqref{eq:supp-lbfgsb-objective}. The fitted model is then evaluated using
\begin{equation}\label{eq:supp-bic}
\operatorname{BIC}(\lambda_\beta,\lambda_\rho) = \log \left[ \frac{1}{N} \sum_{a=1}^{N} \ell_\tau(\widehat u_a) \right] + \frac{\log N}{N} \left(K_0+K_yK_x+K_y^2 \right),
\end{equation}
where $\ell_\tau(u) = u\left\{\tau-\mathbf{1}(u<0) \right\}$ is the exact check loss and $\widehat u_a$ is the fitted residual corresponding to observation $a$. Thus, the differentiable approximation is used to facilitate numerical optimization, whereas the tuning criterion is evaluated using the exact quantile loss. The candidate pair minimizing \eqref{eq:supp-bic} is selected and used in the final model fit.

\begin{algorithm}[!ht]
\caption{L-BFGS-B optimization for the smoothed penalized Stage~2 IV quantile criterion}\label{alg:lbfgsb}
\DontPrintSemicolon

\KwIn{Response vector $\bm y$; IV-projected design matrix $\widehat{\bm\Pi}$; penalty matrix $\bm{\mathcal P}_{\lambda_\beta,\lambda_\rho,\varepsilon}$;
quantile level $\tau$; smoothing constant $\alpha$; L-BFGS-B control parameters.}

\KwOut{Estimated coefficient vector $\widehat{\bm\theta}$.}

Set the initial coefficient vector to $\bm\theta^{(0)}=\bm 0$.\;

For a candidate coefficient vector $\bm\theta$, compute the residuals $u_a(\bm\theta)=y_a-\widehat{\bm\pi}_a^\top\bm\theta$, $a=1,\ldots,N$.\;

Evaluate the smoothed penalized objective $Q_{\tau,\alpha}(\bm\theta)$ in \eqref{eq:supp-lbfgsb-objective}.\;

Evaluate the analytic gradient $\nabla_{\bm\theta}Q_{\tau,\alpha}(\bm\theta)$ in \eqref{eq:supp-lbfgsb-gradient}.\;

Apply the L-BFGS-B algorithm using the initial value $\bm\theta^{(0)}$, the objective function, and the analytic gradient.\;

Set the optimizer solution equal to $\widehat{\bm\theta}$.\;

Extract the coefficient blocks corresponding to $\widehat{\bm b}_0$, $\widehat{\bm b}^{(\tau)}$, and $\widehat{\bm\rho}^{(\tau)}$.\;

Reconstruct $\widehat\beta_\tau(t,s)$ and $\widehat\rho_\tau(t,u)$ using their tensor-product B-spline representations.\;

Return $\widehat{\bm\theta}$.\;
\end{algorithm}

\section{Theoretical proofs for Theorem 3.1}

We now establish the fundamental large-$n$ behavior of the estimators. The required assumptions for the consistency and asymptotic normality of the estimators are collected below.

\begin{assumption}\label{ass:A1}
Let $\bm W=(w_{ij})_{1\le i,j\le n}$ denote the row-normalized spatial weight matrix. Assume:
\begin{enumerate}[(i)]
\item $w_{ii}=0$ for all $i$.
\item $\sum_{j=1}^n w_{ij}=1$ for every $i$.
\item There exists $a_w<\infty$ such that 
      $\|\bm W\|_{\infty}\le a_w$.
\item Assume that $(\mathcal R_{\rho_\tau,\bm W}\bm f)_i(t) = \sum_{j=1}^n w_{ij} \int_{\mathcal I_Y} f_j(u)\rho_\tau(t,u)\,du$ defines a bounded linear operator on $(\mathcal{L}^p)^n[0,1]$, and that $\|\rho_\tau\|_{\infty} < \|\bm W\|_{\infty}^{-1}$. Since $\|\mathcal R_{\rho_\tau,\bm W}\|_{\mathrm{op}}\le \|\bm{W}\|_\infty\|\rho_\tau\|_\infty<1$, the operator $\mathbb I-\mathcal R_{\rho_\tau,\bm W}$ is bijective and admits the bounded inverse $(\mathbb I-\mathcal R_{\rho_\tau,\bm W})^{-1} =\sum_{r=0}^{\infty}\mathcal R_{\rho_\tau,\bm W}^r$, where the series converges in operator norm.   
\end{enumerate}
\end{assumption}

\begin{assumption}\label{ass:A2}
Let $p^*\ge 1$ be the degree of the B-splines and $q^*=2$ the order of derivative penalization. The spline dimensions $K_y$ and $K_x$ are fixed and do not depend on $n$. 
\begin{enumerate}[(i)]
\item There exists a unique $\bm\theta_{\tau,K}^{0}\in\mathbb R^{K_y^2+K_yK_x}$ (augmented by the intercept block) satisfying the population IV quantile moment condition of the working spline model,
\begin{equation*}
\E\bigl[\bm\Gamma_{i,r} \varphi_\tau\{\Y_i(t_r)-\bm d_{ir}^\top\bm\theta_{\tau,K}^{0}\}\bigr] =\bm 0,\qquad r=1,\ldots,R .
\end{equation*}
The induced sieve targets are $\rho_{\tau,K}(t,u)=\{\bm\phi^\top(u)\otimes\bm\phi^\top(t)\} \widetilde{\bm\rho}^{(\tau),0}$ and $\beta_{\tau,K}(t,s)=\{\bm\psi^\top(s)\otimes\bm\phi^\top(t)\} \widetilde{\bm b}^{(\tau),0}$. The approximation errors $a_K=\|\rho_\tau-\rho_{\tau,K}\|_\infty+\|\beta_\tau-\beta_{\tau,K}\|_\infty$ are finite but are not required to vanish.
\item The block penalty matrix $\bm{\mathcal P}_{\mathrm{surf},\lambda_\rho,\lambda_\beta}$ is constructed from the Gram matrices of the integrated squared $q^*$th B-spline derivatives, as in Section~3 of the main text.
\end{enumerate}
\end{assumption}

\begin{assumption}\label{ass:A3}
As $n\to\infty$, the smoothing parameters satisfy $\lambda_\rho\vee\lambda_\beta=o(n^{-1/2})$; consequently $\sqrt n\,\bm{\mathcal P}_{\mathrm{surf},\lambda_\rho,\lambda_\beta}\, \bm\theta=o(1)$ uniformly over compact sets of $\bm\theta$. In addition, the smoothing constant of the differentiable check-loss approximation satisfies $\alpha_n\downarrow 0$ with $\sqrt n\,\alpha_n\to 0$.
\end{assumption}

\begin{assumption}\label{ass:A4}
Let $\bm{\gamma}_{i}^{(h)} \in \mathbb{R}^{K_x}$ be the B-spline score vector of the $h$th instrument function $\X_i^{(h)}(s) = (\bm W^h \X)_i(s)$, $h = 0, 1, \ldots, H-1$, as defined in Section 3. Define the stacked instrument-score vector and matrix by
\begin{equation*}
\bm \Gamma_i = (\gamma_i^{{(0)}\top}, \gamma_i^{{(1)}\top}, \ldots, \gamma_i^{{(H-1)}\top})^\top, \qquad \bm \Gamma = (\bm \Gamma_1, \ldots, \bm \Gamma_n)^\top.
\end{equation*}
Consider the Stage-1 reduced form for the lag score and the implied reduced form for the vectorized response:
\begin{equation*}
Q_\tau \{\widetilde{\Y} \vert \bm \Gamma\} = \bm \Gamma \bm \Theta, \qquad Q_\tau \{{\rm vec}(\Y) \vert \bm \Gamma \} = \bm \Gamma \bm \Xi,
\end{equation*}
where $\bm \Theta$ and $\bm \Xi$ are the corresponding reduced-form coefficient matrices/vectors at the quantile level $\tau$. Define the reduced-form error terms
\begin{equation*}
\bm \nu = {\rm vec}(\Y) - \bm \Gamma \bm \Xi, \qquad \bm \zeta = {\rm vec}(\widetilde{\Y}) - \bm \Gamma \bm \Theta.
\end{equation*}
and let $\bm \nu_i, \bm \zeta_i, \bm \Gamma_i$ denote their $i$th block/rows. Then, the sequence of triplets $\{(\bm \nu_i, \bm \zeta_i, \bm \Gamma_i) \}_{i=1}^n$ consists of i.i.d. random elements.
\end{assumption}

\begin{assumption}\label{ass:A5}
\begin{enumerate}[(i)]
\item The third moment of the instrument-score vector is finite, $\mathbb{E}(\Vert \bm \Gamma_i \Vert^3) < \infty$, $\forall i$.
\item Let $\bm \theta^{(\tau)}$ be the Stage-2 structural parameter vector collecting the tensor product B-spline coefficients of $\rho_\tau (t,u)$ and $\beta_\tau (t,s)$, and let $\bm A (\bm \Theta)$ denote the linear mapping induced by substituting the Stage-1 reduced form $Q_\tau \{\widetilde{\Y} \vert \bm \Gamma \} = \bm \Gamma \bm \Theta$ into the finite-dimensional structural SFoF-QR representation. Assume that $\bm A (\bm \Theta )$ has full column rank, ensuring identifiability of $\bm \theta^{(\tau)}$.
\item Let $g_\tau^*(\cdot \mid \bm \Gamma_i)$ and $h_\tau^*(\cdot \mid \bm \Gamma_i)$ be the conditional density functions of $\nu_{i,r}$ and $\zeta_{i,r}$, respectively. These conditional densities are Lipschitz continuous in their arguments for all $\bm \Gamma_i$. In addition, the matrices $\bm Q_{0(\nu)} = \E \bigl\{ g_\tau^*(0 \mid \bm \Gamma_i ) \bm \Gamma_i \bm \Gamma_i^\top \bigr\}$ and $\bm Q_{0(\zeta)} = \E \bigl\{ h_\tau^*(0 \mid \bm \Gamma_i ) \bm \Gamma_i \bm \Gamma_i^\top \bigr\}$ are finite and positive definite.
\item $\E \{\varphi_\tau (\nu_{i,r}) \mid \bm \Gamma_i \} = \E \{\varphi_\tau (\zeta_{i,r}) \mid \bm \Gamma_i \} = 0$ for all $r=1,\ldots,R$.
\end{enumerate}
\end{assumption}

Assumption~\ref{ass:A1} imposes the standard regularity conditions required for a well-defined SAR structure. In particular, it guarantees identifiability of the spatial lag component and rules out degeneracies such as singular or explosive behavior of the associated spatial operator, thereby ensuring a stable and well-posed estimation problem \citep[cf.][]{Kelejian1998, Lee2004, BSGARC2025}. Assumption~\ref{ass:A2} and Assumption~\ref{ass:A3} specify the working spline spaces, the pseudo-true centering parameter, and the negligibility conditions on the penalty and the loss-smoothing constant. These conditions justify representing the underlying functions through suitable basis expansions and incorporating roughness penalties, which together provide regularization, mitigate overfitting, and enhance numerical stability of the resulting estimators. The spline-approximation background relevant for the sieve interpretation in Remark~\ref{rem:sieve-interp} is classical: for a univariate $f\in\mathcal C^{p^*+1}$ there exist spline coefficients such that $\sup_t\vert f(t)-\sum_{k=1}^K f_k \upsilon_k(t)\vert =\mathcal O(K^{-(p^*+1)})$ \citep{Claeskens2009,deBoor2001}, and analogous bounds hold for tensor-product spline systems \citep{Mobler2009,Belloni2015,Lyche2018}. These rates are used only in Remark~\ref{rem:sieve-interp}; they play no role in the fixed-dimension Theorem~3.1. Finally, Assumption~\ref{ass:A4} and Assumption~\ref{ass:A5} provide the additional regularity required to derive consistency and asymptotic normality of the estimators constructed in the finite-dimensional approximation space using the two-stage procedure of \cite{Kim_Muller2004}.

\begin{remark}\label{rem:sieve-interp}
If the spline dimensions are increased along an auxiliary sequence with $K_y\to\infty$ and $K_x\to\infty$, classical Jackson-type inequalities for tensor-product B-splines imply $a_K = \mathcal O\left\{ K_y^{-(p^*+1)} + K_x^{-(p^*+1)} \right\} \longrightarrow0$. This is an approximation-theoretic statement only. Replacing the fixed-spline centering functions $\rho_{\tau,K}$ and $\beta_{\tau,K}$ by the unrestricted true surfaces in a stochastic limit would require a separate growing-dimensional theory, including uniform Bahadur expansions and control of the increasing parameter dimension. No such result is claimed here.
\end{remark}

\begin{proof}[Proof of Theorem 3.1]
Throughout, $K_y$, $K_x$, and $H$ are fixed, so the dimension $D$ of the instrument-score design vector $\bm\Gamma_{i,r}$ (and of all parameter blocks below, including the unpenalized intercept block) is fixed and does not depend on $n$. The estimator is centered at the pseudo-true working-model parameter $\bm\theta_{\tau,K}^{0}$ of Assumption~\ref{ass:A2}(i), and the reconstructed surfaces are centered at the induced sieve targets $\rho_{\tau,K}$ and $\beta_{\tau,K}$. No claim is made about the distance between the sieve targets and the true surfaces; that distance is a fixed deterministic quantity under Assumption~\ref{ass:A2}(i) (cf.\ Remark~\ref{rem:sieve-interp}).

We now work with the spline-projected model. For each $i\in\{1,\dots,n\}$ and each grid index $r\in\{1,\dots,R\}$, let $\nu_{i,r}$ denote the reduced-form residual as defined in Assumption~\ref{ass:A4}, and let $\bm{\Gamma}_{i,r}$ be the corresponding $D\times 1$ design vector stacking the spatial and non-spatial components and instruments. The quantile score at level $\tau$ is $\varphi_\tau(c) = \tau - \mathbf{1}(c<0)$. For $e\in\mathbb{R}^D$, define
\begin{equation*}
m^*(\omega_{i,r},e) = \bm{\Gamma}_{i,r} \varphi_\tau\bigl(a \nu_{i,r} - \bm{\Gamma}_{i,r}^\top e\bigr), \qquad \omega_{i,r} = (\nu_{i,r}, \bm{\Gamma}_{i,r}^\top)^\top,
\end{equation*}
with $a>0$ a fixed constant as in \cite{Kim_Muller2004}. Set
\begin{equation}\label{eq:def-Mstar}
\mathcal{M}^*(e) = n^{-1/2}\sum_{i=1}^n\sum_{r=1}^R m^*(\omega_{i,r},e), \qquad V^*(e) = \mathcal{M}^*(e) - \mathbb{E}\{\mathcal{M}^*(e)\}.
\end{equation}

We recall the following technical conditions from \cite{Andrews1994}, adapted to our notation.
\begin{description}
\item[C$_1$] There exists an envelope function $\widetilde{\mathcal{M}}(\omega_{i,r})$ such that the class $\mathcal{F} = \bigl\{ m^*(\cdot,e) : e\in\mathbb{R}^D\bigr\}$ satisfies Pollard's entropy condition relative to $\widetilde{\mathcal{M}}$, i.e. its uniform covering numbers obey the polynomial bound required in Theorem~2 of \cite{Andrews1994}.
\item[C$_2$] There exists $\delta>2$ such that the uniform moment condition
\begin{equation*}
\sup_{n} \frac{1}{n}\sum_{i=1}^n\sum_{r=1}^R
\mathbb{E}\left\{\widetilde{\mathcal{M}}(\omega_{i,r})^\delta\right\}
< \infty
\end{equation*}
holds.
\end{description}

\begin{lemma}\label{lem:equicontinuity}
Under Assumptions~\ref{ass:A4} and \ref{ass:A5}, conditions \textup{C$_1$} and \textup{C$_2$} hold with envelope $\widetilde{\mathcal{M}}(\omega_{i,r}) := \max(1,\|\bm{\Gamma}_{i,r}\|)$, and, consequently, for any fixed $N^*>0$,
\begin{equation}\label{eq:A1f-new}
\sup_{\|e_1-e_2\|\le N^*}\bigl\|V^*(e_1)-V^*(e_2)\bigr\| = o_p(1),
\end{equation}
where $N^* = n^{-1/2}N$ and $\|\cdot\|$ denotes the Euclidean norm on $\mathbb{R}^D$.
\end{lemma}

\begin{proof}[Proof of Lemma~\ref{lem:equicontinuity}]
We first verify \textup{C$_1$} and \textup{C$_2$}, and then invoke Theorem~2 of \cite{Andrews1994}. Write $m^*(\omega_{i,r},e) = f_1(\omega_{i,r},e) f_2(\omega_{i,r},e)$ with
\begin{equation*}
f_1(\omega_{i,r},e) = \bm{\Gamma}_{i,r}, \qquad f_2(\omega_{i,r},e) = \varphi_\tau\bigl(a \nu_{i,r} - \bm{\Gamma}_{i,r}^\top e\bigr).
\end{equation*}
For each $e$, $f_2(\cdot,e)$ is a Type~I class in the sense of \cite{Andrews1994}, with envelope identically equal to $1$, and $f_1(\cdot,e)$ is Type~I with envelope $\|\bm{\Gamma}_{i,r}\|$. By closure of Type~I classes under finite products (see Theorems~2 and 3 in \cite{Andrews1994}), the class $\mathcal{F}=\{m^*(\cdot,e): e\in\mathbb{R}^D\}$ is again Type~I with envelope $\widetilde{\mathcal{M}}(\omega_{i,r}) = \max(1,\|\bm{\Gamma}_{i,r}\|)$. For Type~I classes, Pollard's entropy condition is automatically satisfied; thus \textup{C$_1$} holds.

By Assumption~\ref{ass:A5}(i), there exists $\delta>2$ such that
\begin{equation*}
\sup_{n} \frac{1}{n}\sum_{i=1}^n\sum_{r=1}^R \mathbb{E}(\|\bm{\Gamma}_{i,r}\|^\delta) < \infty.
\end{equation*}
Since $\max(1,\|\bm{\Gamma}_{i,r}\|)^\delta \le C\bigl(1+\|\bm{\Gamma}_{i,r}\|^\delta\bigr)$ for some constant $C>0$, it follows that
\begin{equation*}
\sup_{n} \frac{1}{n}\sum_{i=1}^n\sum_{r=1}^R \mathbb{E}\left\{\widetilde{\mathcal{M}}(\omega_{i,r})^\delta\right\}
< \infty.
\end{equation*}
Thus, \textup{C$_2$} holds.

With \textup{C$_1$} and \textup{C$_2$} verified, Theorem~2 of \cite{Andrews1994} yields the stochastic equicontinuity \eqref{eq:A1f-new}.
\end{proof}

For local asymptotic analysis, introduce a $\sqrt{n}$-scale reparameterization. Fix $N>0$ and write
\begin{equation*}
e = n^{-1/2}\triangle, \qquad \triangle\in\mathbb{R}^D,\quad \|\triangle\|\le N.
\end{equation*}
Define
\begin{equation}\label{eq:def-M}
\mathcal{M}(\triangle) = n^{-1/2}\sum_{i=1}^n\sum_{r=1}^R m(\omega_{i,r},\triangle), \qquad m(\omega_{i,r},\triangle) = \bm{\Gamma}_{i,r} \varphi_\tau\bigl(a \nu_{i,r} - n^{-1/2}\bm{\Gamma}_{i,r} ^\top\triangle\bigr),
\end{equation}
and $V(\triangle) = \mathcal{M}(\triangle)-\mathbb{E}\{\mathcal{M}(\triangle)\}$. By the change of variables $e_j = n^{-1/2}\triangle_j$ and $N^* = n^{-1/2}N$, Lemma~\ref{lem:equicontinuity} is equivalent to
\begin{equation}
\sup_{\|\triangle_1-\triangle_2\|\le N}\bigl\|V(\triangle_1)-V(\triangle_2)\bigr\| = o_p(1).
\label{eq:A2f-new}
\end{equation}

Taking $\triangle_2 = \bm{0}$ in \eqref{eq:A2f-new} yields
\begin{equation}
\sup_{\|\triangle\|\le N}\bigl\|\mathcal{M}(\triangle)-\mathcal{M}(\bm{0})-[\mathbb{E}\{\mathcal{M}(\triangle)\}-\mathbb{E}\{\mathcal{M}(\bm{0})\}]\bigr\| = o_p(1).
\label{eq:A3f-new}
\end{equation}
Under Assumptions~\ref{ass:A5}(i) and \ref{ass:A5}(iii), the standard expansion of the quantile score \citep[see, e.g.,][]{Kim_Muller2004} gives
\begin{equation*}
\mathbb{E}\{\mathcal{M}(\triangle)\} - \mathbb{E}\{\mathcal{M}(\bm{0})\}
= -a^{-1}\bm{Q}_{0(\nu)} \triangle + o(\|\triangle\|),
\end{equation*}
uniformly on $\{\|\triangle\|\le N\}$, where $\bm{Q}_{0(\nu)}$ is nonsingular by Assumption~\ref{ass:A5}(iii). Substituting into \eqref{eq:A3f-new} yields
\begin{equation}
\sup_{\|\triangle\|\le N} \bigl\|\mathcal{M}(\triangle)-\mathcal{M}(\bm{0}) + a^{-1}\bm{Q}_{0(\nu)} \triangle\bigr\| = o_p(1).
\label{eq:A4f-new}
\end{equation}

Let $\widehat{\bm{\Xi}}$ and $\widehat{\bm{\Theta}}$ be the first-stage spline estimators, with population values $\bm{\Xi}$ and $\bm{\Theta}$. Define
\begin{equation}\label{eq:Delta0-def}
\widehat{\triangle}_0 = (a-1)\sqrt{n}(\widehat{\bm{\Xi}}-\bm{\Xi}) + \sqrt{n}(\widehat{\bm{\Theta}}-\bm{\Theta}) \widetilde{\bm{\rho}}^{(\tau), 0}.
\end{equation}
Assumptions~\ref{ass:A4} and \ref{ass:A5} imply
\begin{equation*}
\sqrt{n}(\widehat{\bm{\Xi}}-\bm{\Xi}) = \mathcal{O}_p(1), \qquad \sqrt{n}(\widehat{\bm{\Theta}}-\bm{\Theta}) = \mathcal{O}_p(1),
\end{equation*}
hence $\widehat{\triangle}_0 = \mathcal{O}_p(1)$. Evaluating \eqref{eq:A4f-new} at $\triangle=\widehat{\triangle}_0$ gives
\begin{equation}\label{eq:M-Delta0}
\mathcal{M}(\widehat{\triangle}_0) = \mathcal{M}(\bm{0}) - a^{-1}\bm{Q}_{0(\nu)} \widehat{\triangle}_0 + o_p(1).
\end{equation}
Since $\varphi_\tau(a \nu_{i,r})=\varphi_\tau(\nu_{i,r})$ for any $a>0$,
\begin{equation*}
\mathcal{M}(\bm{0}) = n^{-1/2}\sum_{i=1}^n\sum_{r=1}^R \bm{\Gamma}_{i,r} \varphi_\tau(\nu_{i,r}).
\end{equation*}
By the Lindeberg-Feller central limit theorem and Assumptions~\ref{ass:A4}, \ref{ass:A5}(i), and \ref{ass:A5}(iv), $\mathcal{M}(\bm{0}) = \mathcal{O}_p(1)$, $\bm{Q}_{0(\nu)} \widehat{\triangle}_0 = \mathcal{O}_p(1)$, and hence
\begin{equation}\label{eq:A5f-new}
\mathcal{M}(\widehat{\triangle}_0) = \mathcal{O}_p(1).
\end{equation}

For $\delta\in\mathbb{R}^{K_y^2 + K_yK_x}$, consider the affine perturbation
\begin{equation}\label{eq:Delta1-def}
\widehat{\triangle}_1(\delta) = \bm{A}(\widehat{\bm{\Theta}}) \delta + \widehat{\triangle}_0,
\end{equation}
where $\bm{A}(\cdot)$ collects the relevant design blocks for the structural parameter $\bm{\theta}^{(\tau)}$ \citep[cf.][]{KoenkerZhao1996}. For any fixed $N_1>0$, \eqref{eq:A4f-new} implies
\begin{equation}
\sup_{\|\delta\|\le N_1} \bigl\|\mathcal{M}\{\widehat{\triangle}_1(\delta)\} -\mathcal{M}(\bm{0}) + a^{-1}\bm{Q}_{0(\nu)} \widehat{\triangle}_1(\delta)\bigr\| = o_p(1).
\label{eq:A7f-new}
\end{equation}
Introduce the transformed process $\widetilde{\mathcal{M}}(\delta) = \bm{A}(\widehat{\bm{\Theta}})^\top \mathcal{M}\{\widehat{\triangle}_1(\delta)\}$. Assumption~\ref{ass:A4} implies $\widehat{\bm{\Theta}}-\bm{\Theta}=o_p(1)$ and
\begin{equation*}
\|\bm{A}(\widehat{\bm{\Theta}})\|^2 = \operatorname{tr}\{\bm{A}(\widehat{\bm{\Theta}})\bm{A}(\widehat{\bm{\Theta}})^\top\} = \mathcal{O}_p(1).
\end{equation*}
Using \eqref{eq:A5f-new} and \eqref{eq:A7f-new}, together with the argument between (A.7) and (A.8) in \cite{Powel1983}, we obtain
\begin{equation}\label{eq:A8f-new}
\sup_{\|\delta\|\le N_1} \Bigl\| \widetilde{\mathcal{M}}(\delta)  - \bm{A}(\bm{\Theta})^\top \mathcal{M}(\widehat{\triangle}_0) + a^{-1}\bm{Q}_z \delta \Bigr\| = o_p(1),
\end{equation}
where $\bm{Q}_z = \bm{A}(\bm{\Theta})^\top \bm{Q}_{0(\nu)} \bm{A}(\bm{\Theta})$, which is positive definite by Assumption~\ref{ass:A5}(ii) and the identification conditions.

The penalized IV-quantile estimator $\widehat{\bm{\theta}}^{(\tau)}$ satisfies
\begin{equation*}
\frac{1}{n}\sum_{i=1}^n\sum_{r=1}^R \widehat{\bm{\Pi}}_{i,r} \varphi_\tau\bigl\{\Y_i(t_r) - [\widehat{\bm{\Pi}}\widehat{\bm{\theta}}^{(\tau)}]_{i,r}\bigr\} + \bm{\mathcal P}_{\mathrm{surf},\lambda_\rho,\lambda_\beta}\widehat{\bm{\theta}}^{(\tau)} = \bm{o}_p(n^{-1/2}),
\end{equation*}
where $\widehat{\bm{\Pi}}_{i,r}$ is the $(i,r)$-th row of the Stage~2 design matrix and $\bm{\mathcal P}_{\mathrm{surf},\lambda_\rho,\lambda_\beta}$ is the block penalty matrix. Multiplying by $\sqrt{n}$ and using the empirical process notation, we can rewrite this as
\begin{equation*}
\mathcal{M}\{\widehat{\Delta}_1(\widehat{\delta})\} + \sqrt{n} \bm{\mathcal P}_{\mathrm{surf},\lambda_\rho,\lambda_\beta}\widehat{\bm{\theta}}^{(\tau)} = o_p(1),
\end{equation*}
where $\widehat{\delta} = \sqrt{n} \left(\widehat{\bm\theta}_{\tau} - \bm\theta_{\tau,K}^{0} \right)$.

\begin{lemma}\label{lem:smoothing}
Let $\widehat{\bm\theta}^{(\tau)}_{\alpha}$ denote the minimizer of the smoothed penalized criterion with smoothing constant $\alpha=\alpha_n$, and let $\widehat{\bm\theta}^{(\tau)}$ denote the minimizer of the exact penalized check-loss criterion. Under Assumptions~\ref{ass:A3}, \ref{ass:A4}, and \ref{ass:A5}, $\widehat{\bm\theta}^{(\tau)}_{\alpha} -\widehat{\bm\theta}^{(\tau)}=o_p(n^{-1/2})$.
\end{lemma}

\begin{proof}
The derivative of the smoothed loss satisfies $|\ell_{\tau,\alpha}'(u)-\varphi_\tau(u)| =\{1+\exp(|u|/\alpha)\}^{-1}\le\exp(-|u|/\alpha)$, so the smoothed and exact subgradient processes differ, at any $\bm\theta$, by at most $n^{-1/2}\sum_{i,r}\|\bm\Gamma_{i,r}\|\exp\{-|u_{ir}(\bm\theta)|/\alpha\}$. By Assumption~\ref{ass:A5}(iii) the conditional density of the residual at zero is bounded, so $\E[\exp\{-|u_{ir}|/\alpha\}]=\mathcal O(\alpha)$, and the difference of the two score processes is $\mathcal O_p(\sqrt n\,\alpha_n)=o_p(1)$ uniformly over $o_p(1)$ neighborhoods of $\bm\theta_{\tau,K}^{0}$. Standard convexity arguments \citep[e.g.,][]{Pollard1991} then yield the claim.
\end{proof}

We now invoke Lemma~A.4 of \cite{KoenkerZhao1996}, which requires the following conditions:
\begin{itemize}
\item[C$_3$] There exists $c^*>1$ such that $\delta^\top \widetilde{\mathcal{M}}(c^*\delta) \ge \delta^\top \widetilde{\mathcal{M}}(\delta)$ for all $\delta\in\mathbb{R}^{d_0}$ where $d_0$ denotes the (fixed) dimension of $\bm \theta^{(\tau)}$ including the intercept block;
\item[C$_4$] $\|\bm{A}(\bm{\Theta})^\top \mathcal{M}(\widehat{\triangle}_0)\| = \mathcal{O}_p(1)$;
\item[C$_5$] $\widetilde{\mathcal{M}}(\widehat{\delta}) = o_p(1)$;
\item[C$_6$] $\bm{Q}_z$ is positive definite.
\end{itemize}

Condition~C$_6$ has already been verified. Condition~C$_4$ follows from \eqref{eq:A5f-new}. For C$_3$, consider
\begin{equation*}
h^*(c) = \sum_{i=1}^n\sum_{r=1}^R \eta_\tau\left\{a \nu_{i,r}
- n^{-1/2}\bm{\Gamma}_{i,r}^\top\bm{A}(\widehat{\bm{\Theta}})\delta c
- n^{-1/2}\bm{\Gamma}_{i,r}^\top\widehat{\Delta}_0 \right\},
\end{equation*}
where $\eta_\tau(c) = c \varphi_\tau(c)$ is the check loss. The function $c\mapsto h^*(c)$ is convex, so its derivative in $c$, equal to $\delta^\top\widetilde{\mathcal{M}}(c\delta)$, is nondecreasing in $c$. This implies C$_3$.

For C$_5$, note that
\begin{equation*}
\sqrt{n} \widetilde{\mathcal{M}}(\widehat{\delta}) = \left[\left.\frac{\partial S(\bm{\theta}^{(\tau)})}{\partial\bm{\theta}}\right|_{\bm{\theta}=\widehat{\bm{\theta}}^{(\tau)}}\right]_{-},
\end{equation*}
the left-hand partial derivative of the penalized objective $S$ evaluated at its minimizer $\widehat{\bm{\theta}}^{(\tau)}$. First-order optimality and the fact that the penalty term is $o_p(n^{-1/2})$ imply that this derivative is $o_p(1)$, and thus $\widetilde{\mathcal{M}}(\widehat{\delta})=o_p(1)$.

Hence all conditions C$_3$-C$_6$ hold, and Lemma~A.4 of \cite{KoenkerZhao1996} yields
\begin{equation}\label{eq:Bahadur-theta}
\widehat{\delta} = \mathcal{O}_p(1), \qquad \widehat{\delta} = a \bm{Q}_z^{-1} \bm{A}(\bm{\Theta})^\top \mathcal{M}(\widehat{\Delta}_0) + o_p(1).
\end{equation}
By Lemma~\ref{lem:smoothing}, the same expansion holds for the smoothed estimator computed in practice.

Next, incorporate the asymptotic linear representations of the first-stage estimators \citep[cf.][]{Kim_Muller2004}:
\begin{align*}
\sqrt{n}(\widehat{\bm{\Xi}}-\bm{\Xi}) &= \bm{Q}_{0(\nu)}^{-1} n^{-1/2}\sum_{i=1}^n\sum_{r=1}^R \bm{\Gamma}_{i,r} \varphi_\tau(\nu_{i,r}) + o_p(1), \\
\sqrt{n}(\widehat{\bm{\Theta}}-\bm{\Theta}) &= \bm{Q}_{0(\zeta)}^{-1} n^{-1/2}\sum_{i=1}^n\sum_{r=1}^R \bm{\Gamma}_{i,r} \varphi_\tau(\zeta_{i,r}) + o_p(1),
\end{align*}
where $\zeta_{i,r}$ are auxiliary residuals from the first stage and $\bm{Q}_{0(\zeta)}$ is nonsingular. Substituting these into \eqref{eq:Delta0-def} and then \eqref{eq:Bahadur-theta}, and collecting terms, we obtain
\begin{align}\label{eq:A10f-new}
\sqrt{n}\bigl(\widehat{\bm{\theta}}^{(\tau)}-\bm{\theta}_{\tau,K}^{0}\bigr) &= n^{-1/2}\sum_{i=1}^n\sum_{r=1}^R \bm{Q}_z^{-1}\bm{A}(\bm{\Theta})^\top \left\{ \bm{\Gamma}_{i,r}\, \varphi_\tau(\nu_{i,r}) - \bm{Q}_{0(\nu)}\bm{Q}_{0(\zeta)}^{-1} \bigl(\widetilde{\bm{\rho}}^{(\tau),0 \top}\otimes\bm I\bigr) \bigl\{\varphi_\tau(\zeta_{i,r})\otimes\bm{\Gamma} _{i,r}\bigr\}\right\} + o_p(1) \nonumber\\
&= \bm{J} n^{-1/2}\sum_{i=1}^n \bm{Z}_i + o_p(1),
\end{align}
where
\begin{equation*}
\bm{Z}_i= \sum_{r=1}^R \bigl\{\varphi_\tau(\widetilde{\mathcal{W}}_{i,r})\otimes\bm{\Gamma}_{i,r}\bigr\}, \qquad \varphi_\tau(\widetilde{\mathcal{W}}_{i,r})
= \bigl\{\varphi_\tau(\nu_{i,r}),\varphi_\tau(\zeta_{i,r})\bigr\}^\top,
\end{equation*}
and the deterministic matrix $\bm{J}$ is implicit in \eqref{eq:A10f-new}. Assumption~\ref{ass:A4} guarantees that $\{\bm{Z}_i\}$ are i.i.d.; since $|\varphi_\tau(\cdot)|\le 1$ and $\mathbb{E}\|\bm{\Gamma}_{i,r}\|^2<\infty$ by Assumption~\ref{ass:A5}(i), we have $\mathbb{E}\|\bm{Z}_i\|^2<\infty$. By the Lindeberg--Feller central limit theorem,
\begin{equation}\label{eq:theta-CLT}
\sqrt{n} \left(\widehat{\bm\theta}_{\tau} - \bm\theta_{\tau,K}^{0} \right) \xrightarrow{d} \mathcal N \left( \bm0, \bm\Sigma_{\theta,\tau} \right).
\end{equation}
with $\bm{\Sigma}_{\theta,\tau} = \bm{J}\, \bm{\mathcal{S}} \bm{J}^\top$ and $\bm{\mathcal{S}} = \mathbb{E}(\bm{Z}_i\bm{Z}_i^\top)$. By Lemma~\ref{lem:smoothing}, the same limit law holds for the smoothed estimator $\widehat{\bm\theta}^{(\tau)}_{\alpha_n}$ used in computation.

By construction,
\begin{equation*}
\widehat{\rho}_\tau(t,u) = \bigl\{\bm{\phi}^\top(u)\otimes\bm{\phi}^\top(t)\bigr\} \widehat{\widetilde{\bm{\rho}}}^{(\tau)}, \qquad \widehat{\beta}_\tau(t,s) = \bigl\{\bm{\psi}^\top(s)\otimes\bm{\phi}^\top(t)\bigr\} \widehat{\widetilde{\bm{b}}}^{(\tau)},
\end{equation*}
where $\widehat{\widetilde{\bm{\rho}}}^{(\tau)}$ and $\widehat{\widetilde{\bm{b}}}^{(\tau)}$ are the subvectors of $\widehat{\bm{\theta}}^{(\tau)}$ corresponding to $\rho_\tau$ and $\beta_\tau$. Define the (row-vector) evaluation operators 
\begin{equation*}
\mathcal{F}_\rho(t,u) = \bm{\phi}^\top(u)\otimes\bm{\phi}^\top(t), \qquad \mathcal{F}_\beta(t,s) = \bm{\psi}^\top(s)\otimes\bm{\phi}^\top(t).
\end{equation*}
Then, for fixed $(t,u)$ and $(t,s)$,
\begin{align*}
\sqrt{n}\bigl\{\widehat{\rho}_\tau(t,u)-\rho_{\tau,K}(t,u)\bigr\}
= \mathcal{F}_\rho(t,u)\, \sqrt{n}\bigl(\widehat{\widetilde{\bm{\rho}}}^{(\tau)}-\widetilde{\bm{\rho}}^{(\tau),0}\bigr), \\
\sqrt{n}\bigl\{\widehat{\beta}_\tau(t,s)-\beta_{\tau,K}(t,s)\bigr\}
= \mathcal{F}_\beta(t,s)\, \sqrt{n}\bigl(\widehat{\widetilde{\bm{b}}}^{(\tau)}-\widetilde{\bm{b}}^{(\tau),0}\bigr).
\end{align*}
Since $\mathcal{F}_\rho(t,u)$ and $\mathcal{F}_\beta(t,s)$ are continuous linear maps and \eqref{eq:theta-CLT} holds, the continuous mapping theorem implies that, for each fixed $(t,u)$ and $(t',u')$,
\begin{align*}
\sqrt{n}\bigl\{\widehat{\rho}_\tau(t,u)-\rho_{\tau,K}(t,u)\bigr\}
&\xrightarrow{d} \mathcal{GP}\left\{0,\Sigma_{\rho_\tau}(t,u;t',u')\right\}, \\
\sqrt{n}\bigl\{\widehat{\beta}_\tau(t,s)-\beta_{\tau,K}(t,s)\bigr\}
&\xrightarrow{d} \mathcal{GP}\left\{0,\Sigma_{\beta_\tau}(t,s;t',s')\right\},
\end{align*}
where $\mathcal{GP}\{0,\cdot\}$ denotes a mean-zero Gaussian process with the indicated covariance kernel:
\begin{align}
\Sigma_{\rho_\tau}(t,u;t',u')
&= \mathcal{F}_\rho(t,u) \bm{\Sigma}_{\bm{\rho}^{(\tau)}} \mathcal{F}_\rho^\top(t',u'), \label{eq:Sigma-rho}\\
\Sigma_{\beta_\tau}(t,s;t',s') &= \mathcal{F}_\beta(t,s) \bm{\Sigma}_{\bm{b}^{(\tau)}} \mathcal{F}_\beta^\top(t',s'), \label{eq:Sigma-beta}
\end{align}
where $\bm{\Sigma}_{\bm{\rho}^{(\tau)}}$ and $\bm{\Sigma}_{\bm{b}^{(\tau)}}$ are the $K_y^2\times K_y^2$ and $K_yK_x\times K_yK_x$ submatrices of $\bm\Sigma_{\theta,\tau}$ corresponding to $\widetilde{\bm{\rho}}^{(\tau)}$ and $\widetilde{\bm{b}}^{(\tau)}$, respectively. In explicit tensor form,
\begin{align*}
\Sigma_{\rho_\tau}(t,u;t',u') &= \bigl\{\bm{\phi}^\top(u)\otimes\bm{\phi}^\top(t)\bigr\} \bm{\Sigma}_{\bm{\rho}^{(\tau)}} \bigl\{\bm{\phi}(u')\otimes\bm{\phi}(t')\bigr\}, \\
\Sigma_{\beta_\tau}(t,s;t',s') &= \bigl\{\bm{\psi}^\top(s)\otimes\bm{\phi}^\top(t)\bigr\} \bm{\Sigma}_{\bm{b}^{(\tau)}} \bigl\{\bm{\psi}(s')\otimes\bm{\phi}(t')\bigr\}.
\end{align*}

Since $K_y$ and $K_x$ are fixed, the maps $\bm c\mapsto\{\bm\phi^\top(u)\otimes\bm\phi^\top(t)\}\bm c$ and $\bm c\mapsto\{\bm\psi^\top(s)\otimes\bm\phi^\top(t)\}\bm c$ are continuous linear maps from $\mathbb R^{K_y^2}$ and $\mathbb R^{K_yK_x}$ into $C([0,1]^2)$. By the continuous mapping theorem applied to \eqref{eq:theta-CLT},
\begin{equation*}
\sqrt n (\widehat\rho_\tau-\rho_{\tau,K})\rightsquigarrow
\mathbb G_{\rho,\tau}\ \text{in }C([0,1]^2), \qquad \sqrt n (\widehat\beta_\tau-\beta_{\tau,K})\rightsquigarrow \mathbb G_{\beta,\tau}\ \text{in }C([0,1]^2),
\end{equation*}
where $\mathbb G_{\rho,\tau}$ and $\mathbb G_{\beta,\tau}$ are the mean-zero finite-rank Gaussian processes with covariance kernels \eqref{eq:Sigma-rho}-\eqref{eq:Sigma-beta}. If the true surfaces belong to the working spline spaces, then $\rho_{\tau,K}=\rho_\tau$ and $\beta_{\tau,K}=\beta_\tau$, and the centering may be taken at the true surfaces. This completes the proof of Theorem~3.1.
\end{proof}

\section{Additional simulation details and results}
\label{sec:supp-simulation}

The general simulation framework, including the spatial weight matrix construction, the functional predictor generation, and the SAR operator specification, follows the setup described in Section 4 of the main text. Here, we provide the detailed data generation processes for the error scenarios omitted from the main manuscript for the sake of brevity.

The first scenario (Case 1) serves as a baseline to evaluate the models under ideal conditions where the assumptions of mean-based spatial regression are fully satisfied. The error processes are independent and identically distributed (i.i.d.) Gaussian white noise: $\epsilon_i(t) \stackrel{i.i.d.}{\sim} \mathcal{N}(0, \sigma_{\epsilon}^2)$, where the scale parameter is set to $\sigma_{\epsilon} = 0.01$. This case assesses the relative efficiency of the SFoF-QR estimator when the least-squares criterion is theoretically optimal.

The third scenario (Case 3) evaluates the robustness of the proposed framework to leptokurtosis (heavy tails) without the influence of skewness. The errors are generated from an Asymmetric Laplace Distribution (ALD) with a skewness parameter $s_p=0.5$, which yields the symmetric Laplace distribution: $\epsilon_i(t) \stackrel{i.i.d.}{\sim} \text{ALD}(\mu=0, \sigma=\sigma_{\epsilon}, s_p=0.5)$. The probability density function is given by $f(\epsilon) = \frac{1}{2b_e} \exp(-\frac{|\epsilon - \mu|}{b_e})$, where $\mu=0$ and $b_e$ is the scale parameter. With a kurtosis of 6 (twice that of a Gaussian distribution), this scenario provides a clear challenge for mean-based models, which are sensitive to the frequent high-magnitude shocks characteristic of heavy-tailed processes.

The results obtained from our Monte Carlo experiments are presented in Tables~\ref{tab:tab_1}-\ref{tab:tab_3}. The results provide overwhelming evidence that correctly incorporating the functional SAR structure is paramount for consistent estimation and prediction, particularly when the spatial dependence strength ($s_d$) is moderate to strong. Models that ignore spatial dependence (M$_3$, M$_4$, M$_5$) suffer catastrophic performance degradation across all metrics, confirming that the SAR specification cannot be neglected in spatially correlated functional data. Furthermore, the comparison between the penalized B-spline methodology (M$_2$, M$_6$) and the FPCA-based methodology (M$_1$, M$_5$) is fundamental to understanding the superior performance of the proposed methods. From the results in Tables~\ref{tab:tab_1}-\ref{tab:tab_3}, the FPCA-based non-spatial quantile method (M$_5$) fails to produce reliable or informative PIs when directly estimating extreme quantiles (e.g., $Q_{0.025}$ and $Q_{0.975}$). This limitation arises primarily from two methodological challenges: the distortion introduced by dimension reduction in the tails and the resulting high variance of the quantile-specific parameter estimates. Therefore, in what follows, we focus our PI performance assessment exclusively on the proposed M$_6$ method.

Across all three data generation scenarios, the B-spline penalized models (M$_2$ and M$_6$) demonstrate a massive and consistent advantage in coefficient estimation accuracy (RRISPEE($\beta$)) compared to their FPCA-based counterparts (M$_1$ and M$_5$). This superiority holds irrespective of the spatial dependence strength or the sample size, suggesting a fundamental limitation in the FPCA approach for functional regression involving complex bivariate kernels. For instance, examining the large sample size setting ($n=500$) across the three error cases (see Table~\ref{tab:pretab_1}) provides clear quantification of this disparity. 

In the ideal Case 1, $n=500$, and $s_d=0.1$, the penalized mean model (M$_2$) achieved a RRISPEE($\beta$) of $0.074$, whereas the FPCA-based model (M$_1$) resulted in a RRISPEE($\beta$) of $31.249$. Similarly, the proposed $M_6$ achieved $0.078$ compared to $M_5$'s $17.332$. This indicates that $M_2$ recovered the true coefficient surface $\beta(t,s)$ with over 422 times the accuracy of $M_1$, and $M_6$ showed over 222 times the accuracy of $M_5$. The discrepancy is equally profound in Case 3, where $M_6$ at $n=500, s_d=0.9$ (RRISPEE($\beta$)=0.102) was vastly superior to $M_5$ (RRISPEE($\beta$)=17.722). This systematic performance gap suggests that relying on FPCA for dimension reduction, typically set to retain 95\% of total variance, is insufficient when the true functional coefficients ($\beta(t,s)$ and $\rho(t,u)$) possess complex, non-separable structures. These functional parameters are themselves infinite-dimensional objects, and their accurate representation often requires capturing components that contribute minimally to the total variance of the data process $\X_i(s)$ but are crucial for describing the regression relationship. The FPCA method projects the data onto a truncated basis derived from the data's covariance operator, potentially missing the necessary structure required for precise recovery of the complex regression surface. In contrast, the penalized B-spline approach (M$_2$, M$_6$) bypasses this limitation. By expressing the coefficient surfaces directly via tensor-product B-splines and applying quadratic roughness penalties, the method ensures an optimal degree of smoothness while fitting the loss function. This regularization strategy allows the model to flexibly and accurately recover the true functional shapes without the destructive approximation error inherent in truncating the functional space based purely on principal component variance. Consequently, the success of both M$_2$ and M$_6$ hinges critically on this sophisticated penalized estimation procedure.

\begin{small}
\begin{center}
\tabcolsep 0.255in
\renewcommand{\arraystretch}{0.92}
\begin{longtable}{@{}lcccccc@{}} 
\caption{Estimation accuracy benchmark: Penalized B-Splines vs.\ FPCA (RRISPEE$(\beta)$ for $n=500$).}\label{tab:pretab_1} \\
\toprule
Case / $s_d$ & $M_{1}$ & $M_{2}$ & $M_{5}$ & $M_{6}$ & $M_{1}/M_{2}$ Ratio & $M_{5}/M_{6}$ Ratio \\
\midrule
\endfirsthead
\toprule
Case / $s_d$ & $M_{1}$ & $M_{2}$ & $M_{5}$ & $M_{6}$ & $M_{1}/M_{2}$ Ratio & $M_{5}/M_{6}$ Ratio \\
\midrule
\endhead
\midrule
\multicolumn{7}{r}{Continued on next page} \\ 
\endfoot
\endlastfoot
Case~1 / 0.1 & 31.249 & 0.074 & 17.332 & 0.078 & $\sim 422:1$ & $\sim 222:1$ \\
Case~2 / 0.5 & 32.775 & 0.252 & 17.375 & 0.146 & $\sim 130:1$ & $\sim 119:1$ \\
Case~3 / 0.9 & 31.317 & 0.172 & 17.722 & 0.102 & $\sim 182:1$ & $\sim 174:1$ \\
\bottomrule
\end{longtable}
\end{center}
\end{small}

Under Case 1, M$_6$ (SFoF-QR, estimating the median) performs similarly to, or marginally worse than, its conditional mean counterpart (M$_2$), as expected by statistical theory. For the largest sample size ($n=500$) and weak spatial dependence ($s_d=0.1$), M$_2$ achieved RRISPEE($\beta$)=0.074, marginally better than M$_6$'s $0.078$. Under strong dependence ($s_d=0.9$), M$_2$ achieved $0.075$ compared to M$_6$'s $0.234$. This finite-sample variance trade-off is consistent with the general statistical property that quantile regression, while robust, is typically less efficient than least squares when the error distribution is perfectly Gaussian.

A significant observation across both M$_2$ and M$_6$ is the consistently higher estimation error for the spatial kernel, RRISPEE($\rho$), compared to the predictor kernel, RRISPEE($\beta$). For $n=500, s_d=0.9$, RRISPEE($\beta$) was $0.075$ (M$_2$) versus RRISPEE($\rho$)=$6.578$ (M$_2$). The discrepancy is attributable to the methodological complexity inherent in estimating the SAR functional kernel $\rho(t,u)$. The estimation of $\rho(t,u)$ relies on a two-stage IV procedure required to address endogeneity of the lagged response $\widetilde{\Y}_{i}(u)$. This introduces additional uncertainty and estimation variance in the first stage (where $\widetilde{\Y}_{i}$ is instrumented), propagating into the final estimate $\widehat{\rho}(t,u)$. Thus, the estimation of the spatial dependence structure is significantly more challenging than the estimation of the exogenous covariate effect.

The RMSPE results underscore the critical importance of explicitly modeling spatial dependence, especially as its strength increases. For weak spatial dependence ($s_d=0.1$), non-spatial models show competitive, albeit inferior, performance compared to the spatial models. For $n=500$, the non-spatial penalized mean (M$_3$) achieved an RMSPE of $1.754$, which is substantially higher than M$_2$'s $0.203$ and M$_6$'s $0.423$, but relatively low in absolute terms. However, the necessity of the SAR structure becomes overwhelming as $s_d$ increases. The non-spatial models (M$_3$, M$_4$, M$_5$) yield prediction errors that are 10 to 13 times higher than the best spatial models. This dramatic failure is a direct consequence of model misspecification: when $s_d$ is large, the SAR operator accounts for a substantial proportion of the functional response variance. By omitting the spatially lagged response term $\widetilde{\Y}_{i}(u)$, non-spatial models assume independent errors. Since the true errors in these models contain the spatial spillover effect, they become highly correlated with the functional predictor $\X_i(s)$, violating the fundamental assumption of strict exogeneity in the regression. This leads to severely biased and inconsistent estimates for $\widehat{\beta}(t,s)$, rendering the out-of-sample predictions (RMSPE) useless. Classical functional regression approaches such as the penalized function-on-function regression of \texttt{pffr}/\texttt{refund} (M$_3$) and the boosting-based \texttt{FDboost} (M$_4$) are not equipped to handle this type of SAR structure and thus break down in the presence of strong spatial dependence.

For the interval results in Tables~\ref{tab:tab_1}-\ref{tab:tab_3}, empirical coverage (EC) is defined curvewise over the response evaluation grid. Specifically, a test response curve is counted as covered only when $\widehat Q_{0.025,i}(t_r) \leq \Y_i(t_r) \leq \widehat Q_{0.975,i}(t_r)$ for every $r=1,\ldots,R$. Thus, EC represents simultaneous coverage of the complete discretized response curve, rather than average pointwise coverage over the grid. Because failure at any one grid point makes the curvewise coverage indicator zero, this criterion is substantially stricter than pointwise coverage and helps explain the zero EC values obtained for M$_5$.

The analysis of the prediction intervals constructed by M$_6$ (using $\tau=0.025$ and $\tau=0.975$) confirms the reliability of the SFoF-QR framework for uncertainty quantification. Across all sample sizes and spatial strengths in Case~1, M$_6$ demonstrated high EC. For instance, at $n=500$ and $s_d=0.9$, the EC was $0.989$ against the nominal $95\%$ level, resulting in a CPD of $0.064$ and a low score of $1.405$. The consistent over-coverage (EC $> 0.95$) coupled with a low score indicates that M$_6$ provides prediction intervals that are conservative (slightly wider than strictly necessary) but highly informative (sharp). This reliable, albeit conservative, behavior stems from the two-stage penalized estimation process where BIC selection favors highly smooth functions, which can slightly inflate the residual variance estimates in the tails, providing robust uncertainty bands even in the presence of estimation uncertainty introduced by the IV procedure.

\begin{small}
\begin{center}
\tabcolsep 0.175in
\renewcommand{\arraystretch}{0.92}
\begin{longtable}{@{}ccccccccc@{}} 
\caption{Computed mean $\text{RRISPEE}(\widehat{\beta})$, $\text{RRISPEE}(\widehat{\rho})$, RMSPE, $\text{EC}$, $\text{CPD}$, and $\text{score}$ values with their standard errors (given in brackets) under Case 1. The results are obtained over 250 Monte-Carlo replications using three sample sizes ($n$) and three different strength of spatial dependence ($s_d$). Metrics EC, CPD, and score are not applicable for mean-regression models, and are thus presented as ----.}\label{tab:tab_1} \\
\toprule
{$s_d$} & {$n$} & Method & $\text{RRISPEE}(\widehat{\beta})$ & $\text{RRISPEE}(\widehat{\rho})$ & RMSPE & $\text{EC}$ & $\text{CPD}$ & $\text{score}$ \\ \midrule
\endfirsthead
\toprule
{$s_d$} & {$n$} & Method & $\text{RRISPEE}(\widehat{\beta})$ & $\text{RRISPEE}(\widehat{\rho})$ & RMSPE & $\text{EC}$ & $\text{CPD}$ & $\text{score}$ \\ \midrule
\endhead
\midrule 
\multicolumn{9}{r}{Continued on next page} \\ 
\endfoot
\endlastfoot
0.1 & 100   &  M$_1$  & 31.015 & 152.538 & 5.703 & ---- & ---- & ---- \\
    &       &         & (22.329) & (818.949) & (2.502) & (----) & (----) & (----) \\
    &       &  M$_2$  & 0.130 & 12.601 & 0.300 & ---- & ---- & ---- \\
    &       &         & (0.014) & (3.082) & (0.021) & (----) & (----) & (----) \\
    &       &  M$_3$  & 2.910 & ---- & 1.986 & ---- & ---- & ---- \\
    &       &         & (0.693) & (----) & (0.338) & (----) & (----) & (----) \\
    &       &  M$_4$  & 3.399 & ---- & 5.315 & ---- & ---- & ---- \\
    &       &         & (2.142) & (----) & (1.647) & (----) & (----) & (----) \\
    &       &  M$_5$  & 19.405 & ---- & 4.915 & 0.000 & 0.950 & 15.829 \\
    &       &         & (3.825) & (----) & (1.931) & (0.000) & (0.000) & (11.707) \\
    &       &  M$_6$  & 0.129 & 16.051 & 0.578 & 0.998 & 0.048 & 0.091 \\
    &       &         & (0.014) & (4.852) & (0.094) & (0.001) & (0.001) & (0.001) \\
    &       &  $\text{M}_6^{(\text{Opt})}$  & 0.120 & 14.121 & 0.551 & 0.999 & 0.046 & 0.090 \\
    &       &         & (0.012) & (4.544) & (0.091) & (0.001) & (0.001) & (0.001) \\
\cmidrule(l){2-9}				
    & 250   &  M$_1$  & 33.760 & 115.116 & 5.727 & ---- & ---- & ---- \\
    &       &         & (25.020) & (410.589) & (1.610) & (----) & (----) & (----) \\
    &       &  M$_2$  & 0.094 & 12.334 & 0.239 & ---- & ---- & ---- \\
    &       &         & (0.009) & (3.364) & (0.014) & (----) & (----) & (----) \\
    &       &  M$_3$  & 2.854 & ---- & 1.820 & ---- & ---- & ---- \\
    &       &         & (0.381) & (----) & (0.227) & (----) & (----) & (----) \\
    &       &  M$_4$  & 2.149 & ---- & 4.512 & ---- & ---- & ---- \\
    &       &         & (0.835) & (----) & (1.262) & (----) & (----) & (----) \\
    &       &  M$_5$  & 17.798 & ---- & 3.982 & 0.000 & 0.950 & 10.981 \\
    &       &         & (2.056) & (----) & (1.489) & (0.000) & (0.000) & (6.888) \\
    &       &  M$_6$  & 0.097 & 16.616 & 0.515 & 0.999 & 0.049 & 0.091 \\
    &       &         & (0.010) & (5.547) & (0.114) & (0.001) & (0.001) & (0.001) \\
    &       &  $\text{M}_6^{(\text{Opt})}$  & 0.088 & 15.311 & 0.502 & 0.999 & 0.049 & 0.091 \\
    &       &         & (0.011) & (5.331) & (0.111) & (0.001) & (0.001) & (0.001) \\
\cmidrule(l){2-9}
    & 500   &  M$_1$  & 31.249 & 252.022 & 5.734 & ---- & ---- & ---- \\
    &       &         & (21.009) & (1192.677) & (3.177) & (----) & (----) & (----) \\
    &       &  M$_2$  & 0.074 & 12.385 & 0.203 & ---- & ---- & ---- \\
    &       &         & (0.006) & (3.629) & (0.010) & (----) & (----) & (----) \\
    &       &  M$_3$  & 2.839 & ---- & 1.754 & ---- & ---- & ---- \\
    &       &         & (0.254) & (----) & (0.144) & (----) & (----) & (----) \\
    &       &  M$_4$  & 1.728 & ---- & 4.001 & ---- & ---- & ---- \\
    &       &         & (0.314) & (----) & (0.979) & (----) & (----) & (----) \\
    &       &  M$_5$  & 17.332 & ---- & 3.288 & 0.000 & 0.950 & 9.548 \\
    &       &         & (1.376) & (----) & (1.178) & (0.000) & (0.000) & (5.141) \\
    &       &  M$_6$  & 0.078 & 14.466 & 0.423 & 0.999 & 0.049 & 0.091 \\
    &       &         & (0.007) & (4.853) & (0.109) & (0.001) & (0.001) & (0.001) \\
    &       &  $\text{M}_6^{(\text{Opt})}$  & 0.073 & 13.899 & 0.400 & 0.999 & 0.049 & 0.090 \\
    &       &         & (0.007) & (4.611) & (0.107) & (0.001) & (0.001) & (0.001) \\
\midrule

0.5 & 100   &  M$_1$  & 33.795 & 119.034 & 6.428 & ---- & ---- & ---- \\
    &       &         & (46.907) & (1442.225) & (10.310) & (----) & (----) & (----) \\
    &       &  M$_2$  & 0.130 & 7.930 & 0.340 & ---- & ---- & ---- \\
    &       &         & (0.014) & (0.561) & (0.042) & (----) & (----) & (----) \\
    &       &  M$_3$  & 6.530 & ---- & 5.337 & ---- & ---- & ---- \\
    &       &         & (3.557) & (----) & (1.306) & (----) & (----) & (----) \\
    &       &  M$_4$  & 4.647 & ---- & 7.153 & ---- & ---- & ---- \\
    &       &         & (2.209) & (----) & (2.183) & (----) & (----) & (----) \\
    &       &  M$_5$  & 21.371 & ---- & 6.872 & 0.000 & 0.950 & 43.943 \\
    &       &         & (7.576) & (----) & (2.276) & (0.000) & (0.000) & (21.190) \\
    &       &  M$_6$  & 0.130 & 7.844 & 0.798 & 0.999 & 0.049 & 0.163 \\
    &       &         & (0.014) & (0.966) & (0.155) & (0.001) & (0.001) & (0.001) \\
    &       &  $\text{M}_6^{(\text{Opt})}$  & 0.125 & 7.522 & 0.761 & 0.999 & 0.049 & 0.161 \\
    &       &         & (0.011) & (0.950) & (0.151) & (0.001) & (0.001) & (0.001) \\
\cmidrule(l){2-9}
					
    & 250   &  M$_1$  & 33.967 & 49.458 & 5.884 & ---- & ---- & ---- \\
    &       &         & (24.072) & (215.981) & (1.825) & (----) & (----) & (----) \\
    &       &  M$_2$  & 0.093 & 7.480 & 0.281 & ---- & ---- & ---- \\
    &       &         & (0.009) & (0.735) & (0.033) & (----) & (----) & (----) \\
    &       &  M$_3$  & 4.447 & ---- & 4.728 & ---- & ---- & ---- \\
    &       &         & (1.921) & (----) & (0.905) & (----) & (----) & (----) \\
    &       &  M$_4$  & 2.669 & ---- & 6.083 & ---- & ---- & ---- \\
    &       &         & (0.907) & (----) & (1.638) & (----) & (----) & (----) \\
    &       &  M$_5$  & 17.874 & ---- & 5.759 & 0.000 & 0.950 & 37.881 \\
    &       &         & (2.079) & (----) & (1.625) & (0.000) & (0.000) & (17.220) \\
    &       &  M$_6$  & 0.100 & 7.538 & 0.757 & 0.999 & 0.049 & 0.163 \\
    &       &         & (0.011) & (1.350) & (0.201) & (0.001) & (0.001) & (0.001) \\
    &       &  $\text{M}_6^{(\text{Opt})}$  & 0.101 & 7.542 & 0.759 & 0.999 & 0.049 & 0.164 \\
    &       &         & (0.011) & (1.354) & (0.206) & (0.001) & (0.001) & (0.001) \\
\cmidrule(l){2-9}

    & 500   &  M$_1$  & 35.288 & 93.437 & 6.317 & ---- & ---- & ---- \\
    &       &         & (64.404) & (480.723) & (6.521) & (----) & (----) & (----) \\
    &       &  M$_2$  & 0.075 & 7.075 & 0.253 & ---- & ---- & ---- \\
    &       &         & (0.006) & (0.756) & (0.032) & (----) & (----) & (----) \\
    &       &  M$_3$  & 3.719 & ---- & 4.456 & ---- & ---- & ---- \\
    &       &         & (1.342) & (----) & (0.643) & (----) & (----) & (----) \\
    &       &  M$_4$  & 2.068 & ---- & 5.483 & ---- & ---- & ---- \\
    &       &         & (0.454) & (----) & (1.284) & (----) & (----) & (----) \\
    &       &  M$_5$  & 17.348 & ---- & 5.114 & 0.000 & 0.950 & 28.207 \\
    &       &         & (1.380) & (----) & (1.216) & (0.000) & (0.000) & (15.786) \\
    &       &  M$_6$  & 0.083 & 7.650 & 0.623 & 0.999 & 0.049 & 0.164 \\
    &       &         & (0.010) & (1.862) & (0.255) & (0.001) & (0.001) & (0.001) \\
    &       &  $\text{M}_6^{(\text{Opt})}$  & 0.081 & 7.649 & 0.620 & 0.999 & 0.049 & 0.161 \\
    &       &         & (0.010) & (1.859) & (0.251) & (0.001) & (0.001) & (0.001) \\
\midrule

0.9 & 100   &  M$_1$  & 32.315 & 24.613 & 10.530 & ---- & ---- & ---- \\
    &       &         & (22.619) & (23.124) & (12.346) & (----) & (----) & (----) \\
    &       &  M$_2$  & 0.130 & 7.669 & 0.855 & ---- & ---- & ---- \\
    &       &         & (0.014) & (0.355) & (0.304) & (----) & (----) & (----) \\
    &       &  M$_3$  & 14.962 & ---- & 13.919 & ---- & ---- & ---- \\
    &       &         & (9.316) & (----) & (5.119) & (----) & (----) & (----) \\
    &       &  M$_4$  & 8.753 & ---- & 14.724 & ---- & ---- & ---- \\
    &       &         & (4.316) & (----) & (5.335) & (----) & (----) & (----) \\
    &       &  M$_5$  & 26.136 & ---- & 14.579 & 0.000 & 0.950 & 108.837 \\
    &       &         & (11.086) & (----) & (5.389) & (0.000) & (0.000) & (76.185) \\
    &       &  M$_6$  & 0.201 & 8.459 & 2.078 & 0.992 & 0.053 & 0.921 \\
    &       &         & (0.070) & (2.103) & (0.710) & (0.048) & (0.037) & (0.553) \\
    &       &  $\text{M}_6^{(\text{Opt})}$  & 0.194 & 8.229 & 2.008 & 0.995 & 0.051 & 0.913 \\
    &       &         & (0.068) & (2.100) & (0.706) & (0.044) & (0.034) & (0.547) \\
\cmidrule(l){2-9}
					
    & 250   &  M$_1$  & 33.171 & 21.388 & 9.028 & ---- & ---- & ---- \\
    &       &         & (22.859) & (9.893) & (9.391) & (----) & (----) & (----) \\
    &       &  M$_2$  & 0.094 & 7.116 & 0.774 & ---- & ---- & ---- \\
    &       &         & (0.009) & (0.498) & (0.189) & (----) & (----) & (----) \\
    &       &  M$_3$  & 9.209 & ---- & 11.593 & ---- & ---- & ---- \\
    &       &         & (4.996) & (----) & (4.073) & (----) & (----) & (----) \\
    &       &  M$_4$  & 4.976 & ---- & 12.234 & ---- & ---- & ---- \\
    &       &         & (2.039) & (----) & (4.302) & (----) & (----) & (----) \\
    &       &  M$_5$  & 20.121 & ---- & 12.055 & 0.000 & 0.950 & 84.579 \\
    &       &         & (6.749) & (----) & (4.279) & (0.000) & (0.000) & (59.331) \\
    &       &  M$_6$  & 0.261 & 11.565 & 2.629 & 0.969 & 0.066 & 1.595 \\
    &       &         & (0.108) & (2.386) & (1.017) & (0.082) & (0.059) & (3.102) \\
    &       &  $\text{M}_6^{(\text{Opt})}$  & 0.257 & 11.207 & 2.309 & 0.977 & 0.061 & 1.438 \\
    &       &         & (0.101) & (2.226) & (1.006) & (0.078) & (0.053) & (3.055) \\
\cmidrule(l){2-9}

    & 500   &  M$_1$  & 31.656 & 19.932 & 7.054 & ---- & ---- & ---- \\
    &       &         & (20.704) & (2.529) & (4.780) & (----) & (----) & (----) \\
    &       &  M$_2$  & 0.075 & 6.578 & 0.793 & ---- & ---- & ---- \\
    &       &         & (0.006) & (0.506) & (0.164) & (----) & (----) & (----) \\
    &       &  M$_3$  & 6.866 & ---- & 10.259 & ---- & ---- & ---- \\
    &       &         & (3.433) & (----) & (3.334) & (----) & (----) & (----) \\
    &       &  M$_4$  & 3.611 & ---- & 10.732 & ---- & ---- & ---- \\
    &       &         & (1.375) & (----) & (3.601) & (----) & (----) & (----) \\
    &       &  M$_5$  & 17.774 & ---- & 10.583 & 0.000 & 0.950 & 72.421 \\
    &       &         & (2.167) & (----) & (3.525) & (0.000) & (0.000) & (42.127) \\
    &       &  M$_6$  & 0.234 & 9.158 & 2.261 & 0.989 & 0.064 & 1.405 \\
    &       &         & (0.109) & (6.115) & (1.050) & (0.101) & (0.082) & (3.403) \\
    &       &  $\text{M}_6^{(\text{Opt})}$  & 0.222 & 8.981 & 2.188 & 0.993 & 0.059 & 1.378 \\
    &       &         & (0.096) & (6.102) & (1.003) & (0.099) & (0.084) & (3.307) \\
\bottomrule
\end{longtable}
\end{center}
\end{small}

Case 2 specifically examines the robustness and efficiency gains of quantile regression under challenging conditions. The simulation results presented in Table~\ref{tab:tab_2} definitively establish the superiority of the SFoF-QR ($M_6$) framework for coefficient recovery when the error distribution is asymmetric and contaminated. Examining estimation accuracy (RRISPEE($\beta$)) across increasing sample sizes shows that M$_6$ consistently outperforms M$_2$. The significant estimation error reduction provided by M$_6$ (up to $42\%$ more accurate than M$_2$ at $n=500, s_d=0.5$) confirms the theoretical advantage of quantile regression. Because the errors are contaminated with one-sided positive shocks that scale with the underlying signal, the conditional error distribution is right-skewed. The least squares criterion employed by M$_2$ minimizes the squared residual errors, which disproportionately weights large errors caused by the heteroscedasticity. This sensitivity pulls the fitted conditional mean function away from the true function, resulting in a biased estimate for $\widehat{\beta}(t,s)$ and thus a higher RRISPEE($\beta$). In contrast, M$_6$, minimizing the check loss, is intrinsically robust to heteroskedasticity and asymmetric tails, allowing it to estimate the true median function without bias, leading to far more accurate recovery of the underlying functional coefficient surfaces.

The superior coefficient estimation of M$_6$ translates directly into improved out-of-sample predictive performance (RMSPE). Under strong spatial dependence ($s_d=0.9$) and large sample size ($n=500$): M$_6$ achieved an RMSPE of $1.110$ and M$_2$ achieved an RMSPE of $1.340$. By accurately capturing the central tendency (median) of the conditional distribution, which is not unduly influenced by the asymmetric contamination, M$_6$ provides predictions that are closer to the robust center of the true conditional response, maintaining predictive consistency that M$_2$ loses due to bias toward the heavy tail. The failure of non-spatial models (M$_3$, M$_4$, M$_5$) remains pronounced. For $n=500, s_d=0.9$, M$_3$ yielded an RMSPE of $10.010$, and M$_5$ yielded $10.325$, confirming that spatial misspecification (ignoring $\rho(t,u)$) remains the dominant source of error, even when robustness measures are employed (as in M$_5$).

The performance of the prediction intervals constructed by M$_6$ confirms its utility for uncertainty quantification in heterogeneous environments. Even when the error variance is non-constant (heteroscedasticity) and dependent on the underlying signal magnitude, M$_6$ maintains robust coverage quality. For example, for $n=500, s_d=0.9$, the EC remained nearly perfect at $0.999$, with a CPD of $0.050$.

\begin{small}
\begin{center}
\tabcolsep 0.175in
\renewcommand{\arraystretch}{0.92}
\begin{longtable}{@{}ccccccccc@{}} 
\caption{Computed mean $\text{RRISPEE}(\widehat{\beta})$, $\text{RRISPEE}(\widehat{\rho})$, RMSPE, $\text{EC}$, $\text{CPD}$, and $\text{score}$ values with their standard errors (given in brackets) under Case 2. The results are obtained over 250 Monte-Carlo replications using three sample sizes ($n$) and three different strength of spatial dependence ($s_d$). Metrics EC, CPD, and score are not applicable for mean-regression models, and are thus presented as ----.}\label{tab:tab_2} \\
\toprule
{$s_d$} & {$n$} & Method & $\text{RRISPEE}(\widehat{\beta})$ & $\text{RRISPEE}(\widehat{\rho})$ & RMSPE & $\text{EC}$ & $\text{CPD}$ & $\text{score}$ \\ \midrule
\endfirsthead
\toprule
{$s_d$} & {$n$} & Method & $\text{RRISPEE}(\widehat{\beta})$ & $\text{RRISPEE}(\widehat{\rho})$ & RMSPE & $\text{EC}$ & $\text{CPD}$ & $\text{score}$ \\ \midrule
\endhead
\midrule
\multicolumn{9}{r}{Continued on next page} \\ 
\endfoot
\endlastfoot
0.1 & 100   &  M$_1$  & 29.768 & 96.111 & 5.606 & ---- & ---- & ---- \\
    &       &         & (19.418) & (249.357) & (1.558) & (----) & (----) & (--------) \\
    &       &  M$_2$  & 0.468 & 62.103 & 1.090 & ---- & ---- & ---- \\
    &       &         & (0.098) & (26.556) & (0.114) & (----) & (----) & (----) \\
    &       &  M$_3$  & 2.948 & ---- & 2.000 & ---- & ---- & ---- \\
    &       &         & (0.685) & (----) & (0.401) & (----) & (----) & (----) \\
    &       &  M$_4$  & 3.448 & ---- & 5.148 & ---- & ---- & ---- \\
    &       &         & (2.462) & (----) & (1.758) & (----) & (----) & (----) \\
    &       &  M$_5$  & 19.729 & ---- & 4.791 & 0.000 & 0.950 & 15.946 \\
    &       &         & (3.767) & (----) & (2.023) & (0.000) & (0.000) & (11.932) \\
    &       &  M$_6$  & 0.256 & 30.996 & 0.911 & 0.931 & 0.081 & 0.475 \\
    &       &         & (0.036) & (10.163) & (0.133) & (0.071) & (0.051) & (0.452) \\
    &       &  $\text{M}_6^{(\text{Opt})}$  & 0.233 & 28.106 & 0.841 & 0.951 & 0.072 & 0.411 \\
    &       &         & (0.022) & (8.252) & (0.111) & (0.068) & (0.047) & (0.444) \\
\cmidrule(l){2-9}				
    & 250   &  M$_1$  & 36.072 & 164.129 & 5.901 & ---- & ---- & ---- \\
    &       &         & (27.220) & (756.275) & (1.863) & (----) & (----) & (----) \\
    &       &  M$_2$  & 0.337 & 51.106 & 1.037 & ---- & ---- & ---- \\
    &       &         & (0.056) & (18.857) & (0.073) & (----) & (----) & (----) \\
    &       &  M$_3$  & 2.853 & ---- & 1.837 & ---- & ---- & ---- \\
    &       &         & (0.400) & (----) & (0.224) & (----) & (----) & (----) \\
    &       &  M$_4$  & 2.047 & ---- & 4.427 & ---- & ---- & ---- \\
    &       &         & (0.668) & (----) & (1.239) & (----) & (----) & (----) \\
    &       &  M$_5$  & 17.892 & ---- & 3.768 & 0.000 & 0.950 & 10.761 \\
    &       &         & (2.055) & (----) & (1.427) & (0.000) & (0.000) & (6.278) \\
    &       &  M$_6$  & 0.191 & 30.471 & 0.845 & 0.968 & 0.059 & 0.275 \\
    &       &         & (0.022) & (10.685) & (0.148) & (0.040) & (0.024) & (0.186) \\
    &       &  $\text{M}_6^{(\text{Opt})}$  & 0.167 & 25.551 & 0.778 & 0.972 & 0.052 & 0.233 \\
    &       &         & (0.020) & (8.554) & (0.131) & (0.035) & (0.026) & (0.177) \\
\cmidrule(l){2-9}
    & 500   &  M$_1$  & 32.678 & 119.643 & 5.812 & ---- & ---- & ---- \\
    &       &         & (20.090) & (347.512) & (2.021) & (----) & (----) & (----) \\
    &       &  M$_2$  & 0.253 & 60.353 & 0.998 & ---- & ---- & ---- \\
    &       &         & (0.050) & (30.114) & (0.065) & (----) & (----) & (----) \\
    &       &  M$_3$  & 2.821 & ---- & 1.784 & ---- & ---- & ---- \\
    &       &         & (0.260) & (----) & (0.195) & (----) & (----) & (----) \\
    &       &  M$_4$  & 1.723 & ---- & 3.983 & ---- & ---- & ---- \\
    &       &         & (0.306) & (----) & (1.052) & (----) & (----) & (----) \\
    &       &  M$_5$  & 17.330 & ---- & 3.257 & 0.000 & 0.950 & 9.204 \\
    &       &         & (1.419) & (----) & (1.179) & (0.000) & (0.000) & (5.226) \\
    &       &  M$_6$  & 0.145 & 32.619 & 0.871 & 0.976 & 0.055 & 0.235 \\
    &       &         & (0.018) & (14.310) & (0.209) & (0.031) & (0.018) & (0.092) \\
    &       &  $\text{M}_6^{(\text{Opt})}$  & 0.135 & 23.402 & 0.795 & 0.983 & 0.051 & 0.222 \\
    &       &         & (0.016) & (12.352) & (0.188) & (0.030) & (0.015) & (0.090) \\
\midrule

0.5 & 100   &  M$_1$  & 32.211 & 52.869 & 5.910 & ---- & ---- & ---- \\
    &       &         & (20.913) & (155.701) & (1.733) & (----) & (----) & (----) \\
    &       &  M$_2$  & 0.460 & 14.325 & 1.170 & ---- & ---- & ---- \\
    &       &         & (0.101) & (4.721) & (0.194) & (----) & (----) & (----) \\
    &       &  M$_3$  & 6.631 & ---- & 5.222 & ---- & ---- & ---- \\
    &       &         & (3.482) & (----) & (1.410) & (----) & (----) & (----) \\
    &       &  M$_4$  & 4.529 & ---- & 7.077 & ---- & ---- & ---- \\
    &       &         & (2.527) & (----) & (2.330) & (----) & (----) & (----) \\
    &       &  M$_5$  & 24.999 & ---- & 6.623 & 0.000 & 0.950 & 42.958 \\
    &       &         & (11.190) & (----) & (2.451) & (0.000) & (0.000) & (18.479) \\
    &       &  M$_6$  & 0.254 & 9.612 & 1.008 & 0.957 & 0.069 & 0.626 \\
    &       &         & (0.037) & (1.850) & (0.219) & (0.063) & (0.041) & (0.584) \\
    &       &  $\text{M}_6^{(\text{Opt})}$  & 0.221 & 8.404 & 1.001 & 0.967 & 0.063 & 0.611 \\
    &       &         & (0.033) & (1.724) & (0.211) & (0.063) & (0.040) & (0.556) \\
\cmidrule(l){2-9}
					
    & 250   &  M$_1$  & 35.768 & 29.065 & 5.901 & ---- & ---- & ---- \\
    &       &         & (26.837) & (40.640) & (1.599) & (----) & (----) & (----) \\
    &       &  M$_2$  & 0.336 & 12.527 & 1.082 & ---- & ---- & ---- \\
    &       &         & (0.055) & (4.171) & (0.146) & (----) & (----) & (----) \\
    &       &  M$_3$  & 4.500 & ---- & 4.663 & ---- & ---- & ---- \\
    &       &         & (1.920) & (----) & (0.852) & (----) & (----) & (----) \\
    &       &  M$_4$  & 2.575 & ---- & 5.925 & ---- & ---- & ---- \\
    &       &         & (0.792) & (----) & (1.620) & (----) & (----) & (----) \\
    &       &  M$_5$  & 18.294 & ---- & 5.553 & 0.000 & 0.950 & 39.791 \\
    &       &         & (3.342) & (----) & (1.576) & (0.000) & (0.000) & (18.067) \\
    &       &  M$_6$  & 0.190 & 9.261 & 0.944 & 0.980 & 0.056 & 0.422 \\
    &       &         & (0.021) & (2.272) & (0.246) & (0.038) & (0.022) & (0.275) \\
    &       &  $\text{M}_6^{(\text{Opt})}$  & 0.177 & 8.348 & 0.888 & 0.990 & 0.050 & 0.405 \\
    &       &         & (0.019) & (2.056) & (0.232) & (0.036) & (0.018) & (0.263) \\
\cmidrule(l){2-9}

    & 500   &  M$_1$  & 32.775 & 35.732 & 5.801 & ---- & ---- & ---- \\
    &       &         & (20.029) & (56.879) & (1.435) & (--------) & (----) & (----) \\
    &       &  M$_2$  & 0.252 & 13.858 & 1.020 & ---- & ---- & ---- \\
    &       &         & (0.047) & (6.471) & (0.111) & (----) & (----) & (----) \\
    &       &  M$_3$  & 3.670 & ---- & 4.451 & ---- & ---- & ---- \\
    &       &         & (1.287) & (----) & (0.702) & (----) & (----) & (----) \\
    &       &  M$_4$  & 2.019 & ---- & 5.421 & ---- & ---- & ---- \\
    &       &         & (0.400) & (----) & (1.328) & (----) & (----) & (----) \\
    &       &  M$_5$  & 17.375 & ---- & 5.045 & 0.000 & 0.950 & 32.446 \\
    &       &         & (1.419) & (----) & (1.224) & (0.000) & (0.000) & (17.241) \\
    &       &  M$_6$  & 0.146 & 9.403 & 1.037 & 0.988 & 0.052 & 0.365 \\
    &       &         & (0.018) & (2.847) & (0.286) & (0.023) & (0.010) & (0.093) \\
    &       &  $\text{M}_6^{(\text{Opt})}$  & 0.133 & 8.112 & 1.002 & 0.993 & 0.050 & 0.333 \\
    &       &         & (0.013) & (2.609) & (0.284) & (0.019) & (0.009) & (0.094) \\
\midrule

0.9 & 100   &  M$_1$  & 83.165 & 40.747 & 12.268 & ---- & ---- & ---- \\
    &       &         & (522.018) & (181.636) & (31.114) & (----) & (----) & (----) \\
    &       &  M$_2$  & 0.465 & 10.145 & 2.168 & ---- & ---- & ---- \\
    &       &         & (0.099) & (2.358) & (1.148) & (----) & (----) & (----) \\
    &       &  M$_3$  & 15.783 & ---- & 13.325 & ---- & ---- & ---- \\
    &       &         & (9.311) & (----) & (5.575) & (----) & (----) & (----) \\
    &       &  M$_4$  & 8.342 & ---- & 14.135 & ---- & ---- & ---- \\
    &       &         & (4.164) & (----) & (5.752) & (----) & (----) & (----) \\
    &       &  M$_5$  & 31.919 & ---- & 13.938 & 0.000 & 0.950 & 103.405 \\
    &       &         & (12.744) & (----) & (5.914) & (0.000) & (0.000) & (72.741) \\
    &       &  M$_6$  & 0.278 & 8.482 & 2.107 & 0.948 & 0.080 & 4.072 \\
    &       &         & (0.048) & (1.514) & (0.547) & (0.142) & (0.119) & (9.677) \\
    &       &  $\text{M}_6^{(\text{Opt})}$  & 0.261 & 7.883 & 2.044 & 0.970 & 0.068 & 3.852 \\
    &       &         & (0.044) & (1.489) & (0.533) & (0.133) & (0.101) & (7.023) \\
\cmidrule(l){2-9}
					
    & 250   &  M$_1$  & 35.997 & 20.690 & 8.546 & ---- & ---- & ---- \\
    &       &         & (26.348) & (4.435) & (8.030) & (----) & (----) & (----) \\
    &       &  M$_2$  & 0.337 & 9.320 & 1.659 & ---- & ---- & ---- \\
    &       &         & (0.057) & (2.430) & (0.648) & (----) & (----) & (----) \\
    &       &  M$_3$  & 9.304 & ---- & 11.344 & ---- & ---- & ---- \\
    &       &         & (5.140) & (----) & (4.041) & (----) & (----) & (----) \\
    &       &  M$_4$  & 4.775 & ---- & 11.978 & ---- & ---- & ---- \\
    &       &         & (1.956) & (----) & (4.280) & (----) & (----) & (----) \\
    &       &  M$_5$  & 21.629 & ---- & 11.736 & 0.000 & 0.950 & 85.881 \\
    &       &         & (8.729) & (----) & (4.287) & (0.000) & (0.000) & (56.761) \\
    &       &  M$_6$  & 0.293 & 8.740 & 1.637 & 0.973 & 0.062 & 2.519 \\
    &       &         & (0.098) & (2.648) & (0.618) & (0.086) & (0.065) & (2.597) \\
    &       &  $\text{M}_6^{(\text{Opt})}$  & 0.281 & 8.455 & 1.589 & 0.988 & 0.058 & 2.499 \\
    &       &         & (0.091) & (2.501) & (0.601) & (0.080) & (0.061) & (2.588) \\
\cmidrule(l){2-9}

    & 500   &  M$_1$  & 32.517 & 20.269 & 7.142 & ---- & ---- & ---- \\
    &       &         & (19.846) & (5.457) & (4.498) & (----) & (----) & (----) \\
    &       &  M$_2$  & 0.254 & 9.341 & 1.340 & ---- & ---- & ---- \\
    &       &         & (0.050) & (3.180) & (0.567) & (----) & (----) & (----) \\
    &       &  M$_3$  & 6.642 & ---- & 10.010 & ---- & ---- & ---- \\
    &       &         & (3.094) & (----) & (3.429) & (----) & (----) & (----) \\
    &       &  M$_4$  & 3.449 & ---- & 10.459 & ---- & ---- & ---- \\
    &       &         & (1.218) & (----) & (3.678) & (----) & (----) & (----) \\
    &       &  M$_5$  & 17.684 & ---- & 10.325 & 0.000 & 0.950 & 71.021 \\
    &       &         & (1.503) & (----) & (3.570) & (0.000) & (0.000) & (42.721) \\
    &       &  M$_6$  & 0.199 & 7.872 & 1.110 & 0.999 & 0.050 & 1.945 \\
    &       &         & (0.106) & (3.418) & (0.754) & (0.099) & (0.051) & (2.241) \\
    &       &  $\text{M}_6^{(\text{Opt})}$  & 0.191 & 7.637 & 1.077 & 0.999 & 0.050 & 1.938 \\
    &       &         & (0.101) & (3.331) & (0.744) & (0.091) & (0.047) & (2.133) \\
\bottomrule
\end{longtable}
\end{center}
\end{small}

Case 3 assesses the performance of the models under leptokurtosis (heavy tails). The results presented in Table~\ref{tab:tab_3} confirm the theoretical efficiency gains of M$_6$ in the presence of heavy tails: $M_6$ consistently demonstrates superior accuracy in coefficient recovery compared to M$_2$. The $L_2$ norm used by M$_2$ (least squares) is highly sensitive to the large residuals generated by the heavy tails of the Laplace distribution, leading to inflated variance in its estimates and thus higher RRISPEE. The $L_1$ norm derived from the check loss function (used by M$_6$) minimizes absolute errors, which is the optimal criterion for leptokurtic data. This intrinsic alignment between the statistical loss function and the data generation process grants M$_6$ greater asymptotic efficiency, leading to demonstrably lower estimation error for both $\widehat{\beta}(t,s)$ and $\widehat{\rho}(t,u)$.

While M$_6$ is statistically more efficient for coefficient estimation, both M$_6$ and M$_2$ show highly competitive RMSPE values for point prediction of the central tendency, particularly as sample size $n$ increases. The observation that M$_2$ remains competitive for point prediction despite being less asymptotically efficient for coefficient recovery suggests that for prediction tasks specifically targeting the central tendency under symmetric distributions, the least squares procedure benefits from variance reduction that slightly compensates for its theoretical sub-optimality compared to the median estimator when $n$ is sufficient. Nevertheless, both M$_2$ and M$_6$ maintain predictive errors far below the non-spatial models (e.g., M$_3$ at $10.588$ and M$_5$ at $10.933$).

Under heavy-tailed errors, the prediction intervals constructed by M$_6$ remain highly reliable, demonstrating near-perfect EC across all tested scenarios (EC $\approx 0.999-1.000$). For example, for $n=500, s_d=0.9$, the EC was $1.000$ (CPD $0.050$), with a score of $1.666$.

\begin{small}
\begin{center}
\tabcolsep 0.175in
\renewcommand{\arraystretch}{0.92}
\begin{longtable}{@{}ccccccccc@{}} 
\caption{Computed mean $\text{RRISPEE}(\widehat{\beta})$, $\text{RRISPEE}(\widehat{\rho})$, RMSPE, $\text{EC}$, $\text{CPD}$, and $\text{score}$ values with their standard errors (given in brackets) under Case 3. The results are obtained over 250 Monte-Carlo replications using three sample sizes ($n$) and three different strength of spatial dependence ($s_d$). Metrics EC, CPD, and score are not applicable for mean-regression models, and are thus presented as ----.}\label{tab:tab_3} \\
\toprule
{$s_d$} & {$n$} & Method & $\text{RRISPEE}(\widehat{\beta})$ & $\text{RRISPEE}(\widehat{\rho})$ & RMSPE & $\text{EC}$ & $\text{CPD}$ & $\text{score}$ \\ \midrule
\endfirsthead
\toprule
{$s_d$} & {$n$} & Method & $\text{RRISPEE}(\widehat{\beta})$ & $\text{RRISPEE}(\widehat{\rho})$ & RMSPE & $\text{EC}$ & $\text{CPD}$ & $\text{score}$ \\ \midrule
\endhead
\midrule
\multicolumn{9}{r}{Continued on next page} \\ 
\endfoot
\endlastfoot
0.1 & 100   &  M$_1$  & 32.366 & 255.993 & 5.978 & ---- & ---- & ---- \\
    &       &         & (48.229) & (2642.872) & (8.879) & (----) & (----) & (----) \\
    &       &  M$_2$  & 0.272 & 29.138 & 0.484 & ---- & ---- & ---- \\
    &       &         & (0.032) & (9.110) & (0.034) & (----) & (----) & (----) \\
    &       &  M$_3$  & 2.907 & ---- & 1.988 & ---- & ---- & ---- \\
    &       &         & (0.684) & (----) & (0.337) & (----) & (----) & (----) \\
    &       &  M$_4$  & 3.387 & ---- & 5.310 & ---- & ---- & ---- \\
    &       &         & (2.119) & (----) & (1.666) & (----) & (----) & (----) \\
    &       &  M$_5$  & 19.444 & ---- & 4.894 & 0.000 & 0.950 & 17.378 \\
    &       &         & (3.824) & (----) & (1.896) & (0.000) & (0.000) & (13.086) \\
    &       &  M$_6$  & 0.242 & 27.949 & 0.780 & 0.994 & 0.048 & 0.156 \\
    &       &         & (0.028) & (8.947) & (0.131) & (0.008) & (0.003) & (0.009) \\
    &       &  $\text{M}_6^{(\text{Opt})}$  & 0.211 & 21.602 & 0.699 & 0.998 & 0.040 & 0.131 \\
    &       &         & (0.021) & (7.665) & (0.122) & (0.005) & (0.002) & (0.005) \\
\cmidrule(l){2-9}				
    & 250   &  M$_1$  & 37.912 & 204.875 & 5.939 & ---- & ---- & ---- \\
    &       &         & (65.486) & (1170.712) & (2.494) & (----) & (----) & (----) \\
    &       &  M$_2$  & 0.214 & 28.609 & 0.390 & ---- & ---- & ---- \\
    &       &         & (0.023) & (8.639) & (0.022) & (----) & (----) & (----) \\
    &       &  M$_3$  & 2.876 & ---- & 1.840 & ---- & ---- & ---- \\
    &       &         & (0.385) & (----) & (0.228) & (----) & (----) & (----) \\
    &       &  M$_4$  & 2.140 & ---- & 4.480 & ---- & ---- & ---- \\
    &       &         & (0.834) & (----) & (1.326) & (----) & (----) & (----) \\
    &       &  M$_5$  & 17.776 & ---- & 3.973 & 0.000 & 0.950 & 11.459 \\
    &       &         & (2.052) & (----) & (1.508) & (0.000) & (0.000) & (6.904) \\
    &       &  M$_6$  & 0.190 & 26.481 & 0.645 & 0.998 & 0.049 & 0.152 \\
    &       &         & (0.020) & (8.356) & (0.160) & (0.002) & (0.001) & (0.001) \\
    &       &  $\text{M}_6^{(\text{Opt})}$  & 0.181 & 24.388 & 0.558 & 0.999 & 0.047 & 0.149 \\
    &       &         & (0.018) & (7.881) & (0.158) & (0.001) & (0.001) & (0.001) \\
\cmidrule(l){2-9}
    & 500   &  M$_1$  & 31.361 & 142.312 & 5.547 & ---- & ---- & ---- \\
    &       &         & (20.698) & (421.357) & (1.403) & (----) & (----) & (----) \\
    &       &  M$_2$  & 0.173 & 31.245 & 0.332 & ---- & ---- & ---- \\
    &       &         & (0.020) & (11.165) & (0.017) & (----) & (----) & (----) \\
    &       &  M$_3$  & 2.852 & ---- & 1.763 & ---- & ---- & ---- \\
    &       &         & (0.257) & (----) & (0.163) & (----) & (----) & (----) \\
    &       &  M$_4$  & 1.735 & ---- & 4.065 & ---- & ---- & ---- \\
    &       &         & (0.326) & (----) & (1.041) & (----) & (----) & (----) \\
    &       &  M$_5$  & 17.277 & ---- & 3.289 & 0.000 & 0.950 & 9.914 \\
    &       &         & (1.334) & (----) & (1.192) & (0.000) & (0.000) & (0.000) \\
    &       &  M$_6$  & 0.131 & 24.658 & 0.508 & 0.999 & 0.049 & 0.152 \\
    &       &         & (0.011) & (10.343) & (0.142) & (0.001) & (0.001) & (0.001) \\
    &       &  $\text{M}_6^{(\text{Opt})}$  & 0.128 & 22.441 & 0.489 & 0.999 & 0.048 & 0.150 \\
    &       &         & (0.010) & (10.002) & (0.137) & (0.001) & (0.001) & (0.001) \\
\midrule

0.5 & 100   &  M$_1$  & 42.521 & 159.902 & 7.941 & ---- & ---- & ---- \\
    &       &         & (141.507) & (1317.453) & (17.312) & (----) & (----) & (----) \\
    &       &  M$_2$  & 0.273 & 9.291 & 0.519 & ---- & ---- & ---- \\
    &       &         & (0.033) & (1.632) & (0.062) & (----) & (----) & (----) \\
    &       &  M$_3$  & 6.335 & ---- & 5.354 & ---- & ---- & ---- \\
    &       &         & (3.507) & (----) & (1.300) & (----) & (----) & (----) \\
    &       &  M$_4$  & 4.757 & ---- & 7.207 & ---- & ---- & ---- \\
    &       &         & (2.385) & (----) & (2.227) & (----) & (----) & (----) \\
    &       &  M$_5$  & 22.683 & ---- & 6.874 & 0.000 & 0.950 & 44.757 \\
    &       &         & (9.693) & (----) & (2.295) & (0.000) & (0.000) & (20.311) \\
    &       &  M$_6$  & 0.244 & 8.968 & 0.933 & 0.998 & 0.049 & 0.272 \\
    &       &         & (0.028) & (1.663) & (0.200) & (0.004) & (0.001) & (0.006) \\
    &       &  $\text{M}_6^{(\text{Opt})}$  & 0.202 & 8.445 & 0.901 & 0.999 & 0.047 & 0.255 \\
    &       &         & (0.021) & (1.543) & (0.189) & (0.002) & (0.001) & (0.005) \\
\cmidrule(l){2-9}
					
    & 250   &  M$_1$  & 36.647 & 36.046 & 6.413 & ---- & ---- & ---- \\
    &       &         & (43.469) & (68.345) & (8.081) & (----) & (----) & (----) \\
    &       &  M$_2$  & 0.215 & 8.800 & 0.413 & ---- & ---- & ---- \\
    &       &         & (0.022) & (1.713) & (0.035) & (----) & (----) & (----) \\
    &       &  M$_3$  & 4.550 & ---- & 4.751 & ---- & ---- & ---- \\
    &       &         & (1.953) & (----) & (0.927) & (----) & (----) & (----) \\
    &       &  M$_4$  & 2.686 & ---- & 6.033 & ---- & ---- & ---- \\
    &       &         & (0.930) & (----) & (1.733) & (----) & (----) & (----) \\
    &       &  M$_5$  & 17.782 & ---- & 5.737 & 0.000 & 0.950 & 40.569 \\
    &       &         & (2.050) & (----) & (1.662) & (0.000) & (0.000) & (17.649) \\
    &       &  M$_6$  & 0.192 & 7.624 & 0.683 & 0.999 & 0.049 & 0.272 \\
    &       &         & (0.020) & (1.564) & (0.114) & (0.001) & (0.001) & (0.00) \\
    &       &  $\text{M}_6^{(\text{Opt})}$  & 0.178 & 6.774 & 0.644 & 0.999 & 0.047 & 0.253 \\
    &       &         & (0.017) & (1.434) & (0.111) & (0.001) & (0.001) & (0.00) \\
\cmidrule(l){2-9}

    & 500   &  M$_1$  & 32.216 & 64.847 & 6.123 & ---- & ---- & ---- \\
    &       &         & (23.586) & (209.425) & (4.946) & (----) & (----) & (----) \\
    &       &  M$_2$  & 0.173 & 9.004 & 0.358 & ---- & ---- & ---- \\
    &       &         & (0.020) & (2.198) & (0.027) & (----) & (----) & (----) \\
    &       &  M$_3$  & 3.715 & ---- & 4.494 & ---- & ---- & ---- \\
    &       &         & (1.345) & (----) & (0.703) & (----) & (----) & (----) \\
    &       &  M$_4$  & 2.073 & ---- & 5.572 & ---- & ---- & ---- \\
    &       &         & (0.453) & (----) & (1.372) & (----) & (----) & (----) \\
    &       &  M$_5$  & 17.274 & ---- & 5.15 & 0.000 & 0.950 & 32.508 \\
    &       &         & (1.346) & (----) & (1.219) & (0.000) & (0.000) & (16.351) \\
    &       &  M$_6$  & 0.132 & 7.252 & 0.487 & 0.999 & 0.049 & 0.272 \\
    &       &         & (0.017) & (2.190) & (0.151) & (0.001) & (0.001) & (0.001) \\
    &       &  $\text{M}_6^{(\text{Opt})}$  & 0.119 & 6.335 & 0.451 & 0.999 & 0.047 & 0.251 \\
    &       &         & (0.011) & (2.005) & (0.147) & (0.001) & (0.001) & (0.001) \\
\midrule

0.9 & 100   &  M$_1$  & 32.427 & 25.650 & 10.512 & ---- & ---- & ---- \\
    &       &         & (22.769) & (21.546) & (11.791) & (----) & (----) & (----) \\
    &       &  M$_2$  & 0.273 & 8.052 & 1.160 & ---- & ---- & ---- \\
    &       &         & (0.032) & (0.934) & (0.518) & (----) & (----) & (----) \\
    &       &  M$_3$  & 14.571 & ---- & 14.099 & ---- & ---- & ---- \\
    &       &         & (9.279) & (----) & (5.546) & (----) & (----) & (----) \\
    &       &  M$_4$  & 8.720 & ---- & 14.893 & ---- & ---- & ---- \\
    &       &         & (4.397) & (----) & (5.825) & (----) & (----) & (----) \\
    &       &  M$_5$  & 27.936 & ---- & 14.740 & 0.000 & 0.950 & 114.015 \\
    &       &         & (12.409) & (----) & (5.868) & (0.000) & (0.000) & (78.309) \\
    &       &  M$_6$  & 0.268 & 8.319 & 1.987 & 0.999 & 0.049 & 1.402 \\
    &       &         & (0.038) & (1.843) & (0.709) & (0.005) & (0.001) & (0.128) \\
    &       &  $\text{M}_6^{(\text{Opt})}$  & 0.237 & 7.668 & 1.801 & 0.999 & 0.047 & 1.333 \\
    &       &         & (0.035) & (1.756) & (0.687) & (0.004) & (0.001) & (0.128) \\
\cmidrule(l){2-9}
					
    & 250   &  M$_1$  & 32.945 & 21.403 & 8.428 & ---- & ---- & ---- \\
    &       &         & (22.750) & (13.498) & (7.715) & (----) & (----) & (----) \\
    &       &  M$_2$  & 0.214 & 7.495 & 0.895 & ---- & ---- & ---- \\
    &       &         & (0.022) & (1.049) & (0.272) & (----) & (----) & (----) \\
    &       &  M$_3$  & 9.249 & ---- & 11.553 & ---- & ---- & ---- \\
    &       &         & (5.106) & (----) & (4.232) & (----) & (----) & (----) \\
    &       &  M$_4$  & 5.012 & ---- & 12.138 & ---- & ---- & ---- \\
    &       &         & (2.081) & (----) & (4.509) & (----) & (----) & (----) \\
    &       &  M$_5$  & 21.270 & ---- & 12.009 & 0.000 & 0.950 & 87.717 \\
    &       &         & (8.611) & (----) & (4.412) & (0.000) & (0.000) & (60.193) \\
    &       &  M$_6$  & 0.213 & 6.988 & 1.546 & 0.993 & 0.052 & 1.649 \\
    &       &         & (0.091) & (0.970) & (0.964) & (0.031) & (0.017) & (0.788) \\
    &       &  $\text{M}_6^{(\text{Opt})}$  & 0.189 & 6.641 & 1.381 & 0.995 & 0.050 & 1.511 \\
    &       &         & (0.087) & (0.881) & (0.957) & (0.030) & (0.012) & (0.655) \\
\cmidrule(l){2-9}

    & 500   &  M$_1$  & 31.317 & 20.001 & 7.067 & ---- & ---- & ---- \\
    &       &         & (18.984) & (3.937) & (4.893) & (----) & (----) & (----) \\
    &       &  M$_2$  & 0.172 & 7.132 & 0.874 & ---- & ---- & ---- \\
    &       &         & (0.020) & (1.185) & (0.241) & (----) & (----) & (----) \\
    &       &  M$_3$  & 7.038 & ---- & 10.588 & ---- & ---- & ---- \\
    &       &         & (3.637) & (----) & (3.437) & (----) & (----) & (----) \\
    &       &  M$_4$  & 3.694 & ---- & 11.095 & ---- & ---- & ---- \\
    &       &         & (1.405) & (----) & (3.714) & (----) & (----) & (----) \\
    &       &  M$_5$  & 17.722 & ---- & 10.933 & 0.000 & 0.950 & 76.371 \\
    &       &         & (1.380) & (----) & (3.583) & (0.000) & (0.000) & (49.109) \\
    &       &  M$_6$  & 0.102 & 6.314 & 0.971 & 1.000 & 0.050 & 1.666 \\
    &       &         & (0.024) & (1.542) & (0.712) & (0.085) & (0.061) & (0.801) \\
    &       &  $\text{M}_6^{(\text{Opt})}$  & 0.099 & 6.031 & 0.952 & 1.000 & 0.050 & 1.533 \\
    &       &         & (0.021) & (1.432) & (0.705) & (0.081) & (0.064) & (0.805) \\
\bottomrule
\end{longtable}
\end{center}
\end{small}

The most striking conclusion from the Monte Carlo experiments is the catastrophic failure of all non-spatial functional regression models (M$_3$, M$_4$, M$_5$) when spatial dependence ($s_d$) is moderate or strong. This failure is quantified by comparing predictive performance across models when $n=500$ and $s_d=0.9$, see, for example, Table~\ref{tab:pretab_2}. Classical penalized function-on-function regression models, represented by M$_3$, and non-spatial quantile regression models (M$_5$) operate under the assumption that observations are independent across locations. In a functional SAR data generating process, the response $\Y_{i}(t)$ is causally dependent on the neighboring responses $\widetilde{\Y}_{i}(u)$ through the spatial kernel $\rho(t,u)$. By omitting this term, non-spatial models effectively absorb this causal spatial effect into the error term $\epsilon_i(t)$. This critical misspecification violates the strict exogeneity assumption, as the predictor $\X_i(s)$ and the omitted spatial lag term are still related through the spatial weights matrix $\bm{W}$. Consequently, the non-spatial models suffer catastrophic bias and inconsistency in the estimated coefficient surface $\widehat{\beta}(t,s)$, rendering them entirely unsuitable for spatial functional data analysis, irrespective of their specific target (mean or quantile) or regularization method (B-spline or boosting). This structural inadequacy validates the necessity of the functional SAR framework used by both M$_2$ and M$_6$, and explains why they maintain stable, low predictive error even under strong spatial dependence ($s_d=0.9$).

\begin{small}
\begin{center}
\tabcolsep 0.255in
\renewcommand{\arraystretch}{0.92}
\begin{longtable}{@{}lcccccc@{}}
\caption{Predictive performance (RMSPE) and the failure of non-spatial models ($n=500, s_d=0.9$). \% denotes the improvements.}\label{tab:pretab_2} \\
\toprule
Error Case & M$_6$ & $M_{2}$ & $M_{3}$ & $M_{5}$ & \% (M$_6$ over M$_3$) & \% (M$_6$ over M$_5$) \\
\midrule
\endfirsthead
\toprule
Error Case & M$_6$ & $M_{2}$ & $M_{3}$ & $M_{5}$ & \% (M$_6$ over M$_3$) & \% (M$_6$ over M$_5$) \\
\midrule
\endhead
\midrule
\multicolumn{7}{r}{Continued on next page} \\
\endfoot
\endlastfoot
Case~1 & 2.261 & 0.793 & 10.259 & 10.583 & $77.9\%$ & $78.6\%$ \\
Case~2 & 1.110 & 1.340 & 10.010 & 10.325 & $88.9\%$ & $89.2\%$ \\
Case~3 & 0.971 & 0.874 & 10.588 & 10.933 & $90.8\%$ & $91.1\%$ \\
\bottomrule
\end{longtable}
\end{center}
\end{small}

Finally, we evaluate the efficacy of the proposed data-driven smoothing parameter selection via the BIC, denoted as $\text{M}_6^{(\text{Opt})}$. The results, presented in the last rows of Tables~\ref{tab:tab_1}-\ref{tab:tab_3}, demonstrate that optimizing $(\lambda_\rho, \lambda_\beta)$ yields consistent improvements over the fixed-parameter specification (M$_6$) across almost all scenarios. This advantage is most pronounced in challenging regimes characterized by complex error structures and limited sample sizes. For instance, in Case 2 (heteroscedastic errors) with strong spatial dependence ($s_d = 0.9$) and $n=100$, the BIC-optimized model reduces the RRISPEE for $\beta(t,s)$ to 0.261, compared to 0.278 for the fixed-parameter model. These findings underscore the practical value of the BIC in the context of proposed approach. While fixed parameters may suffice for simple Gaussian errors, they lack the flexibility to adapt to varying signal-to-noise ratios found in heteroscedastic or heavy-tailed data. The BIC effectively navigates this bias-variance trade-off, penalizing model complexity to prevent the overfitting of spatial noise while retaining sufficient flexibility to capture the structural functional relationship. Consequently, the BIC-guided approach proves to be a robust and valuable mechanism for automated estimation, ensuring superior accuracy without reliance on arbitrary manual tuning.

\subsection{Coverage of coefficient intervals}

To assess the finite-sample behavior of uncertainty quantification for the coefficient surfaces, we conducted an additional coverage experiment for the proposed method M$_6$ in one representative simulation setting. We considered Case~2 with $n=500$, moderate spatial dependence $s_d=0.5$, and the median surface $\tau=0.5$. Since the current implementation returns the estimated coefficient surfaces but not closed-form pointwise standard-error surfaces, we estimated the pointwise sampling variability of $\widehat{\beta}_{0.5}(t,s)$ and $\widehat{\rho}_{0.5}(t,u)$ from the Monte Carlo replications and used normal intervals of the form
\begin{equation*}
\widehat{\beta}_{0.5}(t,s)\pm 1.96 \widehat{\mathrm{sd}}\{\widehat{\beta}_{0.5}(t,s)\}, \qquad \widehat{\rho}_{0.5}(t,u)\pm 1.96 \widehat{\mathrm{sd}}\{\widehat{\rho}_{0.5}(t,u)\}.
\end{equation*}
Coverage was evaluated pointwise over the evaluation grid and then averaged over grid points and Monte Carlo replications. Specifically,
\begin{equation*}
\mathrm{Cov}_{\beta} = \frac{1}{N_{\mathrm{sim}}RT_x} \sum_{b=1}^{N_{\mathrm{sim}}} \sum_{r=1}^{R} \sum_{g=1}^{T_x} \mathbf{1} \left[ \beta(t_r,s_g) \in \mathrm{CI}^{(b)}_{\beta}(t_r,s_g) \right], \end{equation*}
and
\begin{equation*}
\mathrm{Cov}_{\rho} = \frac{1}{N_{\mathrm{sim}}R^2} \sum_{b=1}^{N_{\mathrm{sim}}} \sum_{r=1}^{R} \sum_{q=1}^{R} \mathbf{1} \left[ \rho(t_r,u_q) \in \mathrm{CI}^{(b)}_{\rho}(t_r,u_q) \right].
\end{equation*}
The results in Table~\ref{tab:coef_coverage} show that the pointwise intervals for $\beta_{0.5}(t,s)$ attain coverage very close to the 95\% level. The coverage for $\rho_{0.5}(t,u)$ is lower, indicating that uncertainty for the spatial spillover surface is harder to quantify in finite samples, especially because this component is estimated through the endogenous spatial-lag block and the first-stage IV projection. Thus, the additional experiment supports the finite-sample reliability of the proposed uncertainty quantification for the regression surface $\beta_\tau$, while also revealing that coefficient intervals for the spatial dependence surface $\rho_\tau$ can be anti-conservative in moderate samples.

\begin{table}[!ht]
\centering
\caption{Empirical pointwise coverage of 95\% coefficient intervals for the coefficient surfaces of M$_6$ under Case~2 with $n=500$, $s_d=0.5$, and $\tau=0.5$. Coverage and average interval width are averaged over the evaluation grid and Monte Carlo replications. Standard deviations are given in parentheses.}
\label{tab:coef_coverage}
\begin{tabular}{lcc}
\toprule
Coefficient surface & Empirical coverage & Average interval width \\
\midrule
$\beta_{0.5}(t,s)$ & 0.949 (0.219) & 0.017 (0.007) \\
$\rho_{0.5}(t,u)$ & 0.819 (0.358) & 0.094 (0.035) \\
\bottomrule
\end{tabular}
\end{table}

\subsection{Computing time}\label{subsec:computing-time}

Table~\ref{tab:comp_time_case2} reports representative wall-clock computing times for one Monte Carlo replication under the heteroskedastic upper-tail contamination design (Case~2). All computations were carried out on a DELL workstation equipped with a 13th Gen Intel(R) Core(TM) i9-13900HX CPU (2.20 GHz) and 64.0 GB RAM. The reported values are in seconds and correspond to the implementation used in the simulation study.

\begin{table}[!htb]
\centering
\caption{Representative computing time in seconds for one Monte Carlo replication under Case~2.}
\label{tab:comp_time_case2}
\setlength{\tabcolsep}{0.22in}
\renewcommand{\arraystretch}{0.90}
\scriptsize
\begin{tabular}{@{}lccccccccccc@{}}
\toprule
& \multicolumn{3}{c}{$s_d=0.1$}
&& \multicolumn{3}{c}{$s_d=0.5$}
&& \multicolumn{3}{c}{$s_d=0.9$} \\
\cmidrule(lr){2-4} \cmidrule(lr){6-8} \cmidrule(lr){10-12}
& $n=100$ & $250$ & $500$
&& $100$ & $250$ & $500$
&& $100$ & $250$ & $500$ \\
\midrule
M$_1$ & 0.22 & 0.27 & 1.63 && 0.20 & 0.22 & 0.63 && 0.22 & 0.24 & 1.00 \\
M$_2$ & 1.57 & 2.94 & 5.71 && 1.00 & 2.45 & 5.64 && 1.00 & 3.03 & 5.73 \\
M$_3$ & 2.54 & 8.70 & 9.69 && 1.98 & 7.73 & 9.45 && 2.72 & 6.22 & 12.14 \\
M$_4$ & 0.57 & 0.81 & 1.34 && 0.41 & 0.82 & 1.53 && 0.50 & 0.92 & 1.51 \\
M$_5$ & 0.39 & 2.18 & 3.81 && 1.22 & 2.16 & 3.65 && 1.21 & 2.08 & 3.51 \\
M$_6$ & 4.31 & 10.72 & 27.39 && 5.13 & 20.25 & 41.73 && 8.72 & 22.83 & 54.34 \\
\bottomrule
\end{tabular}
\end{table}

As expected, the proposed SFoF-QR estimator M$_6$ is computationally more demanding than the mean-based and non-spatial competitors because it combines a first-stage instrumental quantile fit, tensor-product spline bases, smoothing-parameter selection, and L-BFGS-B iterations for the smoothed quantile objective. Nevertheless, the computing times remain moderate in the simulation settings: even for $n=500$, one complete Case~2 replication takes less than one minute for M$_6$ on the workstation described above. The increase in runtime with $n$ and $s_d$ reflects the larger Stage~1 and Stage~2 design matrices, the greater cost of repeated objective and gradient evaluations, and the stronger spatial feedback handled by the two-stage estimator.

\section{Additional results for the air-quality application}
\label{sec:supp-air-quality}

This section provides additional graphical results from the
PM$_{2.5}$/PM$_{10}$ air-quality application. These displays supplement the numerical summaries and the selected coefficient-surface comparisons reported in the main manuscript.

\subsection{Fitted-curve comparison across methods}

Figure~\ref{fig:S_app_methods} presents the complete fitted-curve comparison for the six methods considered in the empirical analysis. The graphical findings are consistent with the fitted-performance measures reported in Table~6 of the main manuscript.

\begin{figure}[!htbp]
\centering
\includegraphics[width=5.05cm]{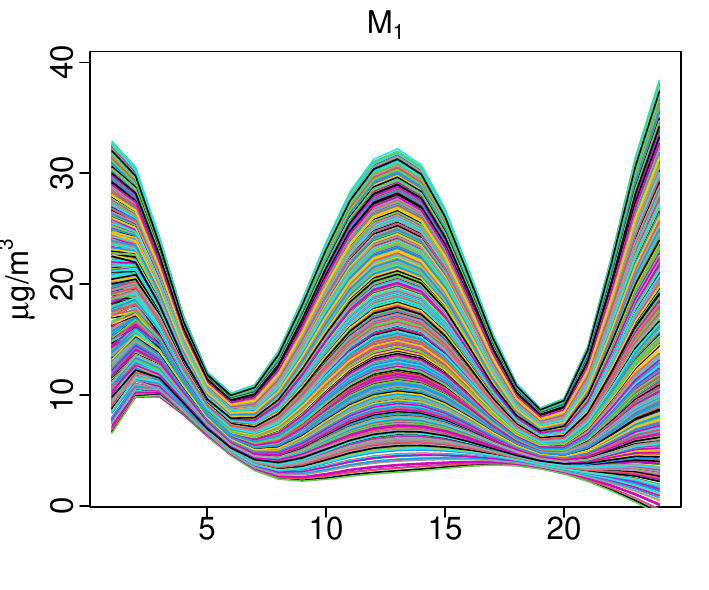}
\quad
\includegraphics[width=5.05cm]{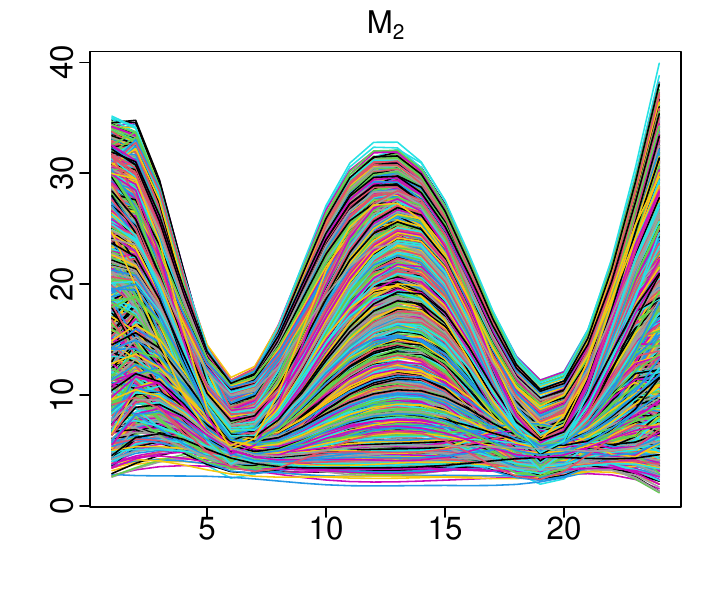}
\quad
\includegraphics[width=5.05cm]{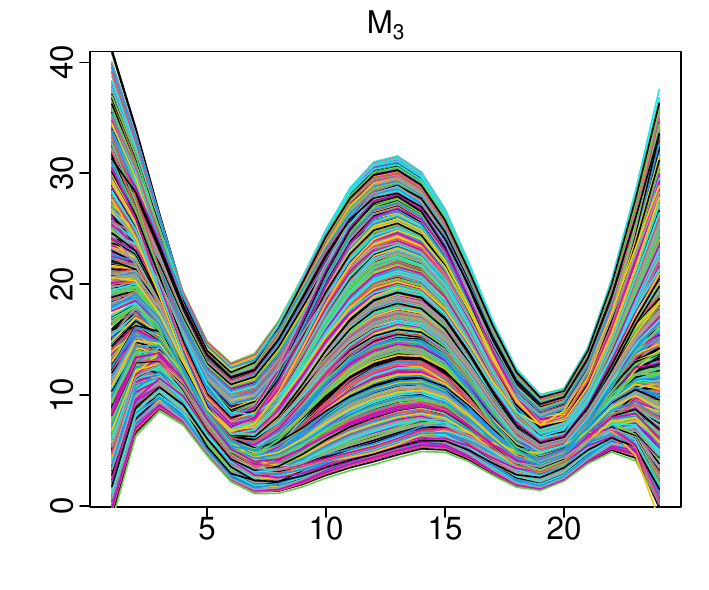}
\\
\includegraphics[width=5.05cm]{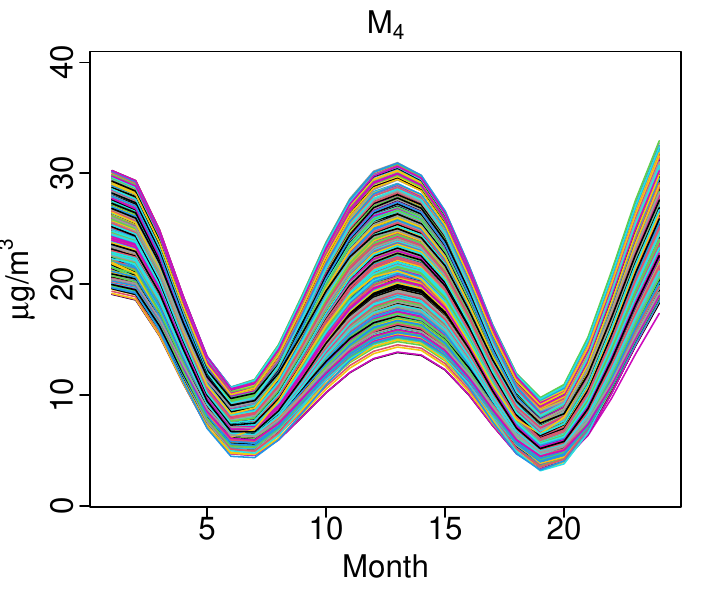}
\quad
\includegraphics[width=5.05cm]{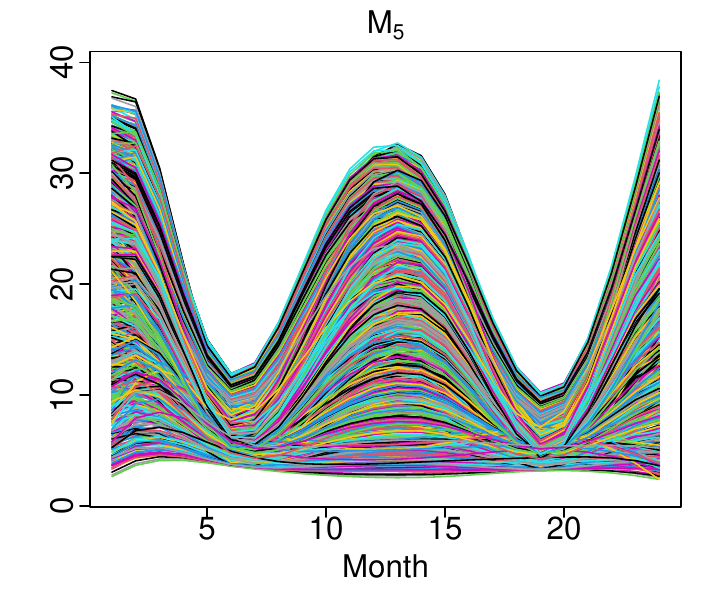}
\quad
\includegraphics[width=5.05cm]{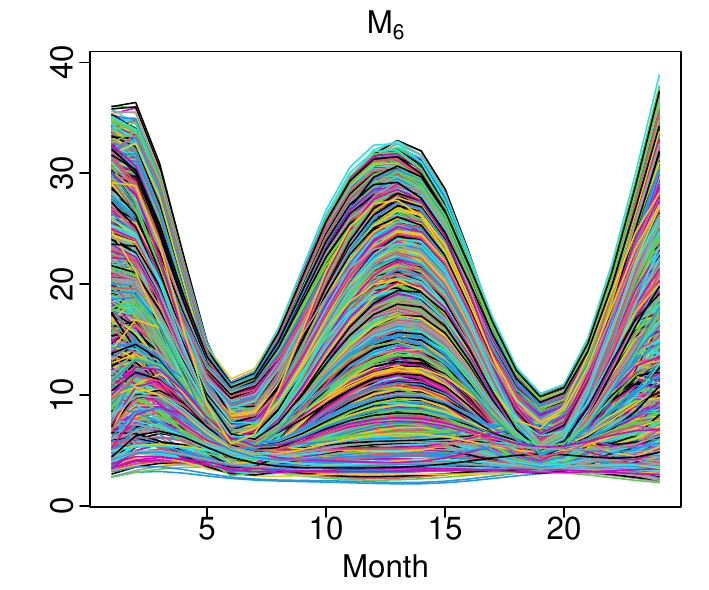}
\caption{Fitted PM$_{2.5}$ trajectories in the air-quality application under M$_1$, M$_2$, M$_3$, M$_4$, M$_5$, and M$_6$. The fitted curves for the quantile-regression methods are obtained at $\tau=0.5$.}
\label{fig:S_app_methods}
\end{figure}

\subsection{Fitted curves across quantile levels}

Figure~\ref{fig:S_app_quantiles} presents the complete collection of fitted conditional quantile curves obtained by the proposed method. The greater separation between the lower and upper quantile curves in winter provides graphical evidence of temporally varying conditional dispersion.

\begin{figure}[!htbp]
\centering
\includegraphics[width=5.05cm]{Fig_4a.pdf}
\quad
\includegraphics[width=5.05cm]{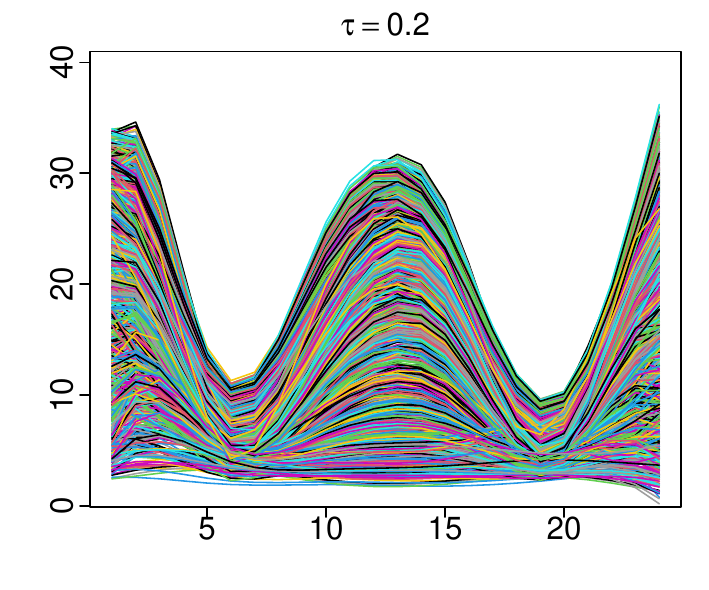}
\quad
\includegraphics[width=5.05cm]{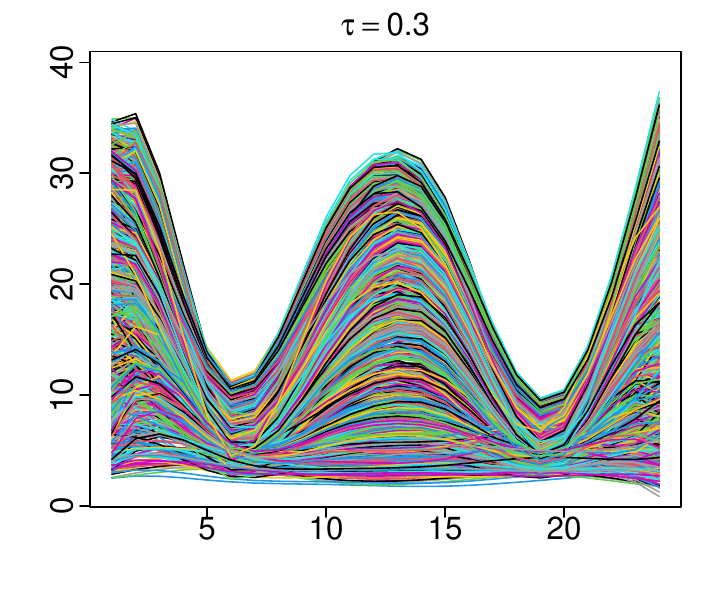}
\\
\includegraphics[width=5.05cm]{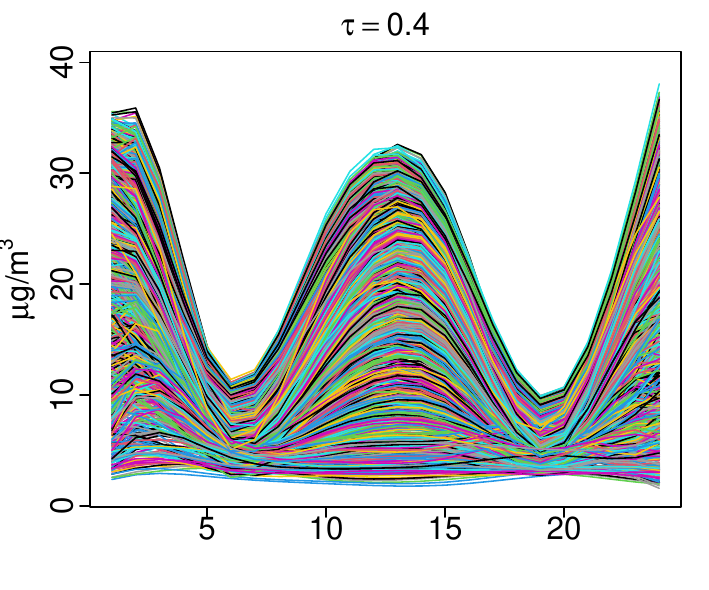}
\quad
\includegraphics[width=5.05cm]{Fig_4e.pdf}
\quad
\includegraphics[width=5.05cm]{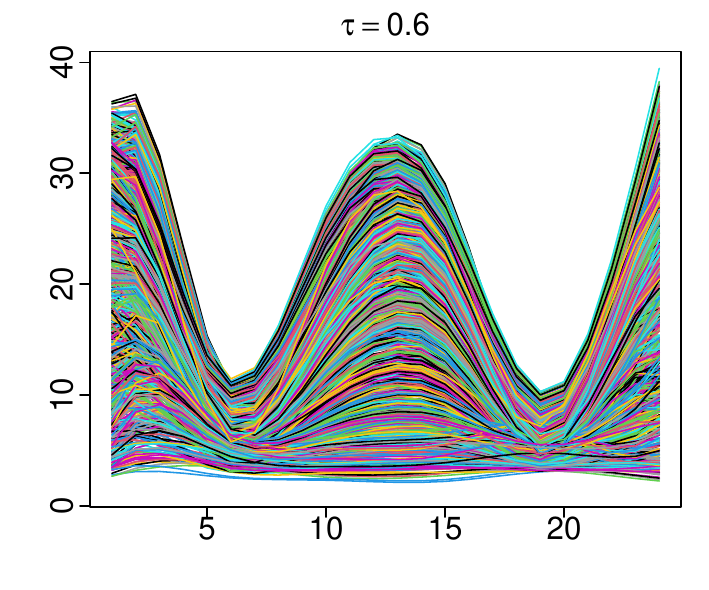}
\\
\includegraphics[width=5.05cm]{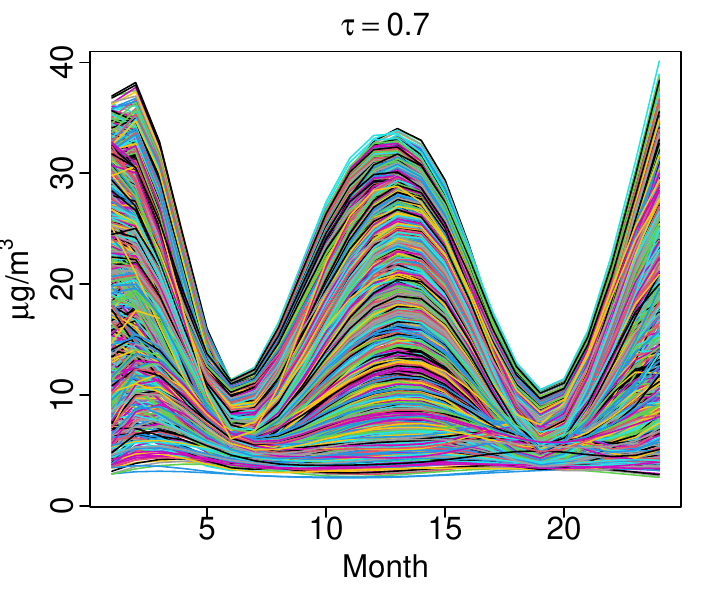}
\quad
\includegraphics[width=5.05cm]{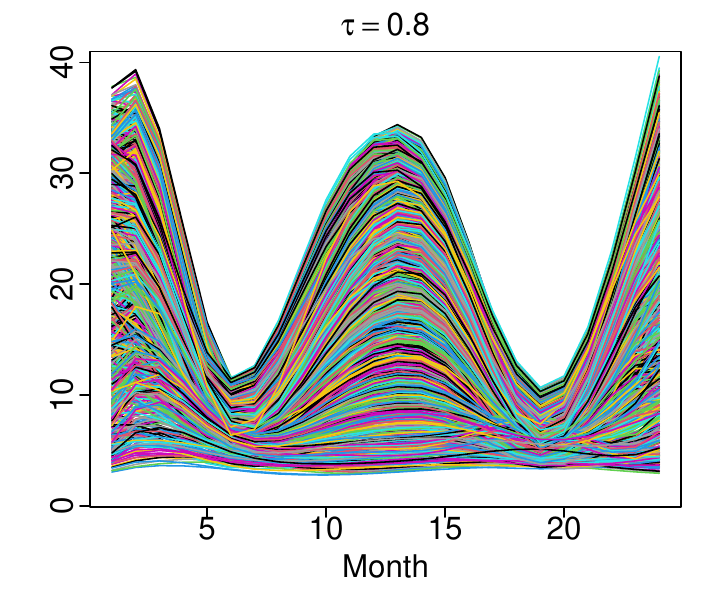}
\quad
\includegraphics[width=5.05cm]{Fig_4i.pdf}
\caption{Fitted conditional PM$_{2.5}$ quantile curves obtained by the
proposed method at
$\tau\in\{0.1,0.2,\ldots,0.9\}$.}
\label{fig:S_app_quantiles}
\end{figure}

\clearpage
\bibliographystyle{agsm}
\bibliography{bibliography.bib}

\end{document}